\documentclass[lettersize,journal]{IEEEtran}
\usepackage{amsmath,amsfonts}
\usepackage{amsmath, amssymb}
\usepackage{amsthm}

\usepackage[ruled,vlined]{algorithm2e}
\usepackage{array}
\usepackage[caption=false,font=normalsize,labelfont=sf,textfont=sf]{subfig}
\usepackage{textcomp}
\usepackage{stfloats}
\usepackage{url}
\usepackage{verbatim}
\usepackage{graphicx}
\usepackage{booktabs}
\usepackage{makecell}
\usepackage{multirow}
\usepackage{tabularx}
\usepackage{threeparttable}
\usepackage{cite}
\usepackage[colorlinks, linkcolor=red]{hyperref}
\DeclareMathSizes{10}{9}{7}{5}
\newtheorem{theorem}{Theorem}

\newtheorem{proposition}[theorem]{Proposition}
\newtheorem{corollary}{Corollary}
\SetKwBlock{Initialization}{Initialization}{end}

\begin{document}

\title{Coverage-Preserving Self-Healing in Partially Observable 3D UAV Swarm Networks via an Advection-Diffusion Graph Transformer}


\author{\IEEEauthorblockN{Dongbin~Jiao\IEEEauthorrefmark{1}, \IEEEmembership{Member,~IEEE,}
Zhihao~Li\IEEEauthorrefmark{1}, Zongsheng~Geng\IEEEauthorrefmark{1}, Dongdong~Zhao\IEEEauthorrefmark{1}, \IEEEmembership{Senior Member, IEEE}, Shengcai~ Liu\IEEEauthorrefmark{3}, \IEEEmembership{Member, IEEE}, and Shi Yan\IEEEauthorrefmark{1}, \IEEEmembership{Senior Member, IEEE}}
\thanks{\IEEEauthorblockA{\IEEEauthorrefmark{1}School of Information Science and Engineering, Lanzhou University, Lanzhou, 730000, P. R. China (e-mail: \{jiaodb, lzhihao2025, gengzsh2023,  ddzhao, yanshi\}@lzu.edu.cn).}}
\thanks{\IEEEauthorblockA{\IEEEauthorrefmark{3}Guangdong Provincial Key Laboratory of Brain-inspired Intelligent Computation, Department of Computer Science and Engineering, Southern University of Science and Technology, Shenzhen, 518055, P. R. China (e-mail: liusc3@sustech.edu.cn).}}
}


\markboth{}
{Shell \MakeLowercase{\textit{et al.}}: A Sample Article Using IEEEtran.cls for IEEE Journals}


\maketitle

\begin{abstract}
This paper investigates the resilience and self-healing of unmanned aerial vehicle swarm communication networks  under severe network partitioning in complex and adversarial environments. To address limited global observability after damage, we develop a damage-state inference mechanism that enables the ground control station (GCS) to infer global failures from sparse local reports without transmitting location-sensitive coordinates. To recover multi-cluster partitioned topologies, we then propose AdDiff-GT, a physics-inspired graph Transformer that incorporates advection-diffusion dynamics into a node-supernode hybrid graph. Local message propagation captures intra-cluster structural relationships, while global attention models inter-cluster interactions. A weighted composite loss jointly considers connectivity restoration, maneuvering cost, motion consistency, inter-cluster coordination, and spatial preservation during model training. We further develop a meta-learning strategy for transferable initialization across damage patterns and swarm scales, enabling efficient online GCS adaptation and distributed airborne inference. Simulations on 8,000 partitioned scenarios, including 5,000 large-scale cases, show that AdDiff-GT achieves effective topology recovery and spatial preservation over six representative baselines. Meta-learning reduces online planning overhead with a minor increase in recovery steps, balancing prior knowledge transfer and online adaptation. 
\end{abstract}


\section{Introduction}
\IEEEPARstart{U}{nmanned} aerial vehicle (UAV) swarm communication network (USNET) have emerged as a promising networking paradigm for disaster response, emergency communications, cooperative sensing, and complex or adversarial operations~\cite{cao2024computational,zha2026large,al2026uav}. However, their dynamic and distributed nature makes them vulnerable to node failures, link disruptions, electromagnetic interference, and malicious attacks. Severe damage can fragment an connected swarm into multiple disconnected clusters, disrupting information exchange and degrading network connectivity. Thus, autonomous topology recovery is essential for resilient USNET operation\cite{mou2022resilient,kong2024resilience}.

Existing resilient UAV network studies have progressed from connectivity restoration to multidimensional resilience involving communication capacity, resource constraints, coverage, quality of service (QoS), and adversarial conditions~\cite{li2023soft,liu2024multistate,kong2024resilience,zhang2024dynamic}. Representative topology-recovery strategies include local interaction and bio-inspired swarm behaviors~\cite{chen2020sidr}, recovery controlled by software-defined networking ~\cite{liu2023sdn}, Delaunay triangulation for relay placement~\cite{huo2025dynamic}, virtual coordinate-based topology maintenance~\cite{gaydamaka2024dynamic}, and robust relay repositioning~\cite{liu2025robusttopology}. Recent studies have further incorporated coverage preservation and QoS guarantees into connectivity restoration~\cite{gao2022coverage,ge2024joint,zheng2025uavswarm,peng2026fastcoverage,wang2026a2g}. However, reliance on centralized optimization, explicit topology information, or repeated planning limits scalability under severe partitioning. Moreover, connectivity-driven repositioning may cause UAV aggregation and degrade spatial distribution, requiring self-healing to jointly consider connectivity, spatial preservation, and online efficiency.

Severe network damage limits global observability from fragmented local observations. Related studies have highlighted partial observability~\cite{ye2022multi} and multi-level damage awareness~\cite{lin2025multi}. However, conventional monitoring depends on continuous collection of node positions and link states~\cite{chen2026ris}, incurring substantial communication overhead and becoming unreliable when complete network information is unavailable~\cite{wei2024survey}. Existing multi-agent reinforcement learning and distributed control approaches, such as MADDPG-APF~\cite{li2025fast}, zero-shot physics-guided learning~\cite{lin2026zero}, and spatio-temporal aggregation~\cite{guo2024fanets,zhuang2026network}, mainly address scalable decision-making and distributed execution rather than the preceding global damage-state inference problem. Therefore, reconstructing global failure states from sparse local observations without continuous transmission of location-sensitive information remains insufficiently addressed.

Graph learning has emerged as an effective paradigm for UAV network recovery and resilient swarm control~\cite{mou2022resilient,ye2022multi,lin2025multi,peng2025prioritized,wang2025uav,peng2026fastcoverage}. Graph neural networks (GNNs) and graph attention networks (GATs) have been widely adopted to learn topology-aware representations for graph-based recovery~\cite{kipf2017semisupervised,wang2025uav}, including graph reinforcement learning under partial observability
~\cite{ye2022multi,peng2025prioritized} and multi-branch damage attention for damage-state modeling~\cite{lin2025multi}. Graph learning has also been extended to recovery planning through online GCN-based recovery with meta-learning~\cite{mou2022resilient}, 
GAT-DRL-based joint neighbor selection and trajectory planning~\cite{wang2025uav}, and energy-constrained graph convolutional networks (GCNs) recovery
~\cite{wang2025energy}. Other recent efforts have explored zero-shot physics-guided learning~\cite{lin2026zero}, distributed hybrid MADDPG-APF frameworks~\cite{li2025fast}, spatio-temporal attention aggregation ~\cite{guo2024fanets}, and heterogeneous GNNs for IoT-assisted systems~\cite{zhuang2026network} to improve scalability and support distributed execution. 

Despite these advances, graph-based recovery methods remain challenged by severely partitioned, large-scale USNET. Local message propagation has limited capability to capture long-range dependencies among disconnected clusters, while separate treatment of intra-cluster preservation and inter-cluster coordination may cause inefficient reconnection and excessive spatial aggregation. Moreover, models trained on specific scales or damage distributions often require costly online optimization for unseen scenarios. These gaps motivate a unified framework for global damage-state inference, UAV- and cluster-level topology modeling, and efficient cross-scenario adaptation.

To address these challenges, we develop an integrated self-healing framework for partially observable USNET based on an advection-diffusion graph Transformer (AdDiff-GT). A damage-state perception mechanism infers the global post-damage state from sparse local reports without sharing location-sensitive coordinates. A node-supernode hybrid graph captures intra-cluster structures and inter-cluster relationships. AdDiff-GT then combines local message propagation with global attention, where advection-diffusion dynamics~\cite{medvedev2014nonlinear} provide a physics-inspired inductive bias for topology recovery. A meta-learning framework transfers recovery knowledge across damage patterns and swarm scales, reducing online planning overhead and enabling distributed airborne inference. A weighted composite loss further balances connectivity restoration, maneuvering cost, motion consistency, inter-cluster coordination, and spatial preservation.


The main contributions of this paper are summarized as follows: (i) We formulate resilient self-healing for severely partitioned USNET under incomplete damage-state information, jointly considering connectivity, spatial preservation, and online efficiency, and develop sparse-observation-based global damage inference. (ii) We propose AdDiff-GT, a physics-inspired graph Transformer that integrates advection-diffusion dynamics with a node-supernode hybrid graph for local-global topology recovery, trained with a weighted composite loss. (iii) We develop scale-aware meta-learning for efficient online adaptation and distributed airborne inference. (iv) Experiments on 8,000 partitioned scenarios, including 5,000 large-scale cases, demonstrate effective connectivity recovery and spatial dispersion, while meta-initialization reduces online planning latency with a moderate increase in recovery steps.

The remainder of this paper is organized as follows. Section~\ref{sec:system} presents the system model and problem formulation. Section~\ref{sec:damage_airground} describes damage assessment and air-ground asymmetric recovery. Section~\ref{sec:AdDiff-GT} presents the proposed AdDiff-GT framework and fast adaptation.  Section~\ref{sec:simulation} provides the experimental evaluation. Section~\ref{sec:conclusion} concludes this paper.

Definitions of the main symbols and operators used in this paper are provided in Appendix~\ref{app:notation}. 

\section{System Model and Problem Formulation}\label{sec:system}
\subsection{UAV Swarm and Communication Architecture}
Consider a USNET consisting of $N$ UAVs. The three-dimensional (3D) position of UAV $i$ at time $t$ is denoted by $\mathbf{p}_{i,t} =[x_{i,t},y_{i,t},z_{i,t}]^{\mathrm{T}}\in \mathbb{R}^3$, and its motion state is defined as $\mathbf{x}_{i,t} =\left[\mathbf{p}_{i,t}^{\mathrm{T}},s_{i,t}\right]^{\mathrm{T}}$,
where $s_{i,t}\in\{0,v_{\mathrm{mov}}\}$ represents whether UAV $i$ is stationary or executing recovery maneuvers at a constant speed $v_{\mathrm{mov}}$.

\subsubsection{Hierarchical Communication Architecture}
The USNET adopts a hierarchical communication architecture comprising short-range information interaction links (SIILs) and long-range data transmission links (LDTLs). SIILs support local information exchange and intra-swarm topology construction, whereas LDTLs provide a low-rate air-ground control channel between the GCS and surviving UAVs for post-damage coordination. This channel carries the exchange of disconnected-neighbor IDs, inferred damage information, model parameters, and recovery commands.

Under adversarial conditions, post-damage air-ground communication is restricted to ID-level state information and model-level parameter updates to limit the exposure of location-sensitive information. In particular, surviving UAVs do not transmit real-time coordinates or recovery trajectories to the GCS. Instead, spatial information is maintained in a GCS-side virtual topology for global damage-state inference and online model adaptation, reducing location-sensitive information exposure while preserving closed-loop self-healing.

\subsubsection{Communication Links}
The inter-UAV and UAV-to-GCS communication channels are modeled as air-to-air (A2A) and air-to-ground (A2G) links, corresponding to SIILs and LDTLs, respectively.

\textbf{A2A Links}: A line-of-sight (LoS) propagation model is adopted for inter-UAV communication~\cite{khuwaja2018survey}. The received power from UAV $j$ to UAV $i$ is given by
\begin{equation}
P^{\mathrm{A2A}}_{ij,t}=P_{\mathrm{tx}}G_iG_j\left(\frac{\lambda}{4\pi d_{ij,t}}\right)^2|h_{ij,t}|^2,
\end{equation}
where $P_{\mathrm{tx}}$ is the transmission power, $G_i$ and $G_j$ are the antenna gains, $\lambda$ is the carrier wavelength, $d_{ij,t}=\|\mathbf{p}_{i,t}-\mathbf{p}_{j,t}\|_2$ is the Euclidean distance between UAVs $i$ and $j$, and $|h_{ij,t}|^2$ denotes the small-scale fading gain. An A2A link is established when the received power exceeds the sensitivity threshold $\tau$, i.e., 
$P^{\mathrm{A2A}}_{ij,t}\ge \tau$, which is equivalently characterized by
\begin{equation}
\|\mathbf{p}_{i,t}-\mathbf{p}_{j,t}\|_2\le R_c,
\label{eq:a2a_distance}
\end{equation}
where $R_c$ denotes the effective communication radius. Thus, Eq.~\eqref{eq:a2a_distance} defines the physical edge criterion for constructing the intra-swarm topology and evaluating connectivity recovery.

\textbf{A2G Links}: A large-scale fading model is adopted for GCS-UAV communication~\cite{khawaja2019survey}. The received power at the GCS from UAV $i$ is modeled as
\begin{equation}
P_{i,t}^{\mathrm{A2G}}=P_{\mathrm{tx}}\beta_0 d_{i,t}^{-\alpha}\chi,
\end{equation}
where $\beta_0$ denotes the channel gain at the reference distance, $d_{i,t}$  is the distance between UAV $i$ and the GCS, $\alpha$ is the path loss exponent, and $\chi$ represents the shadow fading factor. The A2G link is available when $\frac{P_{i,t}^{\mathrm{A2G}}}{\sigma^2}\ge \gamma$, where $\sigma^2$ is the noise power and $\gamma$ is the signal-to-noise ratio threshold. A2G links are used solely for GCS-assisted control and do not participate in intra-swarm topology construction or connectivity recovery.

\subsection{Residual UAV~(\text{RUAV}) Graph Construction}
Let the intact USNET at time $t$ be represented by $\mathcal{G}_t=(\mathcal{V}_t,\mathcal{E}_t,\mathbf{X}_t)$. Following an unpredictable extreme disaster (UED), let $\mathcal{R}_t\subseteq \mathcal{V}_t$ denote the set of surviving
UAVs. We assume that all surviving UAVs can report their status to the GCS through backhaul links.
The resulting \text{RUAV} graph is defined as $\mathcal{G}'_t=(\mathcal{R}_t,\mathcal{E}'_t,\mathbf{X}'_t)$,
where $\mathcal{E}'_t=\{(i,j)\mid i,j\in\mathcal{R}_t,\ i\ne j,\ \|\mathbf{p}_{i,t}-\mathbf{p}_{j,t}\|_2\le R_c\}$,
and $\mathbf{X}'_t=[\mathbf{p}_{r_1,t},\ldots,\mathbf{p}_{r_{|\mathcal{R}_t|},t}]^{\mathrm{T}}$ denotes the position matrix of the surviving UAVs. Thus, $\mathcal{G}'_t$ characterizes the post-damage communication topology. If $\mathcal{G}'_t$ forms a single connected cluster, the surviving UAVs retain global communication capability; otherwise, the RUAV network is partitioned into multiple disconnected clusters, which must be reconnected through cooperative UAV repositioning.

Let $\mathbf{A}_t=(a_{ij,t})$ denote the adjacency matrix of $\mathcal{G}'_t$, where $a_{ij,t}=1$ if $(i,j)\in \mathcal{E}'_t$ and $a_{ij,t}=0$
otherwise. The corresponding degree matrix is $\mathbf{D}_t=\operatorname{diag}\left(\sum_{j}a_{1j,t},\ldots,\sum_{j}a_{|\mathcal{R}_t|j,t}\right)$, and the graph Laplacian matrix is $\mathbf{L}_t=\mathbf{D}_t-\mathbf{A}_t$.


\subsection{Problem Formulation}
When the residual network is partitioned into $C_t>1$ connected clusters, the aim of self-healing is to determine the displacement of surviving UAVs such that global connectivity is restored while limiting mobility cost and preserving the spatial distribution. Let the repaired position of UAV $i$ be 
\begin{equation}
\widehat{\mathbf{p}}_i=\mathbf{p}_{i,t}+\Delta\mathbf{p}_i,\quad i\in\mathcal{R}_t,
\end{equation}
where $\Delta\mathbf{p}_i$ denotes the displacement. Assuming that all surviving UAVs move at a constant speed $v_{\mathrm{mov}}$, the recovery time is determined by the maximum displacement:
\begin{equation}
T_{\mathrm{rec}}=\frac{1}{v_{\mathrm{mov}}}\max_{i\in\mathcal{R}_t}\|\Delta\mathbf{p}_i\|_2.
\end{equation}
The maximum and average displacement are defined as
\begin{equation}
d_{\max}^{\mathrm{mov}}=\max_{i\in\mathcal{R}_t}\|\Delta\mathbf{p}_i\|_2, \quad d_{\mathrm{avg}}^{\mathrm{mov}}=\frac{1}{|\mathcal{R}_t|}\sum_{i\in\mathcal{R}_t}\|\Delta\mathbf{p}_i\|_2.
\end{equation}
The former characterizes the recovery latency and peak mobility burden, whereas the latter measures the overall mobility cost of the swarm.

To quantify the spatial distribution among recovered UAVs, we define the minimum pairwise distance as
\begin{equation}
d_{\min}^{\mathrm{sep}}=\min_{i\ne j}\|\widehat{\mathbf{p}}_i-\widehat{\mathbf{p}}_j\|_2.
\end{equation}
A larger $d_{\min}^{\mathrm{sep}}$ promotes spatial dispersion while maintaining the required network connectivity. 

Accordingly, the self-healing problem is formulated as
\begin{subequations}
\label{eq:problem}
\begin{align}
\min_{\{\Delta\mathbf{p}_i\}}\quad
& J=\alpha_{\max}d_{\max}^{\mathrm{mov}}
\!+\alpha_{\mathrm{avg}}d_{\mathrm{avg}}^{\mathrm{mov}}
-\lambda_{\mathrm{sep}} d_{\min}^{\mathrm{sep}} \\
\mathrm{s.t.}\quad
& C(\widehat{\mathbf{P}})=1, \label{eq:problem_conn}\\
& \|\widehat{\mathbf{p}}_i-\widehat{\mathbf{p}}_j\|_2\ge d_{\mathrm{safe}},\quad \forall (i,j)\in \mathcal{E}'_t(\widehat{\mathbf{P}}), \label{eq:safe-dis}\\
& \widehat{\mathbf{p}}_i\in\Omega,\quad i\in\mathcal{R}_t \label{eq:fly-dis},
\end{align}
\end{subequations}
where $\alpha_{\max}$, $\alpha_{\mathrm{avg}}$, and $\lambda_{\mathrm{sep}}$ are nonnegative weighting coefficients that balance recovery latency, mobility cost, and spatial dispersion, respectively. Here,  $C(\widehat{\mathbf{P}})$ denotes the number of connected clusters induced by the recovered UAV positions $\widehat{\mathbf{P}}$, $\Omega$ is the feasible 3D flight region, and $d_{\mathrm{safe}}$ is the anti-collision safety threshold. Constraint \eqref{eq:problem_conn} ensures global  connectivity, while 
constraint \eqref{eq:safe-dis} and 
constraint \eqref{eq:fly-dis} ensure collision avoidance and flight-space feasibility, respectively. The recovered communication topology is determined by the effective A2A distance criterion in Eq.~\eqref{eq:a2a_distance}.

The optimization problem in \eqref{eq:problem} is nonlinear and nonconvex due to discrete connectivity, continuous UAV positions, and nonsmooth max/min operators. Therefore, conventional online optimization is costly, while purely local strategies lack cross-cluster awareness. We reformulate self-healing as a graph-based displacement prediction task, where AdDiff-GT predicts $\Delta\mathbf{P}$ from the hybrid graph and jointly optimizes connectivity, mobility efficiency, and spatial preservation.

\section{Damage Assessment and Air-Ground Asymmetric Recovery}
\label{sec:damage_airground}

\subsection{Damage-State Inference under Partial Observations}
Following a UED event, the 
UAV set $\mathcal{V}_t$ is partitioned into the surviving set $\mathcal{R}_t$ and the failed set $\mathcal{V}_{\mathrm{fail}}$, with $\mathcal{V}_t=\mathcal{R}_t\cup\mathcal{V}_{\mathrm{fail}}$ and $\mathcal{R}_t\cap\mathcal{V}_{\mathrm{fail}}=\emptyset$. Under partial observability, the GCS cannot directly access $\mathcal{V}_{\mathrm{fail}}$ and must infer the global damage-state from local reports generated by surviving UAVs.

To reduce communication overhead and avoid the exposure of location-sensitive information, an event-triggered local index reporting mechanism is adopted. Specifically, 
a surviving UAV $u\in \mathcal{R}_t$ triggers a report when it detects that a previously connected neighbor has become unreachable. The report contains only the indices of disconnected neighbors, without transmitting real-time coordinates or recovery trajectories. Let  $\mathcal{R}_{\mathrm{rep}}=\{u\in\mathcal{R}_t\mid \exists v\in\mathcal{V}_{\mathrm{fail}},\ (u,v)\in\mathcal{E}_t\}$ denote the reporting-UAV set, and let $\mathcal{D}_{u}^{\mathrm{rep}}$ be the corresponding locally observed failed-neighbor set. The aggregated initial damage boundary is then given by $\mathcal{D}_{\mathrm{rep}}=\bigcup_{u\in\mathcal{R}_{\mathrm{rep}}}\mathcal{D}_{u}^{\mathrm{rep}}$.
\begin{algorithm}[htp]\small
\SetKwInput{KwIn}{Inputs}
\SetKwInput{KwOut}{Output}
\KwIn{Historical topology $\mathcal{G}_t=(\mathcal{V}_t,\mathcal{E}_t)$, reporting node set $\mathcal{R}_{\mathrm{rep}}$, local ID reports $\{\mathcal{D}_{u}^{\mathrm{rep}}\}$}
\KwOut{Inferred damaged node set $\mathcal{D}_{\mathrm{inf}}$}
\Initialization{
    $\mathcal{D}_{\mathrm{inf}}\leftarrow\bigcup_{u\in\mathcal{R}_{\mathrm{rep}}}\mathcal{D}_{u}^{\mathrm{rep}}$\;
    $\mathcal{F}\leftarrow\mathcal{D}_{\mathrm{inf}}$\;
}
\While{$\mathcal{F}$ is not empty}{
    Pop one damaged node $d$ from $\mathcal{F}$\;
    \For{\textbf{each} historical neighbor $n\in\mathcal{N}_{\mathcal{G}_t}(d)$}{
        \If{$n\notin\mathcal{D}_{\mathrm{inf}}$ and $n\notin\mathcal{R}_{\mathrm{rep}}$}{
            $\mathcal{D}_{\mathrm{inf}}\leftarrow\mathcal{D}_{\mathrm{inf}}\cup\{n\}$\;
            Push $n$ into $\mathcal{F}$\;
        }
    }
}
\KwRet $\mathcal{D}_{\mathrm{inf}}$\;
\caption{Damage Inference from Local Reports}
\label{alg:damage_inference}
\end{algorithm}

Based on $\mathcal{D}_{\mathrm{rep}}$, the GCS performs topological closure expansion over the historical graph $\mathcal{G}_t$, as detailed in Algorithm \ref{alg:damage_inference}.
This process follows the boundary diffusion principle under the \emph{silence-as-failure} rule: a node adjacent to a confirmed failed node is classified as failed if it does not belong to the reporting set. The resulting inferred damage set is $\mathcal{D}_{\mathrm{inf}}
=\operatorname{BFS_{Closure}}_{\mathcal{G}_t}\left(\mathcal{D}_{\mathrm{rep}},\mathcal{R}_{\mathrm{rep}}\right)$. This inference mechanism exploits the boundary-shielding property of event-triggered reporting: every surviving node adjacent to a failed UAV is included in $\mathcal{R}_{\mathrm{rep}}$. Hence, a non-reporting node adjacent to a confirmed failed node can be recursively classified as failed, allowing the GCS to reconstruct the damage region from its observable boundary toward the interior of the damaged area.

\subsection{Connectivity Detection of RUAV Graph}
Repeated eigendecomposition of $\mathbf{L}_t$ 
incurs a computational complexity of
$\mathcal{O}(|\mathcal{R}_t|^3)$, which is unsuitable for large-scale online recovery. Therefore, Algorithm~\ref{alg:bfs_connectivity} employs the breadth-first search (BFS) algorithm to identify connected clusters by traversing adjacent UAVs satisfying $[\mathbf{L}_t]_{uv}<0$. Since each UAV and communication edge is visited at most a constant number of times, the resulting complexity is $\mathcal{O}(|\mathcal{R}_t|+|\mathcal{E}'_t|)$.
\begin{algorithm}[htp]\small
\SetKwInput{KwIn}{Inputs}
\SetKwInput{KwOut}{Outputs}
\KwIn{Surviving node set $\mathcal{R}_t$, Laplacian matrix $\mathbf{L}_t$}
\KwOut{Number of connected components $C_t$, connectivity indicator $b_t$}
\Initialization{$\mathbf{v}\leftarrow\mathbf{0}_{|\mathcal{R}_t|}$, $C_t\leftarrow0$\;}
\For{\textbf{each} node index $i=1,\ldots,|\mathcal{R}_t|$}{
    \If{$\mathbf{v}[i]=0$}{
        $C_t\leftarrow C_t+1$, $\mathcal{F}\leftarrow[i]$, $\mathbf{v}[i]\leftarrow1$\;
        \While{$\mathcal{F}$ is not empty}{
            Pop the head node $u$ from $\mathcal{F}$\;
            $\mathcal{N}_u\leftarrow\{j\mid [\mathbf{L}_t]_{u,j}<0\}$\;
            \For{\textbf{each} $j\in\mathcal{N}_u$}{
                \If{$\mathbf{v}[j]=0$}{
                    Push $j$ into $\mathcal{F}$; $\mathbf{v}[j]\leftarrow1$\;
                }
            }
        }
    }
}
$b_t\leftarrow\mathbb{I}\{C_t=1\}$\;
\KwRet $(C_t,b_t)$\;
\caption{BFS-Based Connectivity Detection}
\label{alg:bfs_connectivity}
\end{algorithm}

\vspace{-0.20in}
\subsection{Air-Ground Asymmetric Online Recovery}
Upon completion of damage inference, the GCS obtains the surviving UAV set $\mathcal{R}_t=\mathcal{V}_t\setminus\mathcal{D}_{\mathrm{inf}}$ and the corresponding RUAV graph $\mathcal{G}'_t$. If the BFS-based connectivity detection in Algorithm \ref{alg:bfs_connectivity} returns $C_t=1$,  the surviving network is already connected and no recovery maneuver is required. Otherwise, for $C_t > 1$, the online self-healing procedure in Algorithm \ref{alg:airground} is activated.
\begin{algorithm}[ht]\small
\SetKwInput{KwIn}{Inputs}
\SetKwInput{KwOut}{Outputs}
\KwIn{Historical topology $\mathcal{G}_t$, local ID reports, meta-parameter library $\{\mathbf{\Theta}_n^\star\}_{n=2}^{N}$}
\KwOut{Displacement decision $\Delta\mathbf{p}_i$ for each surviving UAV}
$\mathcal{D}_{\mathrm{inf}}\leftarrow$ \text{Algorithm~\ref{alg:damage_inference}} (Damage Inference)\;
$\mathcal{R}_t\leftarrow\mathcal{V}_t\setminus\mathcal{D}_{\mathrm{inf}}$, and construct $\mathcal{G}'_t$\;
$(C_t,b_t)\leftarrow$ Algorithm~\ref{alg:bfs_connectivity} $(\mathcal{R}_t,\mathbf{L}_t)$\;
\If{$b_t=1$}{
    GCS broadcasts a hold command and all RUAV keep positions\;
}
\Else{
    GCS loads meta prior $\mathbf{\Theta}_{|\mathcal{R}_t|}^\star$, fine-tunes AdDiff-GT in the virtual damaged topology, and broadcasts $\mathcal{D}_{\mathrm{inf}}$ and tuned parameters $\mathbf{\Theta}^{+}$\;
    \For{\textbf{each} surviving UAV $i\in\mathcal{R}_t$}{
        Update local graph view by removing $\mathcal{D}_{\mathrm{inf}}$\;
        Run onboard forward inference with $\mathbf{\Theta}^{+}$\;
        Execute the decoded displacement $\Delta\mathbf{p}_i$\;
    }
}
\caption{Air-Ground Asymmetric AdDiff-GT Recovery} \label{alg:airground}
\end{algorithm}

The recovery procedure adopts an asymmetric control architecture in which the GCS performs global damage inference and online model adaptation, while surviving UAVs independently execute local displacement inference using the adapted parameters. This design leverages GCS-side computation while avoiding high-frequency centralized trajectory control. Air-ground communication is limited to ID-level reports and model-level updates, reducing the exposure of position-sensitive information.



\section{AdDiff-GT for Coverage-Preservin Self-Healing}\label{sec:AdDiff-GT}
The proposed AdDiff-GT self-healing process is illustrated
in Appendix~\ref{app:process}.
\subsection{Hierarchical Repair Graph Representation}
\label{sec:model}
\subsubsection{From Physical Manifold to UAV Repair Graph}
AdDiff-GT is a physics-inspired graph Transformer designed for coverage-preserving self-healing in partially observable USNET. Inspired by physical models~\cite{medvedev2014nonlinear,wu2025supercharging}, we incorporate advection-diffusion dynamics into a hybrid node-supernode graph to introduce physics-guided inductive bias for  displacement-field learning and topology recovery.

We model the constrained 3D task airspace as a compact manifold $\Omega\subset\mathbb{R}^3$ with boundaries.  The tangent space $T_m\Omega$ at any position $m\in\Omega$ characterizes the locally feasible differential maneuvering directions. Accordingly, the position $\mathbf{p}_{i,t}$ of UAV $i$ 
 is treated as a discrete sampling point on $\Omega$, and its recovery displacement $\Delta\mathbf{p}_i \in T_{\mathbf{p}_{i,t}}\Omega$ represents an executable maneuvering vector.

Let $\varphi(m,t)$ denote a continuous repair-potential field over $\Omega$ that captures both local connectivity and cross-cluster bridging potential. Its spatial propagation can be modeled by the diffusion process
\begin{equation}
\frac{\partial \varphi(m,t)}{\partial t}
=\nabla ^{\ast}\left(S(m,t)\odot\nabla \varphi(m,t)\right),
\quad t \ge 0, m\in \Omega,
\label{eq:manifold_diffusion}
\end{equation}
where $\nabla$ and $\nabla ^{\ast}$ denote the gradient and divergence operators on the manifold, respectively, and $S(m,t)$ is the diffusivity that controls the spatial propagation of repair information. 

Diffusion alone cannot capture the directional cooperative motion required for topology recovery. Therefore, we introduce a tangential velocity field $V(m,t)\in T_m\Omega$, whose direction is constrained by the observed communication topology. The resulting advection-diffusion dynamics are given by 
\begin{equation}
\frac{\partial \varphi(m,t)}{\partial t}
= \nabla^{\ast}\left(S(m,t)\odot\nabla \varphi(m,t)\right) 
 + \beta\nabla^{\ast}\left(V(m,t)\cdot\varphi(m,t)\right),
\label{eq:manifold_advdiff}
\end{equation}
where $t\ge 0$, $\beta \ge 0$ controls the intensity of the advection process. The diffusion term propagates repair potential across the task airspace, whereas the advection term introduces topology-constrained directional evolution to guide self-healing.

Node failures caused by UED remove sampling points from the manifold $\Omega$, resulting in abrupt changes in local adjacency structures. Since the continuous state field cannot be fully observed after damage, it is discretized over a hybrid node set $\mathbf{Z}(t)=[\mathbf{z}_u(t)]_{u\in\mathcal{V}_t^h}$, where $\mathcal{V}_t^h$ comprises surviving UAV nodes, representing manifold samples, and supernodes, representing coarse-grained abstractions of disconnected clusters. The resulting RUAV hybrid graph serves as a discrete approximation of the damaged communication manifold under partial observations. AdDiff-GT subsequently learns the recovery state over this hybrid graph and decodes it into 3D displacement increments for the surviving UAVs.


\subsubsection{Hybrid Node-Supernode Representation}
To jointly capture node-level connectivity and cluster-level fragmentation structure,
we extend the RUAV graph
$\mathcal{G}'_t$ into a hybrid node-supernode graph. 
Specifically, Algorithm~\ref{alg:bfs_connectivity} partitions $\mathcal{G}'_t$ into $C_t$  connected clusters $\{\mathcal{C}_1,\mathcal{C}_2,\ldots,\mathcal{C}_{C_t}\}$, satisfying $
\bigcup_{c=1}^{C_t}\mathcal{C}_c=\mathcal{R}_t$, $\mathcal{C}_c \cap \mathcal{C}_{c'} = \varnothing$, $c \neq c'$. Physical A2A links are preserved within each cluster, while each disconnected cluster is represented by an abstract supernode. For cluster $\mathcal{C}_c$, a supernode $s_c$ is introduced with its position defined as the cluster centroid:   
\begin{equation}
\mathbf{s}_{c,t}=\frac{1}{|\mathcal{C}_c|}\sum_{i\in\mathcal{C}_c}\mathbf{p}_{i,t}.
\label{eq:supernode_center}
\end{equation}
Let $q(i)$ denote the index of the connected cluster containing node $i$, i.e., $i\in \mathcal{C}_{q(i)}$. 
Supernodes provide cluster-level representations, enabling AdDiff-GT to jointly model individual UAV movements and global cluster evolution.

The resulting hybrid graph is defined as $\mathcal{G}_t^{h}=(\mathcal{V}_t^{h},\mathcal{E}_t^{h},\mathbf{H}_t)$, where the node set is ${\mathcal{V}_t^{h}=\mathcal{R}_t\cup\mathcal{S}_t}$,  the supernode set is ${\mathcal{S}_t=\{s_c\}_{c=1}^{C_t}}$, and the corresponding node-feature matrix $\mathbf{H}_t = [\mathbf{h}_{i,t}]_{i \in \mathcal{V}_t^{h}}$. The hybrid edge set $\mathcal{E}_t^{h}=\mathcal{E}_{\mathrm{phy}}\cup\mathcal{E}_{\mathrm{vir}}\cup\mathcal{E}_{\mathrm{ns}}\cup\mathcal{E}_{\mathrm{ss}}$ contains four types of edges: $\mathcal{E}_{\mathrm{phy}}$ denotes intra-cluster physical A2A links; $\mathcal{E}_{\mathrm{vir}}$ contains virtual links from each boundary UAV to its nearest UAV in another cluster; $\mathcal{E}_{\mathrm{ns}}$ contains association edges between UAV nodes and their corresponding supernodes; and $\mathcal{E}_{\mathrm{ss}}$ contains fully connected cluster-level edges among supernodes. 

Specifically, physical edges are defined as 
\begin{equation}
(i,j)\in\mathcal{E}_{\mathrm{phy}}\Leftrightarrow
 \|\mathbf{p}_{i,t}-\mathbf{p}_{j,t}\|_2\le R_c, i,j\in \mathcal{C}_{q(i)}, 
\end{equation}
with a physical-edge weight $w^{\mathrm{phy}}_{ij}$.

For a boundary UAV  $i \in \mathcal{B}_{q(i)}$, where $\mathcal{B}_{c}$ denotes the boundary-UAV set of cluster $\mathcal{C}_c$, its nearest UAV in another cluster is selected as
\begin{equation}
j^\star(i)=\arg\min_{j\in \mathcal{R}_t\setminus\mathcal{C}_{q(i)}}\|\mathbf{p}_{i,t}-\mathbf{p}_{j,t}\|_2.
\end{equation}
The corresponding virtual edge is assigned the weight
\begin{equation}
w_{ij}^{\mathrm{vir}}=\frac{1}{\|\mathbf{p}_{i,t}-\mathbf{p}_{j,t}\|_2+\varepsilon}, \quad \varepsilon > 0.
\end{equation}

Each UAV node is additionally connected to its corresponding supernode through an association edge with a fixed weight $w^{\mathrm{ns}}_{is_c}$. Supernodes are fully connected to capture inter-cluster relationships, with
\begin{equation}
w^{\mathrm{ss}}_{s_cs_{c'}}=\frac{1}{\|\mathbf{s}_{c,t}-\mathbf{s}_{c',t}\|_2+\varepsilon},\quad c\ne c' .
\end{equation}

Thus, physical, association, virtual, and supernode edges respectively encode intra-cluster connectivity, UAV-cluster coupling, cross-cluster geometric guidance, and global inter-cluster relationships.

For each UAV node, a 13-dimensional (13D) input feature is constructed as
\begin{equation}
\mathbf{h}_{i,t}=
[\bar{\mathbf{p}}_{i,t},\ \bar{\mathbf{p}}_{i,t}-\bar{\mathbf{c}}_{q(i),t},\
\rho_{q(i),t},\ \delta_{i,t},\ r_{i,t},\ \kappa_{q(i),t},\ b_{i,t},\ \mu_{i,t},\ \nu_{i,t}],
\label{eq:feature}
\end{equation}
where $\bar{\mathbf{p}}_{i,t}\in\mathbb{R}^3$ denotes the normalized UAV position, $\bar{\mathbf{c}}_{q(i),t}\in\mathbb{R}^3$ is the normalized center of the cluster containing UAV $i$, and $\rho_{q(i),t}=|\mathcal{C}_{q(i)}|/|\mathcal{R}_t|$ is the cluster-size ratio. The degree ratio $\delta_{i,t}$ is normalized by the maximum node degree within its cluster, $r_{i,t} =\|\bar{\mathbf{p}}_{i,t}-\bar{\mathbf{g}_t}\|_2$ is the radial distance to the global swarm center $\bar{\mathbf{g}_t}$, $\kappa_{q(i),t}$ is the normalized cluster index. Furthermore, $b_{i,t}\in\{0,1\}$ indicates whether UAV $i$ is a boundary node, $\mu_{i,t} = \min(d_{i,t}^{\mathrm{near}}/R_c,5)$ is the truncated normalized distance to the nearest UAV in another cluster, and $\nu_{i,t}$ denotes the normalized index of the nearest other cluster. These features jointly encode absolute geometry, intra-cluster relative position, cluster scale, local connectivity, global spatial relations, cluster identity, boundary status, and inter-cluster proximity, enabling AdDiff-GT to identify UAVs with high potential for cross-cluster bridging. Supernodes use the same 13D feature space, with their positions determined by Eq.~\eqref{eq:supernode_center}  and cluster-level attributes derived from the corresponding UAV cluster. The relative-position and boundary-specific terms are set according to their cluster-level semantics, while the remaining geometric and structural features are computed consistently at the supernode level. This shared representation enables UAV nodes and supernodes to be processed within a unified AdDiff-GT architecture.

\subsection{Advection-Diffusion Graph Transformer}
As illustrated in Fig. \ref{fig:model_architecture}, AdDiff-GT maps the surviving communication topology to a self-healing displacement field. The model first constructs the RUAV graph $\mathcal{G}'_t$ from the surviving UAVs and then augments it with cluster-level supernodes to form the hybrid graph $\mathcal{G}_t^h$. The resulting 13D node features are  processed by advection-diffusion propagation layers to capture topology-aware feature evolution, followed by a linear decoder that generates 3D UAV displacement vectors.
\begin{figure*}[htbp]
    \centering
   \includegraphics[width=0.9\linewidth]{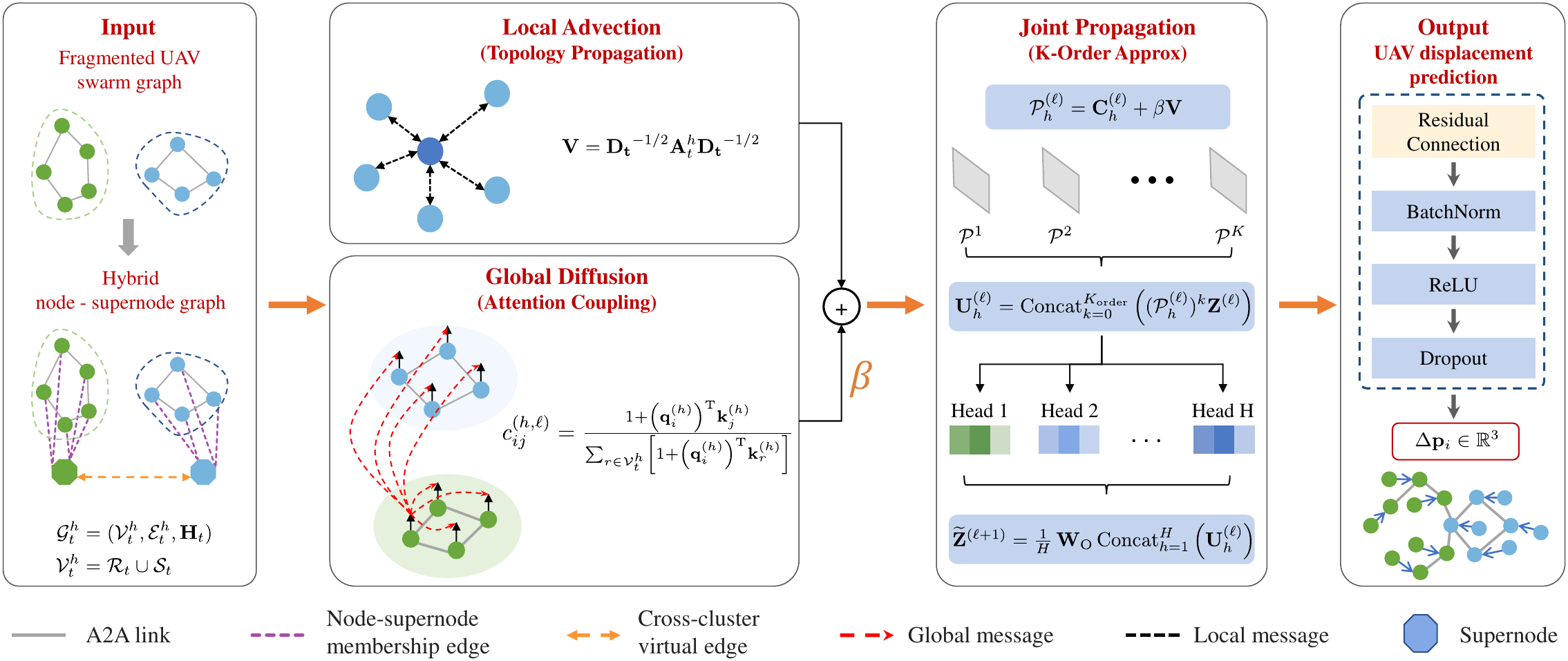}
    \caption{Architecture of AdDiff-GT, comprising five modules: 1) \emph{Residual Graph Processing}  constructs a hybrid graph from surviving UAVs and supernodes; 2) \emph{Local Advection} models intra-cluster topology and local motion; 3) \emph{Global Diffusion} captures cross-cluster interactions via multi-head attention; 4) \emph{Joint Propagation} integrates advection and diffusion through learnable multi-order propagation; and 5) \emph{Decoding Output} generates 3D recovery displacements.}
    \label{fig:model_architecture}
\end{figure*}

\subsubsection{Dynamics Modeling and Discretization} The core idea of AdDiff-GT is to model self-healing displacement prediction through advection-diffusion propagation over the hybrid graph. To connect the graph representation with the continuous dynamics in Eq.~\eqref{eq:manifold_advdiff}, we define graph gradient and divergence operators on  $\mathcal{G}_t^h$. For a node representation matrix $\mathbf{Z}=[\mathbf{z}_i]_{i\in\mathcal{V}_t^h}$, the graph gradient along edge $(i,j)$ is defined as  $(\nabla_{\mathcal{G}}\mathbf{Z})_{ij}=\mathbf{z}_j-\mathbf{z}_i$, while the divergence of an edge signal $\mathbf{E} = [\mathbf{e}_{ij}]$ is given by $(\nabla_{\mathcal{G}} ^{\ast} \mathbf{E})_i
=\sum_{j:(i,j)\in\mathcal{E}_t^h}\mathbf{e}_{ij}$. Under a linearized graph discretization of Eq.~\eqref{eq:manifold_advdiff}, the resulting dynamics are written as
\begin{equation}
\frac{\partial\mathbf{Z}(t)}{\partial t}
=\left[\mathbf{C}+\beta\mathbf{V}-\mathbf{I}\right]\mathbf{Z}(t),
\quad 0\le t\le T_{\mathrm{prop}},
\label{eq:advection_diffusion}
\end{equation}
where $\mathbf{V}=\mathbf{D}^{-1/2}_t\mathbf{A}_t^{h}\mathbf{D}^{-1/2}_t$ is the normalized local structure operator induced by the hybrid-graph adjacency matrix $\mathbf{A}_t^{h}$, and $\mathbf{C}$ denotes a non-local diffusion coupling matrix generated by multi-head attention. The operator $\mathbf{V}$ propagates information along observable communication links, $\mathbf{C}$ captures non-local interactions and cluster-level dependencies, $T_{\mathrm{prop}}$ denotes the propagation duration of intra-layer state evolution, and $\mathbf{Z}(0)$ denotes the initial node representation at the input of the layer. Thus, the two operators jointly encode local topology and global inter-cluster relationships for self-healing decisions.

\subsubsection{Multi-Order Propagation and Learnable Series Approximation} At the $\ell$-th layer, let $\mathbf{Z}^{(\ell)}$ denote the input node representation. The diffusion coupling weight matrix $\mathbf{C}_h^{(\ell)}$ for the $h$-th attention head is determined by the Query–Key similarity. Specifically, for each query node $i$ and key node $j$, linear projection and $L_2$-normalization are applied. Let $\mathbf{W}_{Q,h}$ and $\mathbf{W}_{K,h}$ denote the learnable query and key projection matrices, respectively, and let $\varepsilon_q>0$ be a numerical smoothing constant introduced to avoid division by zero. The normalized query and key vectors are given by
\begin{equation}
\mathbf{q}_i^{(h)} = \frac{\mathbf{W}_{Q,h}\mathbf{z}_i^{(\ell)}}
                        {\|\mathbf{W}_{Q,h}\mathbf{z}_i^{(\ell)}\|_2+\varepsilon_q}, \quad
\mathbf{k}_j^{(h)} = \frac{\mathbf{W}_{K,h}\mathbf{z}_j^{(\ell)}}
                        {\|\mathbf{W}_{K,h}\mathbf{z}_j^{(\ell)}\|_2+\varepsilon_q},
\end{equation}
and the non-negative similarity score is defined as
\begin{equation}
\eta_{ij}^{(h,\ell)}
=1+\left(\mathbf{q}_i^{(h)}\right)^{\mathrm T}\mathbf{k}_j^{(h)}.
\end{equation}
Since the normalized inner product satisfies $\bigl(\mathbf{q}_i^{(h)}\bigr)^{\!\top} \mathbf{k}_j^{(h)} \in [-1,1]$, we have $\eta_{ij}^{(h,\ell)}\ge 0$. Row normalization is then applied to obtain the  diffusion coupling weights: 
\begin{equation}
c_{ij}^{(h,\ell)}
=\frac{\eta_{ij}^{(h,\ell)}}{\sum_{r\in\mathcal{V}_t^h}\eta_{ir}^{(h,\ell)}},
\quad
\mathbf{C}_h^{(\ell)}
=[c_{ij}^{(h,\ell)}]_{i,j\in\mathcal{V}_t^h},\\
h=1,\ldots,H .
\label{eq:positive_attention_head}
\end{equation}
Here, the summation index $r$ traverses the node set $\mathcal{V}_t^h = \mathcal{R}_t \cup \mathcal{S}_t$ in the hybrid graph. Therefore,  $\mathbf{C}_h^{(\ell)}$ is a non-negative row-stochastic matrix that characterizes non-local diffusion coupling induced by attention head $h$. In contrast to the local structure operator $\mathbf{V}$, which propagates information through explicit hybrid-graph edges, $\mathbf{C}_h^{(\ell)}$ enables adaptive interactions among UAV and supernode pairs according to their learned feature similarity.
When 
$\mathbf{C}_h^{(\ell)}$  and $\mathbf{V}$ are fixed within a single propagation layer, the linearized dynamics in Eq. \eqref{eq:advection_diffusion} admit the closed-form solution
\begin{equation}
\mathbf{Z}(T_{\mathrm{prop}})
=\mathrm{e}^{(\mathbf{C}+\beta\mathbf{V}-\mathbf{I})T_{\mathrm{prop}}}\mathbf{Z}(0).
\label{eq:matrix_exp_solution}
\end{equation}

However, directly computing the matrix exponential for large-scale graphs incurs high computational costs. To avoid this overhead, we construct a learnable finite-order propagation approximation using propagation bases from the zeroth to the $K_{\mathrm{order}}$-th order. Specifically, for the $h$-th attention head, we define the joint advection-diffusion propagation operator as $\mathcal{P}_h^{(\ell)}=\mathbf{C}_h^{(\ell)}+\beta\mathbf{V}$. The propagation features from the order 0 to order $K_{\mathrm{order}}$ are then stacked as 
\begin{equation}
\mathbf{U}_h^{(\ell)}
=\operatorname{Concat}_{k=0}^{K_{\mathrm{order}}}
\left((\mathcal{P}_h^{(\ell)})^{k}\mathbf{Z}^{(\ell)}\right).
\label{eq:series_stack}
\end{equation}
Subsequently, the multi-order propagation features from all attention heads are concatenated along the feature dimension and aggregated through a single learnable linear output mapping:
\begin{equation}
\widetilde{\mathbf{Z}}^{(\ell+1)}
=
\frac{1}{H}
\mathbf{W}_{\mathrm{o}}
\operatorname{Concat}_{h=1}^{H}
\left(\mathbf{U}_h^{(\ell)}\right),
\label{eq:series_forward}
\end{equation}
where $\operatorname{Concat}_{h=1}^{H}(\cdot)$ denotes the concatenation of the multi-order propagation features from all attention heads,  $\mathbf{W}_{\mathrm{o}}$ is a learnable linear projection matrix that maps the concatenated high-dimensional features to the target dimension, and the scaling factor $1/H$ balances the contributions of different attention heads. Accordingly, the proposed propagation can be viewed as a learnable finite-order polynomial approximation of the layer-wise advection-diffusion evolution, where information from different propagation orders and attention heads is adaptively fused through a single linear mapping:
\begin{equation}
\mathbf{Z}(T_{\mathrm{prop}}) \approx \frac{1}{H}\, \mathbf{W}_{\mathrm{o}} \,
\operatorname{Concat}_{h=1}^{H}
\Bigl[ \mathbf{Z}(0),\,
\mathcal{P}_h\mathbf{Z}(0), 
 \ldots,\,
\mathcal{P}_h^{K_{\mathrm{order}}}\mathbf{Z}(0) \Bigr].
\label{eq:learnable_series}
\end{equation}

To further enhance the stability of deep propagation, the output of each layer is processed by a residual connection, batch normalization (BN), ReLU activation, and dropout regularization: 
\begin{align}
\bar{\mathbf{Z}}^{(\ell+1)}
&={\mathbf{Z}^{(\ell)}
+\widetilde{\mathbf{Z}}^{(\ell+1)}},\\
\mathbf{Z}^{(\ell+1)}
&={\operatorname{Dropout}\!\left(
\operatorname{ReLU}\!\left(\operatorname{BN}(\bar{\mathbf{Z}}^{(\ell+1)})\right)
\right)}.
\label{eq:series_residual}
\end{align}

The above formulation establishes a correspondence between the continuous advection-diffusion dynamics and the forward propagation of AdDiff-GT. This correspondence is summarized by the following layer-wise consistency property. 

\subsubsection{Output Mapping and Displacement Prediction} Finally, a linear decoder maps the learned node representations to 3D displacement vectors $\Delta\mathbf{P}=f_{\mathbf{\Theta}}(\mathcal{G}_t^{h})\in\mathbb{R}^{|\mathcal{R}_t|\times3}$,
where $f_{\mathbf{\Theta}}$ denotes the forward mapping of AdDiff-GT parameterized by set $\mathbf{\Theta}$.
During inference, the predicted displacements of the $|\mathcal{R}_t|$ surviving UAV nodes are converted into maneuvering commands, whereas supernodes participate only in feature propagation and cluster-level structural modeling.

\subsection{Recovery Objective and Fast Adaptation}
\subsubsection{Smallest Connected Distance (SCD)}
The SCD quantifies the minimum virtual communication radius required for the surviving graph $\mathcal{G}'_t$ to become fully connected. Let $d \ge R_c$ be a candidate communication radius, and let $\mathcal{G} (d)$  be the topology constructed by replacing $R_c$ 
with $d$ as the connectivity threshold. If BFS identifies a single connected cluster in $\mathcal{G} (d)$, then $d$ is feasible. Accordingly, the SCD is defined as $d_{\mathrm{SCD}}=\min\{d\mid C(\mathbf{P};d)=1,\ d\ge R_c\}$.
\begin{algorithm}[h]\small
\SetKwInput{KwIn}{Inputs}
\SetKwInput{KwOut}{Output}
\KwIn{Surviving node positions $\mathbf{P}_t$, physical communication radius $R_c$}
\KwOut{$d_{\mathrm{SCD}}$}
Compute all pairwise distances and sort unique candidates no smaller than $R_c$\;
\For{\textbf{each} candidate radius $d$ in ascending order}{
    Construct adjacency matrix $\mathbf{A}(d)$ by thresholding distances with $d$\;
    $\mathbf{L}(d)\leftarrow\mathbf{D}(d)-\mathbf{A}(d)$\;
    $(C,b)\leftarrow$ Algorithm~\ref{alg:bfs_connectivity} $(\mathcal{R}_t,\mathbf{L}(d))$\;
    \If{$b=1$}{
    $d_{\mathrm{SCD}}\leftarrow d$\;
        \KwRet $d_{\mathrm{SCD}}$\;
    }
}
\KwRet the largest candidate distance\;
\caption{SCD Calculation}
\label{alg:scd}
\end{algorithm}

Computationally, $d_{\mathrm{SCD}}$ is obtained by sorting the pairwise Euclidean distances among surviving UAVs and progressively activating the corresponding edges until the threshold graph $\mathcal{G} (d)$ becomes connected.
Equivalently, $d_{\mathrm{SCD}}$ is the maximum edge weight of a minimum-bottleneck spanning tree of the complete distance-weighted graph. Thus, it represents the minimum virtual communication range required to bridge the connectivity bottleneck of the current partitioned topology.
Increasing the radius virtually from $R_c$ to $d_{\mathrm{SCD}}$ would restore connectivity without physical movement. We define the topological recovery threshold as $T_m=d_{\mathrm{SCD}}-R_c$ to quantify the current degree of network connectivity impairment a priori. Algorithm~\ref{alg:scd} summarizes the connectivity-based procedure for computing $d_{\mathrm{SCD}}$.

Therefore, $d_{\mathrm{SCD}}$ is not a heuristic tuning metric but a topology-aware quantity determined by the connectivity bottleneck of the partitioned topology. In contrast, $T_m$ quantifies the additional spatial scale that physical UAV movement must compensate for under the fixed communication radius $R_c$.



\subsubsection{Loss Function for AdDiff-GT}
To enable efficient and constraint-aware self-healing, we define a weighted composite loss function that jointly accounts for connectivity restoration, maneuvering cost, motion consistency, inter-cluster coordination, and spatial preservation. Let $\mathcal{W} = \{ w_c, w_m, w_{\max}, w_s, w_a, w_p, w_q, w_l \}$ denote the corresponding loss weights. The overall loss is defined as 
\begin{align}
\mathcal{L}
=&\ w_c\mathcal{L}_{\mathrm{conn}}
+w_m\mathcal{L}_{\mathrm{move}}
+w_{\max}\mathcal{L}_{\max}
+w_s\mathcal{L}_{\mathrm{smooth}} \nonumber\\
&+w_a\mathcal{L}_{\mathrm{align}}
+w_p\mathcal{L}_{\mathrm{pull}}
+w_q\mathcal{L}_{\mathrm{soft}}
+w_l\mathcal{L}_{\mathrm{link}} .
\label{eq:loss}
\end{align}
Let $\widehat{\mathbf{p}}_i=\mathbf{p}_{i,t}+\Delta\mathbf{p}_{i}$ denote the candidate recovery position of surviving UAV $i$. The individual loss terms are defined as follows:

\emph{Connectivity Term}:
\begin{equation}
\mathcal{L}_{\mathrm{conn}}=C(\widehat{\mathbf{P}})-1,
\end{equation}
where the corresponding weight coefficient $w_c$ is set sufficiently large to impose a strong penalty on disconnected candidate topologies.

\emph{Average Displacement Term}:
\begin{equation}
\mathcal{L}_{\mathrm{move}}=\frac{1}{|\mathcal{R}_t|}\sum_{i\in\mathcal{R}_t}\|\Delta\mathbf{p}_i\|_2,
\end{equation}
which penalizes the average UAV displacement and serves as a surrogate for overall maneuvering cost.

\emph{Maximum Displacement Violation Term}:
\begin{equation}
\mathcal{L}_{\max}=
\frac{1}{|\mathcal{R}_t|}\sum_{i\in\mathcal{R}_t}
\max\left(0,\|\Delta\mathbf{p}_i\|_2-T_m\right),
\end{equation}
which penalizes UAV displacements that exceed the topology-aware recovery scale $T_m$.

\emph{Local Smoothness Term}:
\begin{equation}
\mathcal{L}_{\mathrm{smooth}}=
\frac{1}{|\mathcal{E}_{\mathrm{phy}}|}\sum_{(i,j)\in\mathcal{E}_{\mathrm{phy}}}
\|\Delta\mathbf{p}_i-\Delta\mathbf{p}_j\|_2,
\end{equation}
which suppresses abrupt displacement differences among physically adjacent UAVs and promotes locally consistent maneuvering.

\emph{Component Alignment Term}:
\begin{equation}
\mathcal{L}_{\mathrm{align}}=
\frac{1}{C_t}\sum_{c=1}^{C_t}\left\|
\frac{1}{|\mathcal{C}_c|}\sum_{i\in\mathcal{C}_c} \Delta\mathbf{p}_i
-\Delta\mathbf{p}_{s_c}
\right\|_2,
\end{equation}
which encourages the UAV-level displacement fields within each cluster to remain consistent with the corresponding cluster-level recovery direction.

\emph{Inter-Cluster Attraction Term}:
\begin{equation}
\mathcal{L}_{\mathrm{pull}}=
\frac{2}{C_t(C_t-1)}
\sum_{1\le c<k\le C_t}
\|\widehat{\mathbf{s}}_c-\widehat{\mathbf{s}}_k\|_2 ,
\end{equation}
where $\widehat{\mathbf{s}}_c=\frac{1}{|\mathcal{C}_c|}\sum_{i\in\mathcal{C}_c}\widehat{\mathbf{p}}_i$ denotes post-recovery center of cluster $\mathcal{C}_c$. This term provides a continuous optimization signal that encourages disconnected clusters to move toward one another and therefore facilitates cross-cluster connectivity restoration.

\emph{Spatial Dispersion Term}:
To prevent the inter-cluster attraction term from causing excessive aggregation, we introduce a soft minimum distance penalty:
\begin{equation}
\mathcal{L}_{\mathrm{soft}}=
\max\left(0,\ d_{\mathrm{tar}} +
\tau_{\mathrm{soft}}\log\sum_{c<k}\mathrm{e}^{\left(-\frac{\|\widehat{\mathbf{s}}_c-\widehat{\mathbf{s}}_k\|_2}{\tau_{\mathrm{soft}}}\right)}\right),
\end{equation}
where $d_{\mathrm{tar}}$ is adaptively determined according to the current connected-cluster ratio,  and $\tau_{\mathrm{soft}}>0$ is the soft-min temperature. This temperature controls the sensitivity of the soft minimum distance to the closest cluster pair.

\emph{Minimum Link Spacing Term}: To regulate the spatial separation of UAVs after recovery, we define
\begin{equation}
\mathcal{L}_{\mathrm{link}}=
\max\left(0,\ d_{\mathrm{link}}-\min_{(i,j)\in\mathcal{E}'_t(\widehat{\mathbf{P}})}
\|\widehat{\mathbf{p}}_i-\widehat{\mathbf{p}}_j\|_2\right),
\end{equation}
which penalizes candidate communicating UAV pairs whose spacing falls below the threshold $d_{\mathrm{link}}$.

The above loss terms provide differentiable objectives or surrogate penalties for the objectives and constraints in Eq. \eqref{eq:problem}. Specifically, $\mathcal{L}_{\mathrm{conn}}$ and $\mathcal{L}_{\mathrm{pull}}$ promote connectivity restoration, while $\mathcal{L}_{\mathrm{align}}$ couples cluster-level recovery with UAV-level displacement. Together, these terms facilitate satisfaction of the connectivity constraint in Eq.~\eqref{eq:problem_conn}. Meanwhile, 
$\mathcal{L}_{\mathrm{smooth}}$, $\mathcal{L}_{\mathrm{move}}$, and $\mathcal{L}_{\max}$ regulate motion consistency and maneuvering cost. Finally, $\mathcal{L}_{\mathrm{soft}}$ and $\mathcal{L}_{\mathrm{link}}$ provide soft surrogates for the spatial-dispersion objective and safety-spacing constraint in Eq.~\eqref{eq:safe-dis}. Thus, the composite loss jointly balances connectivity restoration, maneuvering efficiency, motion consistency, and spatial preservation.
 
During online model adaptation, a greedy caching strategy is employed to retain promising candidate solutions. When the candidate topology remains disconnected, updates prioritize reducing $\mathcal{L}_{\mathrm{conn}}$ and the overall loss. Once connectivity is achieved, a candidate is accepted only if it improves the maximum displacement or minimum link spacing without causing a significant increase in the overall recovery cost. Early stopping is triggered when connectivity remains stable, the relevant penalty terms approach zero, or no effective improvement is observed over consecutive adaptation steps.


\subsubsection{Meta-Learning for Fast Adaptation}
To enable rapid adaptation to unseen
UED induced damage topologies, AdDiff-GT adopts a  first-order meta-initialization strategy inspired by model-agnostic meta-learning (MAML)~\cite{finn2017model}.  Rather than adapting from random initialization or maintaining topology-specific models, we learn a scale-specific 
meta-prior $\mathbf{\Theta}_n^\star$ 
for each surviving-node scale $n\in\{2,\ldots,N\}$,  which is retrieved according to the inferred network scale. The first-order strategy avoids Hessian computation and constructs the meta-update from the parameter deviation between query-set correction and support-set adaptation, thereby reducing offline optimization overhead for large-scale graph networks. 

\textbf{Offline Meta-Learning}:
For the $n$-th meta-task, unconnected support samples $\mathcal{S}_n$ and query samples $\mathcal{Q}_n$ are randomly sampled, where initially connected samples are excluded using BFS. Each sample is converted into a node-supernode hybrid graph. Given the meta-parameters $\mathbf{\Theta}_n^{(e-1)}$ at epoch $e$, support-set adaptation is first performed by initializing $\widetilde{\mathbf{\Theta}}_n^{(e,0)}=\mathbf{\Theta}_n^{(e-1)}$ and applying $M$ gradient steps:
\begin{equation}
\widetilde{\mathbf{\Theta}}_n^{(e,m)}
=\widetilde{\mathbf{\Theta}}_n^{(e,m-1)}
-\alpha_e\nabla_{\mathbf{\Theta}}\mathcal{L}
\left(\widetilde{\mathbf{\Theta}}_n^{(e,m-1)};\mathcal{S}_n^{(e)}\right),
\end{equation}
where $m=1,\ldots,M$ and $\alpha_e$ denotes the support-set adaptation learning rate, which is decayed over training epochs to improve adaptation stability. Starting from $\widetilde{\mathbf{\Theta}}_n^{(e,M)}$, query-set correction is subsequently performed for $Q$ steps:
\begin{equation}
\bar{\mathbf{\Theta}}_n^{(e,q)}
=\bar{\mathbf{\Theta}}_n^{(e,q-1)}
-\alpha_e\nabla_{\mathbf{\Theta}}\mathcal{L}
\left(\bar{\mathbf{\Theta}}_n^{(e,q-1)};\mathcal{Q}_n^{(e)}\right),
\end{equation}
where $q=1,\ldots,Q$, $\bar{\mathbf{\Theta}}_n^{(e,0)}=\widetilde{\mathbf{\Theta}}_n^{(e,M)}$.

The meta-update direction is defined as $\Delta\mathbf{\Theta}_n^{(e)}
=\bar{\mathbf{\Theta}}_n^{(e,Q)}
-\widetilde{\mathbf{\Theta}}_n^{(e,M)}$, the meta-parameters are updated by
\begin{equation}
\mathbf{\Theta}_n^{(e)}
=\mathbf{\Theta}_n^{(e-1)}
+\gamma_e\Delta\mathbf{\Theta}_n^{(e)},
\end{equation}
where $\gamma_e$ denotes the meta-update  rate. After offline training, the resulting scale-specific meta-parameter library $\{\mathbf{\Theta}_n^\star\}$ is stored for online recovery. The complete procedure is summarized in Algorithm \ref{alg:meta}.

\begin{algorithm}[h]\small
\SetKwInput{KwIn}{Inputs}
\SetKwInput{KwOut}{Output}
\KwIn{Surviving scale set $\{2,\ldots,N\}$, communication radius $R_c$, training epochs $E$}
\KwOut{Meta-parameter library $\{\mathbf{\Theta}_n^\star\}$}
\For{$n=2,\ldots,N$}{
    Initialize meta parameters $\mathbf{\Theta}_n^{(0)}$\;
    \For{$e=1,\ldots,E$}{
        Load $\mathbf{\Theta}_n^{(e-1)}$ into AdDiff-GT\;
        Sample disconnected support scenario $\mathcal{S}_n$\;
        Adapt on $\mathcal{S}_n$ to obtain $\widetilde{\mathbf{\Theta}}_n^{(e,M)}$\;
        Sample disconnected query scenario $\mathcal{Q}_n$\;
        Evaluate/correct from $\widetilde{\mathbf{\Theta}}_n^{(e,M)}$ on $\mathcal{Q}_n$ to obtain $\bar{\mathbf{\Theta}}_n^{(e,Q)}$\;
        $\Delta\mathbf{\Theta}_n^{(e)}\leftarrow\bar{\mathbf{\Theta}}_n^{(e,Q)}-\widetilde{\mathbf{\Theta}}_n^{(e,M)}$\;
        $\mathbf{\Theta}_n^{(e)}\leftarrow\mathbf{\Theta}_n^{(e-1)}+\gamma_e\Delta\mathbf{\Theta}_n^{(e)}$\;
    }
    Save $\mathbf{\Theta}_n^\star\leftarrow\mathbf{\Theta}_n^{(E)}$\;
}
\caption{Offline Meta-Learning of AdDiff-GT}
\label{alg:meta}
\end{algorithm}

\textbf{Online Rapid Adaptation}:
Given the inferred surviving-node scale $n_t=|\mathcal{R}_t|$, the GCS retrieves the corresponding meta-initialization $\mathbf{\Theta}_{n_t}^{\star}$ and performs a limited gradient update on the current hybrid graph: 
\begin{equation}
\mathbf{\Theta}^{+}\leftarrow
\mathbf{\Theta}_{n_t}^{\star}
-\eta_{\mathrm{ol}}\nabla_{\mathbf{\Theta}}\mathcal{L}\left(\mathbf{\Theta}_{n_t}^{\star};\mathcal{G}_t^{h}\right),
\end{equation}
where $\eta_{\mathrm{ol}}$ denotes the online adaptation learning rate. The adapted parameters $\mathbf{\Theta}^{+}$ are subsequently broadcast to the surviving UAVs, which perform forward inference locally without backpropagation. This design confines gradient-based adaptation to the GCS while enabling lightweight and low-latency execution on airborne nodes.

\subsection{Theoretical Analysis}\label{sec: Theoretical-Analysis}
This section establishes three theoretical properties of AdDiff-GT: online tuning stability, robustness to topological perturbations, and finite-step connectivity recovery.
\begin{proposition}[Convergence of Online Self-Healing Tuning]
\label{prop:convergence}
Suppose that the differentiable surrogate loss $\mathcal{L}$ used for online tuning is lower-bounded over the bounded flight region and has Lipschitz-continuous gradients with respect to the UAV positions. Then, under a standard step-size condition, gradient-based online tuning generates a monotonically non-increasing loss sequence and converges to a first-order stationary point of the surrogate objective, corresponding to a locally stable self-healing state.
\end{proposition}

This proposition characterizes the stability of the GCS-side online tuning. Since the loss terms are non-negative, $\mathcal{L} \ge 0$ and is therefore bounded below during online optimization. Under the stated smoothness and step-size conditions, each update yields sufficient decrease in the surrogate objective, preventing unbounded oscillations during online tuning. 


\begin{proposition}[Representation Boundedness under Topological Perturbations]
\label{prop:boundedness}
Let a topological perturbation induce a variation $\Delta\mathbf{V}$ in the normalized adjacency operator and a corresponding variation $\Delta\mathbf{C}$ in the global attention coupling matrix. Suppose that, over the trained bounded input domain, there exist constants $a_1,a_2>0$ such that
\begin{equation}
\label{compensation condition}
\|\beta\Delta\mathbf{V}+\Delta\mathbf{C}\|_2
\le a_1\|\Delta\mathbf{V}\|_2
+a_2\|\Delta\mathbf{V}\|_2^2.
\end{equation}
Assume further that the propagation spectral radius, the number of propagation layers, and the Lipschitz constants of encoder and decoder are bounded. Then, the perturbation of the output representation satisfies
\begin{equation}
\|\Delta\mathbf{Z}_{\mathrm{out}}\|_2
\le K_1\|\Delta\mathbf{V}\|_2
+K_2\|\Delta\mathbf{V}\|_2^2,
\end{equation}
where $K_1,K_2>0$ are constants determined by the propagation depth, spectral bounds, and network Lipschitz constants.   
\end{proposition}

This proposition characterizes the sensitivity of AdDiff-GT representations to topological perturbations rather than providing a generalization error bound. Local adjacency changes alter message passing paths, while global attention and supernode channels provide additional cross-cluster coupling. Under the stated compensation condition, representation perturbation is bounded by a second-order function of the adjacency perturbation. 



\begin{corollary}[Finite-Step Connectivity Recovery]
\label{cor:finite_step}
If the initial post-damage topology consists of $C(0)=m$ connected clusters and, for every disconnected state, the proposed update produces a nonzero decrease in the inter-cluster attraction potential that persists until a cluster merge occurs, then there exists a finite number of recovery steps  $T_c<\infty$ such that $C(T_c)=1$.
\end{corollary}

This corollary characterizes the finite-step connectivity-recovery property induced by inter-cluster attraction and supernode coupling. Starting from $m<\infty$ connected clusters, each effective merge reduces the number of  cluster by at least one. Therefore, a finite number of merges is  sufficient to reach a connected topology.

The proofs of the above two propositions and the corollary are provided in Appendix~\ref{app:theory}.


\section{Experimental Evaluation and Analysis}
\label{sec:simulation}
\subsection{Experiment Settings}
\subsubsection{Scenario Configuration}
To evaluate AdDiff-GT, we consider two partitioned USNET settings with different scales.

\textbf{Setting I (Small-to-Medium Scale)}: 
The simulation area is  $1000\times1000\times100 \, \mathrm{m}^3$, with an initial healthy network of $N=200$ UAVs. The number of surviving UAVs $n_t$ is uniformly sampled from $[2,200]$, yielding 3,000 partitioned scenarios with a maximum of $T_{\max}=1000$ recovery steps.

\textbf{Setting II (Large-Scale)}:
The simulation area is enlarged to $2000\times2000\times200\,   \mathrm{m}^3$, with $N=400$ UAVs. The number of surviving UAVs $n_t$ is uniformly sampled from $[2,400]$, generating  5,000  scenarios with  $T_{\max}=3000$.

In both settings, the A2A  communication radius and UAV speed are set to $R_c=120\, \mathrm{m}$ and $1\,\mathrm{m/s}$, respectively. BFS-based connectivity detection is used to exclude initially connected topologies, ensuring that all methods are evaluated on the same partitioned scenarios without additional test-time priors.

\subsubsection{Baseline Methods}
We compare AdDiff-GT with three representative classes of topology recovery methods: geometric methods (HERO~\cite{mi2011hero} and CEN~\cite{zhao2020efficient}), swarm-intelligence-based heuristic method (SIDR~\cite{chen2020sidr}), and learning-based methods (GCN-2017~\cite{kipf2017semisupervised}, CR-MGC~\cite{mou2022resilient}, and GDR-TS~\cite{wang2025uav}). These baselines cover deterministic geometric reconstruction, heuristic swarm-based recovery, and graph-based learning, respectively.

\subsubsection{Implementation Details}
The AdDiff-GT adopts 13D node features, 128-dimensional hidden representations, two advection-diffusion propagation layers, four attention heads, and a propagation order of
 $K_{\mathrm{order}}=2$.
 The decoder outputs 3D displacement vectors for surviving UAVs. The hybrid graph uses  $w^{\mathrm{phy}}_{ij}=1$ for physical edges and
 $w^{\mathrm{ns}}_{is_c} = 0.8$ for node-supernode association edges.  The advection coefficient, dropout rate, and propagation duration are set to $\beta=0.6$, $0.1$, and $T_{\mathrm{prop}}=3000$, respectively. Online tuning employs Adam with learning rate $\eta_{\mathrm{ol}}=10^{-3}$ and weight decay $10^{-5}$. The loss weights are $\mathcal{W} = \{ 10^4,50,220,2,10,500,80,500 \}$, with $d_{\mathrm{link}}=0.75R_c$. For meta-learning, $\alpha_e$ follows a three-stage schedule of $10^{-3}$, $10^{-4}$, and $10^{-5}$ over epochs 1–100, 101–250, and 251–500, respectively, with $\gamma_e=1$, $M=Q=1$, and $E=500$.
  To isolate the effect of meta-learning, AdDiff-GT performs online adaptation without meta-initialization, whereas AdDiff-GT+Meta uses scale-specific meta-parameters $\mathbf{\Theta}_{|\mathcal{R}_t|}^{\star}$ followed by limited online adaptation. Experiments are conducted on an AMD EPYC 7742 64-Core CPU and an NVIDIA RTX 4090 GPU. 
  

\subsubsection{Evaluation Metrics}
We evaluate self-healing performance using the following five metrics:

 \textbf{Recovery Steps~($T_{\mathrm{rec}}$)}: The number of recovery steps required to restore full connectivity ($C_t=1$). For unsuccessful cases, $T_{\mathrm{rec}} = T_{\max}$. We report both the average over all scenarios and the average over successfully recovered scenarios.
 
 \textbf{Online Planning Time~($T_{\mathrm{plan}}$)}: The wall-clock time from receiving the partitioned topology to generating recovery decisions, excluding physical maneuver execution.

\textbf{Communication Cost~($C_\mathrm{cost}$)}: Measured by interaction rounds, transmitted bits, and receiver-side message volume under broadcast and multicast.

\textbf{Spatial Dispersion~($D_{\mathrm{disp}}$)}: The spatial distribution of recovered UAVs is quantified using normalized grid entropy:
 \begin{equation}
    D_{\mathrm{disp}}=
        \frac{
            -\sum_{g=1}^{N_{\mathrm{g}}}p_g\log_2(p_g+\epsilon_{\mathrm{d}})
        }{
            \log_2 N_{\mathrm{g}}
        }, \quad\epsilon_{\mathrm{d}}>0,
\end{equation}
where $p_g$ is the proportion of UAVs in grid $g$, and $N_{\mathrm{g}}$ is the total number of grids. A higher value indicates better spatial dispersion after recovery.

\textbf{Scenario Win Rates (Best Connect. / Disp. Ratio)}: To quantify scenario-level comparative performance, we define the recovery-efficiency win rate $\eta_{\mathrm{con}}$ and spatial-dispersion win rate $\eta_{\mathrm{dis}}$. Let $\mathcal{S}_{\mathrm{test}}$ denote the test scenario set and let $\mathcal{A}$ denote the set of comparative algorithms. For each $a\in\mathcal{A}$ and $s\in \mathcal{S}_{\mathrm{test}}$, let $T_{\mathrm{rec}}^a(s)$ and $D_{\mathrm{disp}}^a(s)$ denote the recovery steps and spatial dispersion achieved by method $a$ in scenario $s$, respectively. The recovery win rate is defined as
\begin{equation}
    \eta_{\mathrm{con}}(a)=
    \frac{100\%}{|\mathcal{S}_{\mathrm{test}}|}
    \sum_{s\in\mathcal{S}_{\mathrm{test}}}
    \frac{\mathbb{I}\!\left[a\in\arg\min_{b\in\mathcal{A}}T_{\mathrm{rec}}^b(s)\right]}
    {\left|\arg\min_{b\in\mathcal{A}}T_{\mathrm{rec}}^b(s)\right|},
    \end{equation}
    \begin{equation}
    \eta_{\mathrm{dis}}(a)=
    \frac{100\%}{|\mathcal{S}_{\mathrm{test}}|}
    \sum_{s\in\mathcal{S}_{\mathrm{test}}}
    \frac{\mathbb{I}\!\left[a\in\arg\max_{b\in\mathcal{A}}D_{\mathrm{disp}}^b(s)\right]}
    {\left|\arg\max_{b\in\mathcal{A}}D_{\mathrm{disp}}^b(s)\right|},
\end{equation}
where the denominator equally divides the winning credit among tied methods, so the ratios sum to 100\%.

\begin{table*}[t]
\centering
\caption{Performance Comparison across Two Experimental Settings}
\label{tab:experiment_results}
\footnotesize
\begin{threeparttable}
\setlength{\tabcolsep}{4.5pt}
\begin{tabular}{lccccccc}
\toprule
\multirow{2}{*}{\textbf{Algorithm}} &
\multicolumn{2}{c}{\textbf{Avg. Recovery Steps}} &
\multirow{2}{*}{\textbf{\makecell{Avg. Spatial\\Dispersion}}} &
\multirow{2}{*}{\textbf{\makecell{Avg. Planning\\Time (s)}}} &
\multirow{2}{*}{\textbf{\makecell{Unconnected\\Cases}}} &
\multirow{2}{*}{\textbf{\makecell{Best Connect.\\Ratio (\%)}}} &
\multirow{2}{*}{\textbf{\makecell{Best Disp.\\Ratio (\%)}}} \\
\cmidrule(lr){2-3}
& \textbf{Successful Only} & \textbf{All Scenarios} & & & & & \\
\midrule
\multicolumn{8}{l}{\textbf{(a) Setting I}: $1000\times1000\times100 \, \mathrm{m}^3$, $N=200$, 3000 scenarios} \\
\midrule
HERO & 378.85 & 455.04 & 0.7757 & 0.0666 & 368 & 0.25 & 3.32 \\
CEN & 283.45 & 283.45 & 0.6752 & \textbf{0.0271} & \textbf{0} & 2.43 & 1.16 \\
SIDR & \textbf{137.17} & 823.41 & 0.9088 & \underline{0.0494} & 2386 & 0.55 & 5.82 \\
GCN-2017 & 297.43 & 316.40 & 0.8041 & 44.5903 & 81 & 1.92 & 0.72 \\
CR-MGC & 256.38 & 256.63 & 0.7851 & 10.9692 & \underline{1} & 4.53 & 1.15 \\
CR-MGC+Meta & 225.68 & 234.71 & 0.8013 & 16.3081 & 35 & 8.44 & 1.08 \\
GDR-TS & 255.49 & 265.91 & 0.8397 & 12.5365 & 42 & 4.98 & 2.72 \\
AdDiff-GT & 162.38 & \underline{162.38} & \textbf{0.9548} & 10.6942 & \textbf{0} & \underline{30.47} & \textbf{43.68} \\
AdDiff-GT+Meta & \underline{150.72} & \textbf{150.72} & \underline{0.9540} & 9.1844 & \textbf{0} & \textbf{46.43} & \underline{40.35} \\
\midrule
\multicolumn{8}{l}{\textbf{(b) Setting II}: $2000\times2000\times200 \, \mathrm{m}^3$, $N=400$, 5000 scenarios} \\
\midrule
HERO & 1795.80 & 1797.49 & 0.4802 & 0.2863 & \underline{7} & 0.00 & 0.09 \\
CEN & 1017.39 & 1017.39 & 0.3493 & \textbf{0.0933} & \textbf{0} & 2.48 & 0.16 \\
SIDR & -- & 3000.00 & 0.9519 & \underline{0.1405} & 5000 & 0.00 & 2.12 \\
GCN-2017 & 980.89 & 1152.51 & 0.7381 & 143.6654 & 425 & 1.17 & 0.42 \\
CR-MGC & 1168.69 & 1203.48 & 0.7652 & 71.7710 & 95 & 1.86 & 0.15 \\
CR-MGC+Meta & 1815.53 & 1824.30 & 0.3873 & 66.0270 & 37 & 0.47 & 0.05 \\
GDR-TS & 1018.97 & 1046.31 & 0.6796 & 1.1336 & 69 & 2.64 & 0.15 \\
AdDiff-GT & \textbf{681.44} & \textbf{681.44} & \textbf{0.9900} & 17.6835 & \textbf{0} & \textbf{59.53} & \textbf{62.38} \\
AdDiff-GT+Meta & \underline{715.73} & \underline{715.73} & \underline{0.9881} & 11.1862 & \textbf{0} & \underline{31.85} & \underline{34.49} \\
\bottomrule
\end{tabular}
\par\vspace{2pt}
\begin{tablenotes}
    \footnotesize
    \item 
    \textit{Note}: 
    The \textbf{best} and \underline{second-best} values are shown in bold and underlined, respectively.
\end{tablenotes}
\end{threeparttable}
\end{table*}
\subsection{Performance Evaluation}
\subsubsection{Comparative Recovery Performance}
Table~\ref{tab:experiment_results} compares the recovery performance under the two experimental settings. In Setting I, AdDiff-GT+Meta achieves an average recovery step $T_{\mathrm{rec}}$ of 150.72 with a 100$\%$ connectivity success rate.
AdDiff-GT without meta-initialization requires slightly more steps (162.38) while maintaining a high spatial dispersion of $D_{\mathrm{disp}}$  of 0.9548. Although SIDR achieves relatively lower recovery steps among successful cases, it fails in 2,386 scenarios, resulting in an all-scenario average of 823.41 when unsuccessful cases are assigned $T_{\max}=1000$. CEN achieves zero recovery failures and lower planning latency, but its spatial dispersion is limited to 0.6752, indicating a more concentrated post-recovery topology. 

In Setting II,  AdDiff-GT achieves an average recovery step of 681.44 with a 100$\%$ connectivity success rate under both successful-only and all-scenario statistics, indicating consistent performance across the evaluated partitioned topologies. Compared with CEN, which also achieves full recovery, AdDiff-GT requires fewer recovery steps and attains substantially higher spatial dispersion (0.9900 versus 0.3493). This result demonstrates that the proposed hybrid graph representation effectively captures inter-cluster relationships while avoiding excessive spatial aggregation during connectivity restoration.

The learning-based baselines exhibit reduced robustness as the network scale increases. In Setting II, GCN-2017, CR-MGC, and GDR-TS obtain all-scenario average recovery steps of 1152.51, 1203.48, and 1046.31,  with 425, 95, and 69 unconnected scenarios, respectively.  GDR-TS achieves a lower planning time of 1.1336 s but obtains a spatial dispersion of 0.6796 and a connectivity win rate of 2.64$\%$. SIDR maintains a high dispersion value 0.9519; however, since it fails to restore connectivity in all 5,000 scenarios, this value primarily reflects the spatial separation of disconnected clusters rather than the quality of a successfully recovered topology.

The comparative experimental results fully demonstrate that AdDiff-GT maintains stable connectivity recovery while preserving high spatial dispersion across both network scales. In particular, its performance remains stable as the number of surviving UAVs and partitioning complexity increase, demonstrating the scalability and robustness of the proposed graph-based self-healing mechanism.

\subsubsection{Recovery Evolution Dynamics}
Fig.~\ref{fig:cluster_convergence_cases} illustrates the evolution of the number of connected clusters $C_t$ over recovery steps for three representative partitioned scenarios in Setting II, with $n_t=$ 28, 193, and 292 surviving UAVs. For the highly fragmented case with $n_t=$ 28, AdDiff-GT restores full connectivity ($C_t=1$) in approximately 1,170 steps, whereas most baselines remain trapped in partially connected states. As the surviving scale increases to $n_t=$ 193 and 292, AdDiff-GT exhibits a more rapid reduction in cluster count and progressively merges the disconnected clusters.
These representative dynamics are consistent with the aggregate results in Table~\ref{tab:experiment_results}. Although Fig.~\ref{fig:cluster_convergence_cases} reflects individual scenarios rather than averaged performance, the sustained reduction in connected-cluster count provides qualitative evidence of the progressive topology recovery achieved by AdDiff-GT and helps explain its consistent recovery performance across large-scale partitioned networks.

\begin{figure*}[ht]
\centering
\begin{minipage}[b]{0.28\textwidth}
\centering
\includegraphics[width=\linewidth]{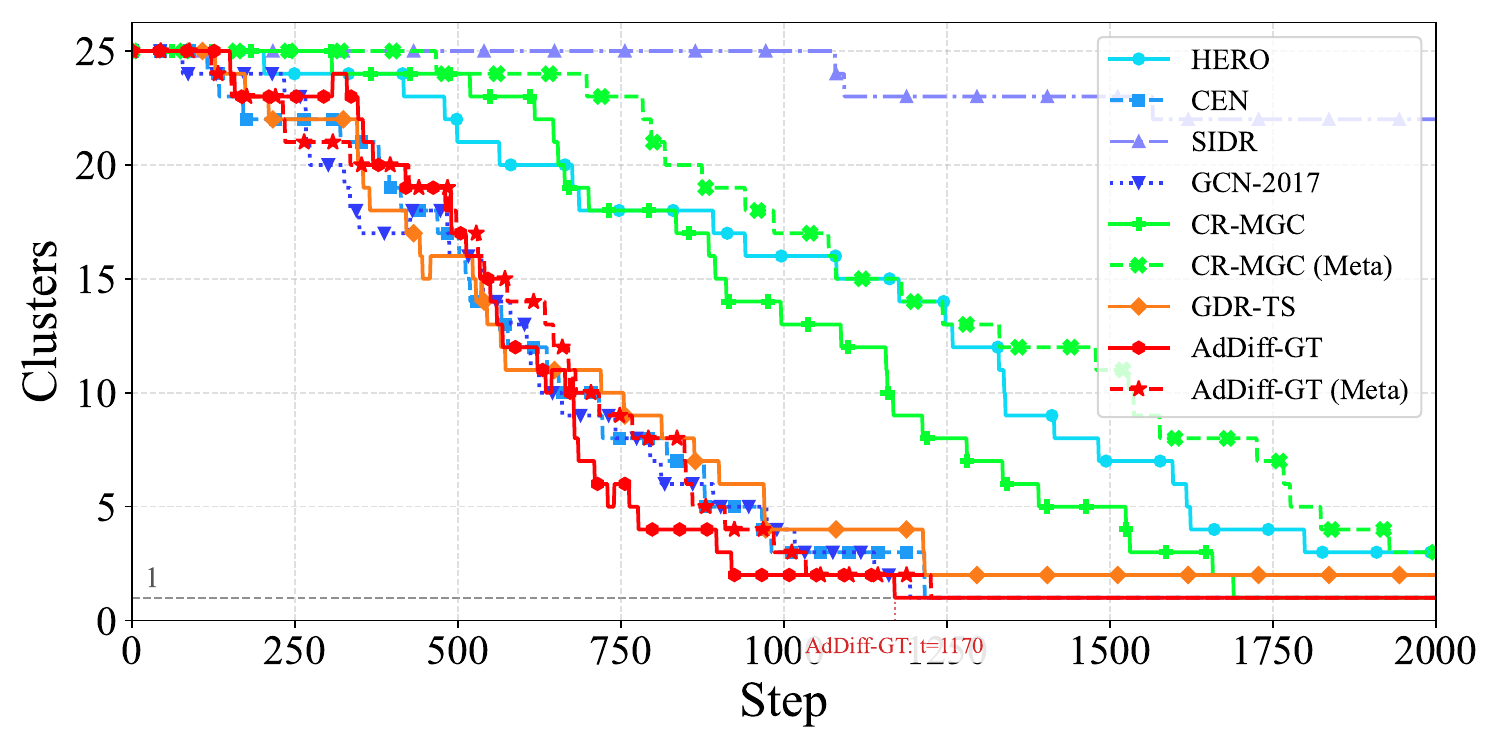}\\
\small (a) $n_t=28$
\end{minipage}
\hfill
\begin{minipage}[b]{0.28\textwidth}
\centering
\includegraphics[width=\linewidth]{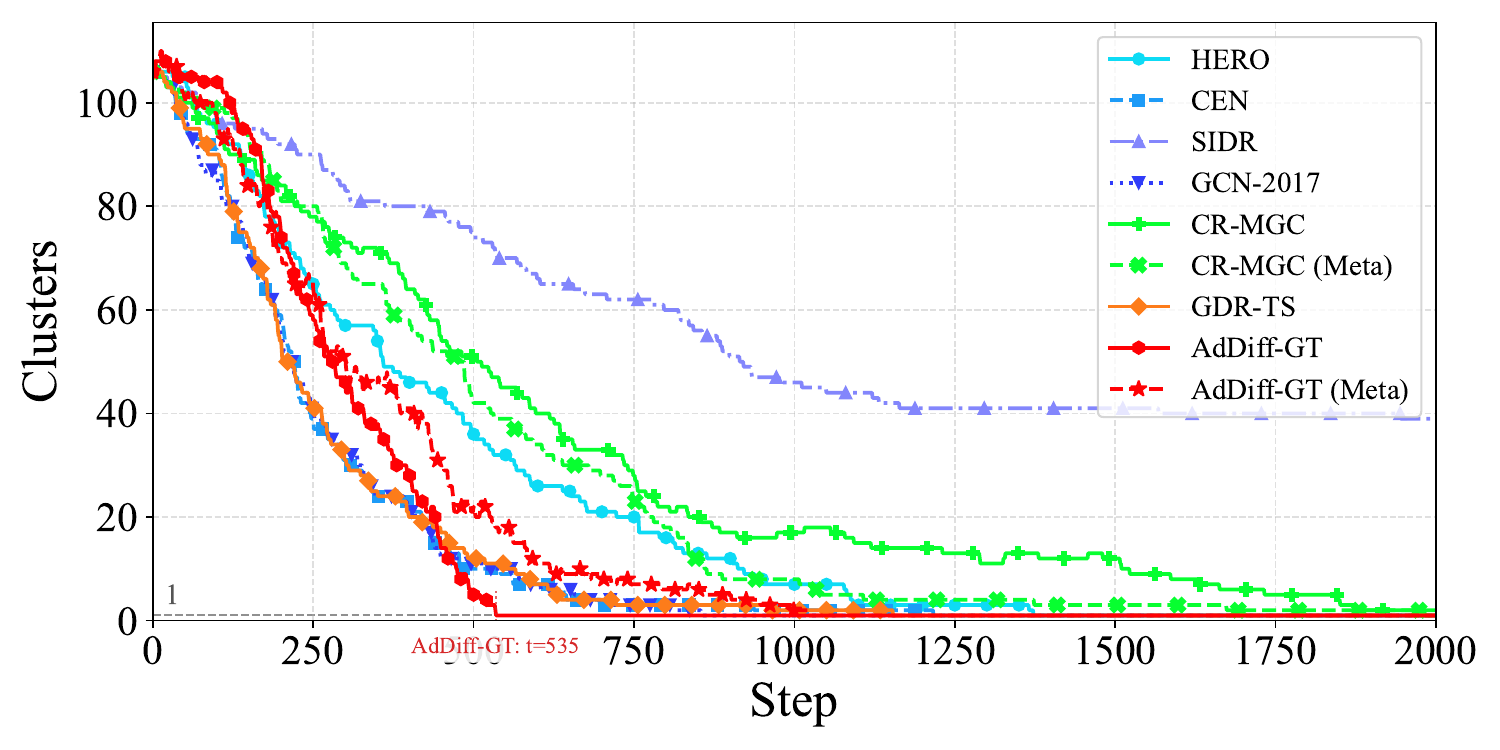}\\
\small (b) $n_t=193$
\end{minipage}
\hfill
\begin{minipage}[b]{0.28\textwidth}
\centering
\includegraphics[width=\linewidth]{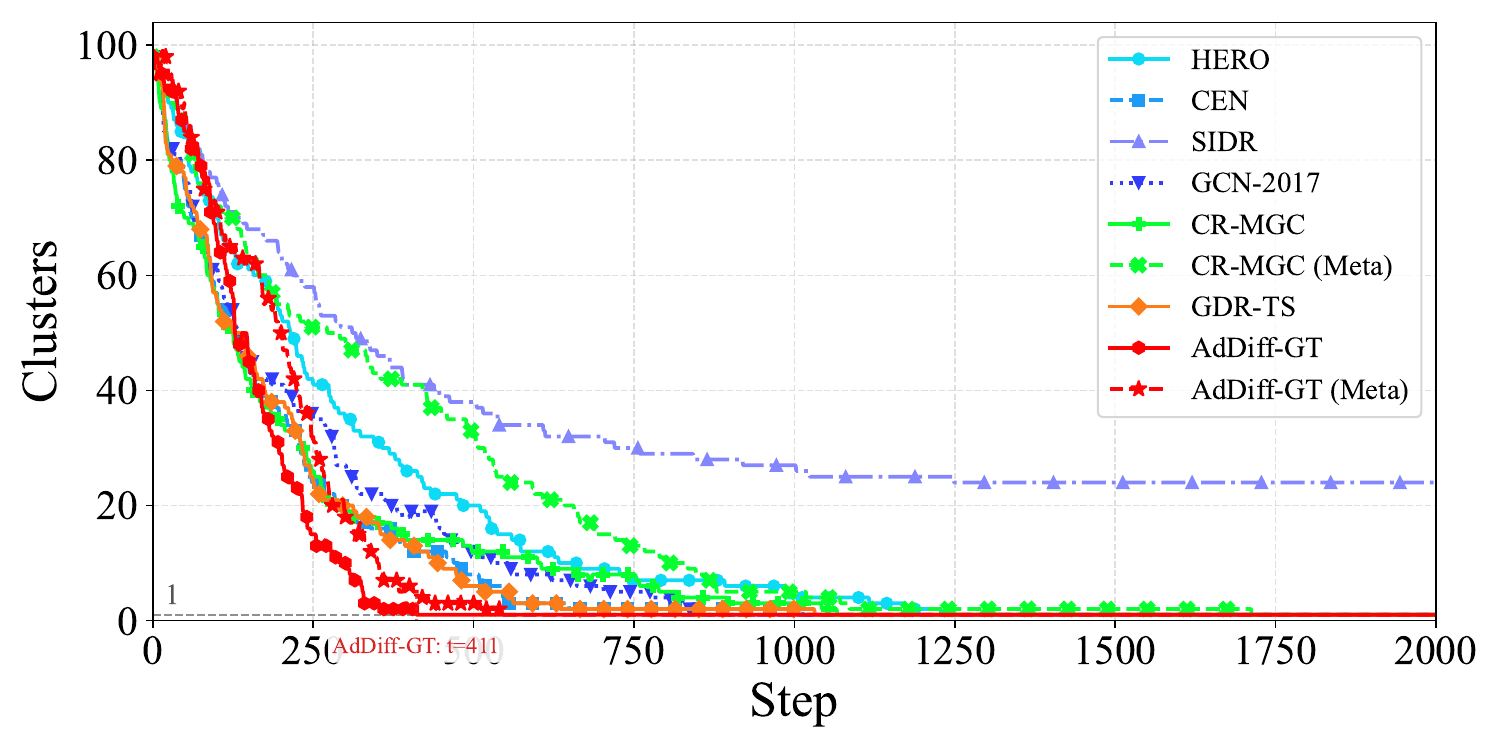}\\
\small (c) $n_t=292$
\end{minipage}
\caption{Representative convergence of the number of clusters under different residual swarm sizes. The dashed horizontal line indicates $C_t=1$, corresponding to a fully connected communication network.}
\label{fig:cluster_convergence_cases}
\end{figure*}
\subsubsection{Computational Efficiency and Meta-Adaptation}
\begin{figure*}[ht]
\centering
\begin{minipage}[b]{0.45\textwidth}
\centering
\includegraphics[width=0.9\linewidth]{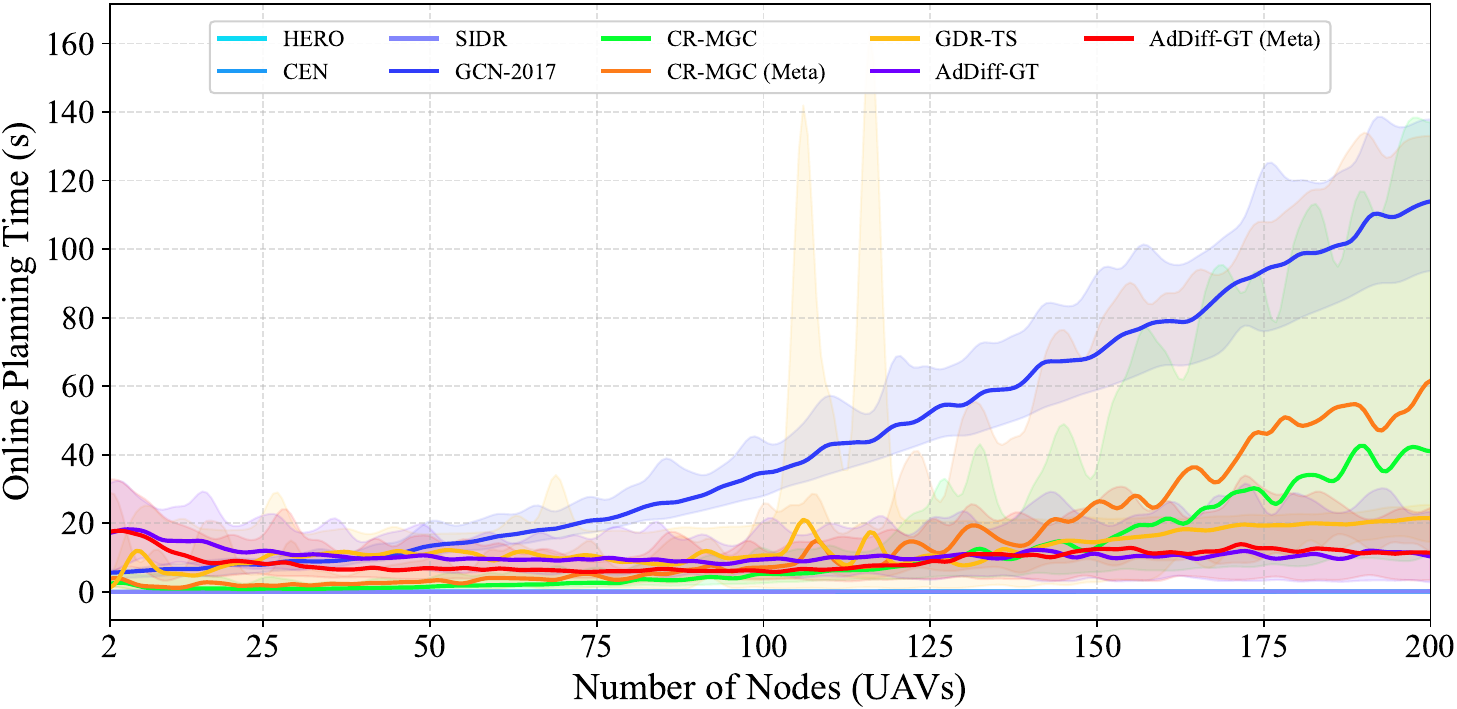}\\
\small (a) Setting I: $1000\times1000\times100 \, \mathrm{m}^3$, $N=200$
\end{minipage}
\hfill
\begin{minipage}[b]{0.45\textwidth}
\centering
\includegraphics[width=0.9\linewidth]{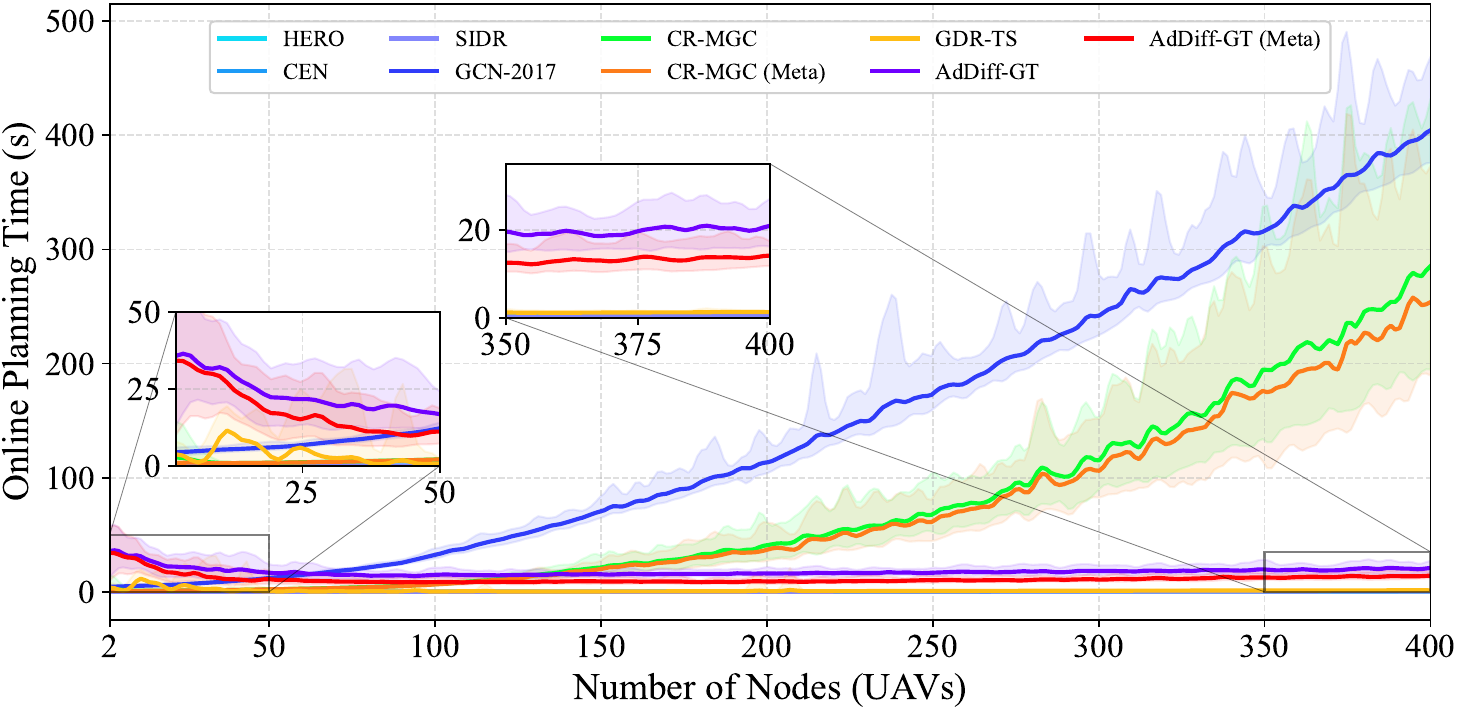}\\
\small (b) Setting II: $2000\times2000\times200 \, \mathrm{m}^3$, $N=400$
\end{minipage}
\caption{Online planning time  comparison across residual swarm sizes and experimental settings.}
\label{fig:inference_time}
\end{figure*}
Fig.~\ref{fig:inference_time} compares the online planning latency $T_{\mathrm{plan}}$ across surviving-node scales. GCN-2017 and CR-MGC exhibit rapidly increasing planning latency as the network size increases, whereas AdDiff-GT incurs additional computational overhead from hybrid graph construction, advection-diffusion propagation, and online parameter adaptation. Meta-initialization substantially reduces the average planning latency from 10.6942 s to 9.1844 s in Setting I and from 17.6835 s to 11.1862 s in Setting II, corresponding to a 36.7$\%$ reduction in the latter. Both variants maintain a 100$\%$ connectivity success rate in both settings. 
\begin{figure}[htp]
    \centering    \includegraphics[width=0.9\linewidth]{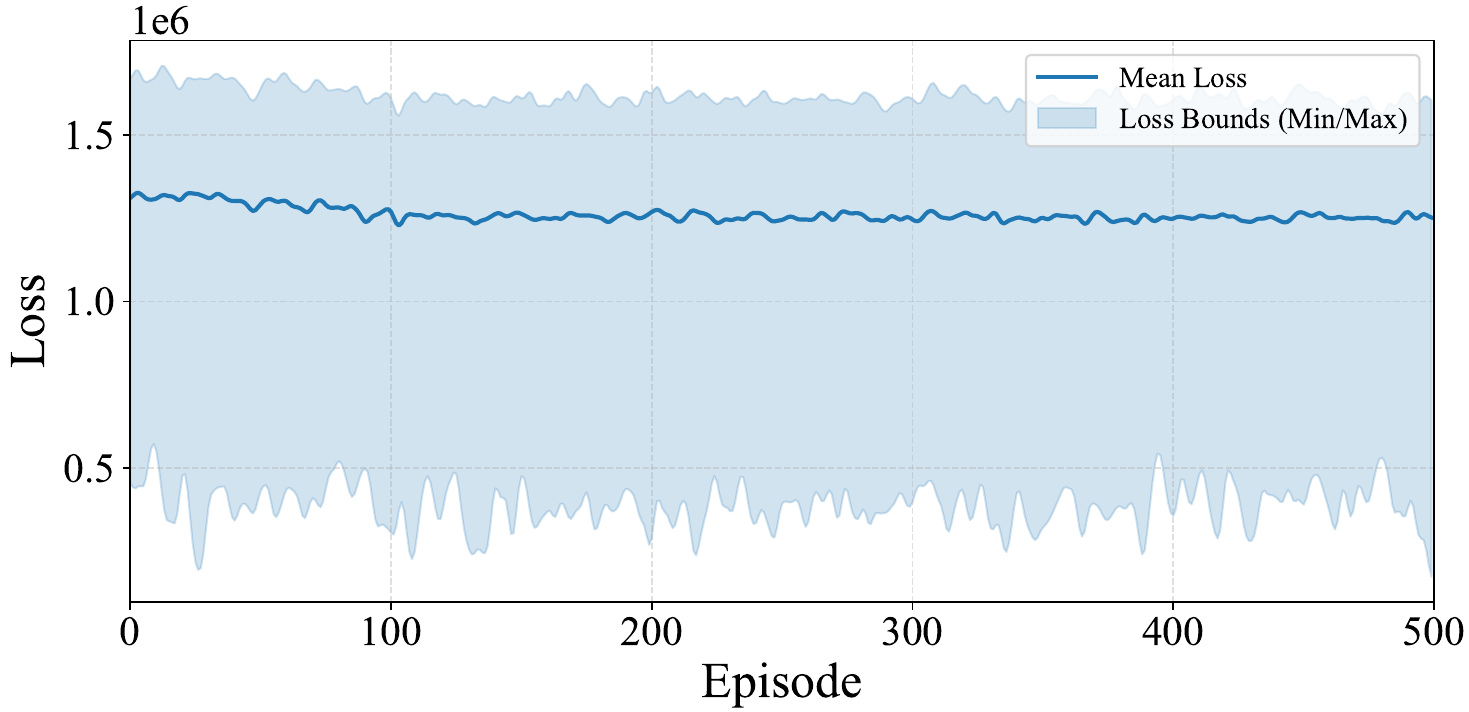}
    \vspace{-0.1in}
    \caption{Overall loss convergence during the offline meta-learning stage.}
    \label{fig:meta_learning_loss}
\end{figure}
\begin{figure}[htp]
    \centering
\includegraphics[width=0.9\linewidth]{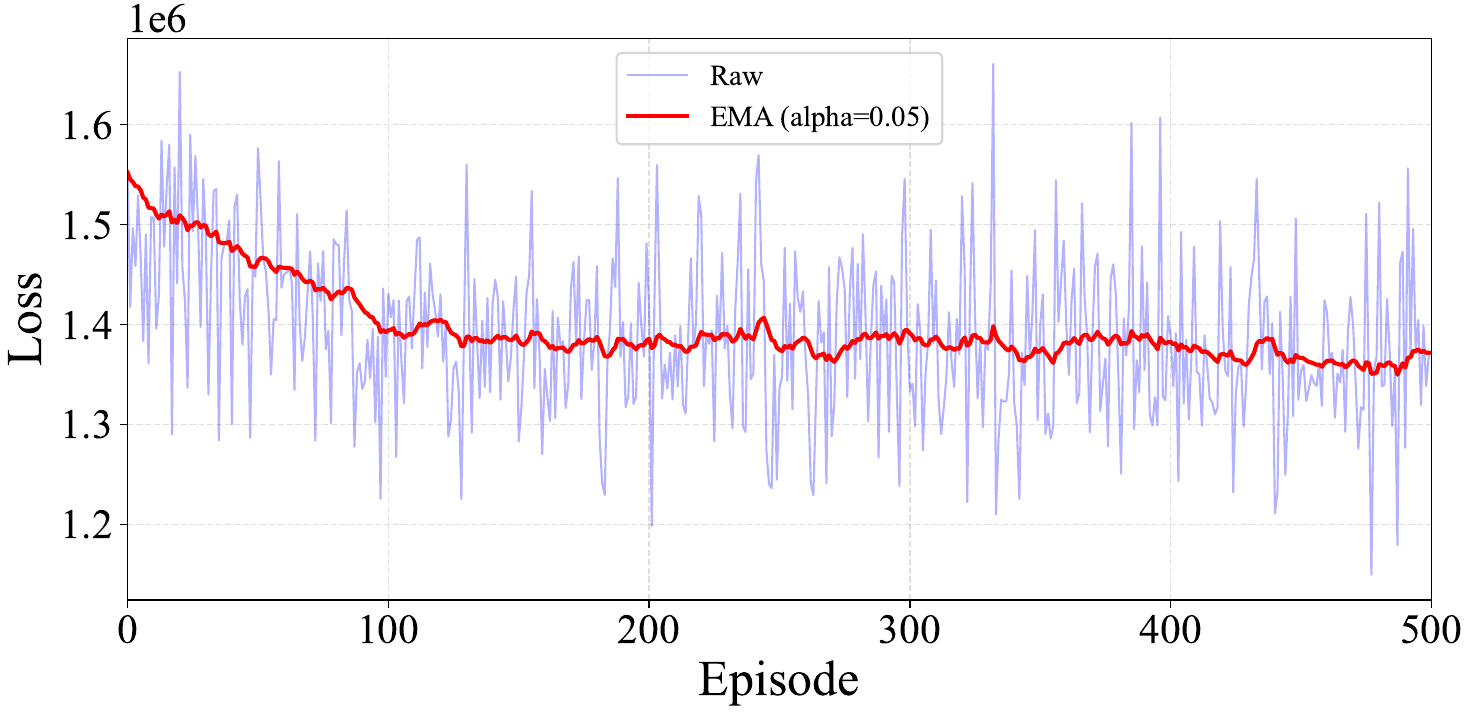}
    \vspace{-0.1in}
    \caption{Meta-learning loss convergence for a residual swarm size of $n=300$.}
    \label{fig:meta_learning_loss_300}
\end{figure}

Meta-learning also yields a favorable trade-off between planning efficiency and recovery performance. In Setting I, AdDiff-GT+Meta reduces the average recovery steps from 162.38 to 150.72. In Setting II, it requires slightly more steps (715.73 vs. 681.44) while maintaining 100$\%$ recovery and high spatial dispersion (0.9881). Figs.~\ref{fig:meta_learning_loss} and \ref{fig:meta_learning_loss_300} further show stable convergence of the offline meta-learning loss, supporting the effectiveness of scale-specific meta-initialization. 

\subsubsection{Communication Efficiency and Security}
\begin{figure*}[htbp]
\centering
\begin{minipage}[b]{0.23\textwidth}
\centering
\includegraphics[width=\linewidth]{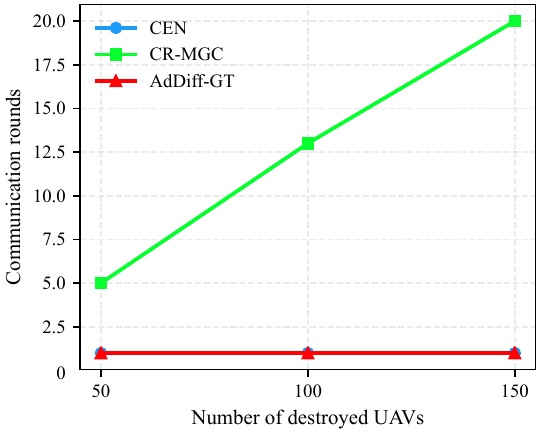}\\
\small (a) Interaction rounds
\end{minipage}\hfill
\begin{minipage}[b]{0.23\textwidth}
\centering
\includegraphics[width=\linewidth]{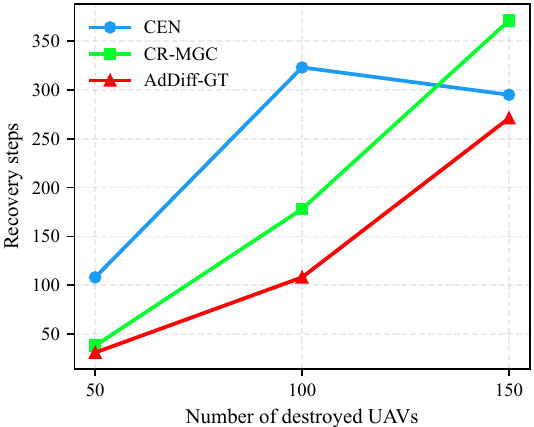}\\
\small (b) Recovery steps
\end{minipage}\hfill
\begin{minipage}[b]{0.23\textwidth}
\centering
\includegraphics[width=\linewidth]{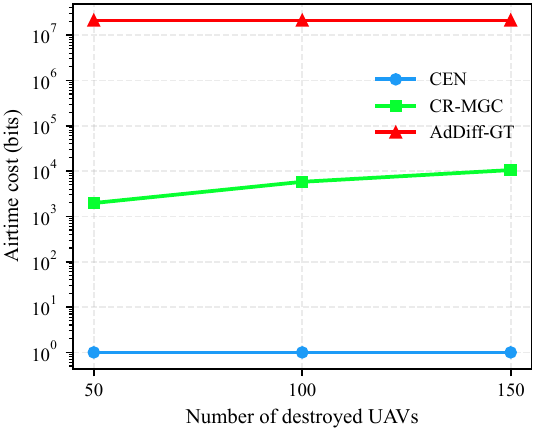}\\
\small (c) Airtime cost
\end{minipage}\hfill
\begin{minipage}[b]{0.23\textwidth}
\centering
\includegraphics[width=\linewidth]{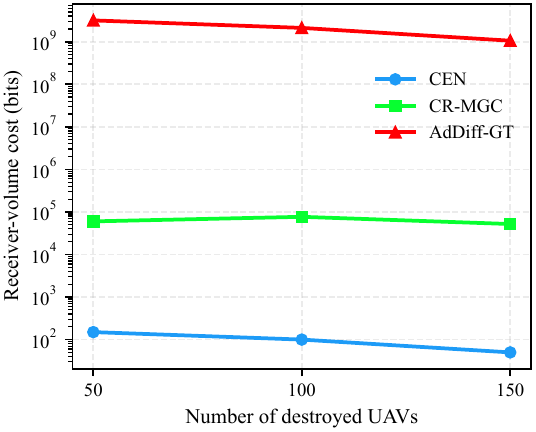}\\
\small (d) Receiver volume
\end{minipage}
\caption{Communication overhead across different damage scales.}
\label{fig:communication_cost}
\end{figure*}
Fig.~\ref{fig:communication_cost} compares the communication overhead of the proposed air-ground asymmetric architecture. As the number of damaged UAVs increases from 50 to 150, CR-MGC requires substantially more interaction rounds
due to repeated GCS-assisted optimization, whereas AdDiff-GT
maintains a single interaction round through local reporting, GCS tuning, and parameter broadcasting. Meanwhile, Fig.~\ref{fig:communication_cost}(b) shows that AdDiff-GT requires fewer recovery steps than CEN and CR-MGC, indicating that reduced interaction does not compromise recovery efficiency. Although AdDiff-GT incurs
higher wireless airtime cost and receiver-volume cost due to
one-shot model-parameter broadcasting, its communication
advantage lies in minimizing interaction rounds and avoiding
position-sensitive information exchange. This design shortens
the post-damage communication exposure window and improves
the security-efficiency trade-off under adversarial conditions.


\subsubsection{Spatial Distribution Preservation}
\begin{figure*}[!t]
\centering
\begin{minipage}[b]{0.45\textwidth}
\centering
\includegraphics[width=0.9\linewidth]{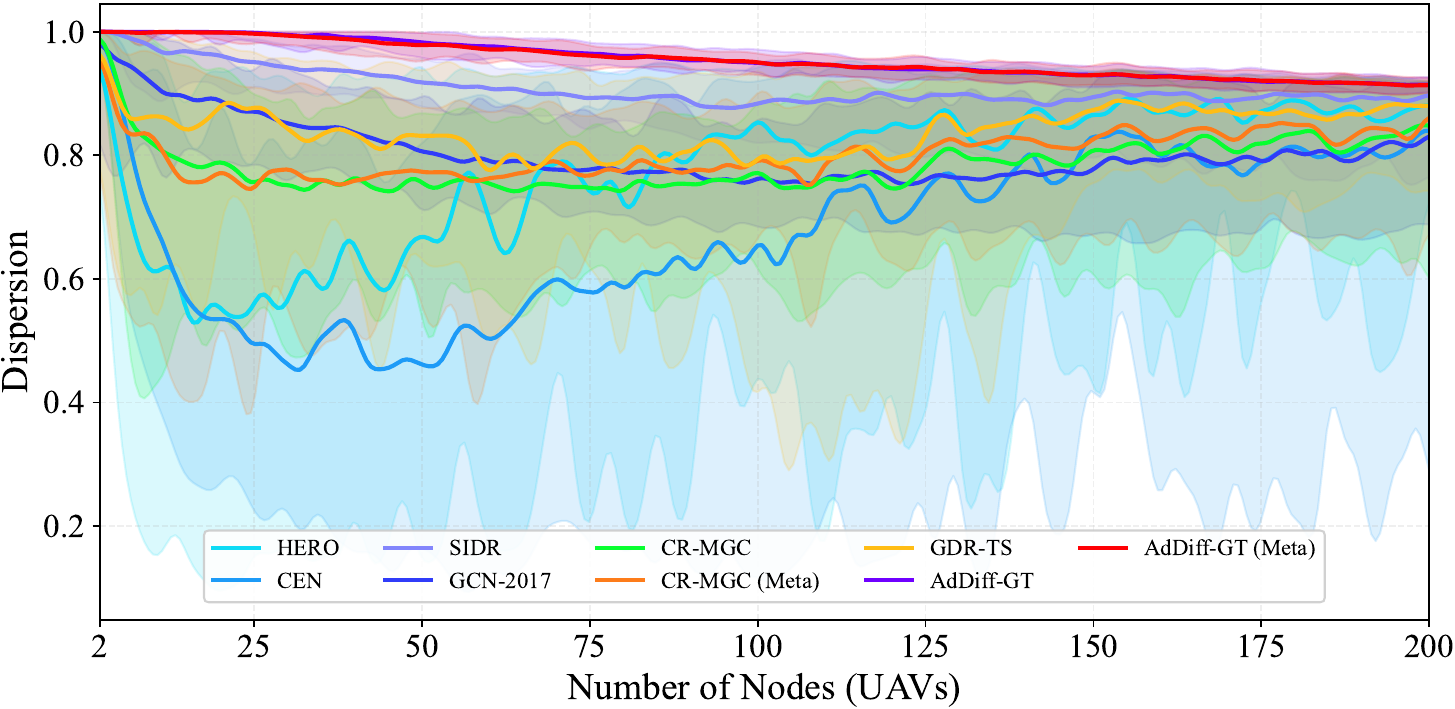}\\
\small (a) Setting I: $1000\times1000\times100 \, \mathrm{m}^3$, $N=200$
\end{minipage}
\hfill
\begin{minipage}[b]{0.45\textwidth}
\centering
\includegraphics[width=0.9\linewidth]{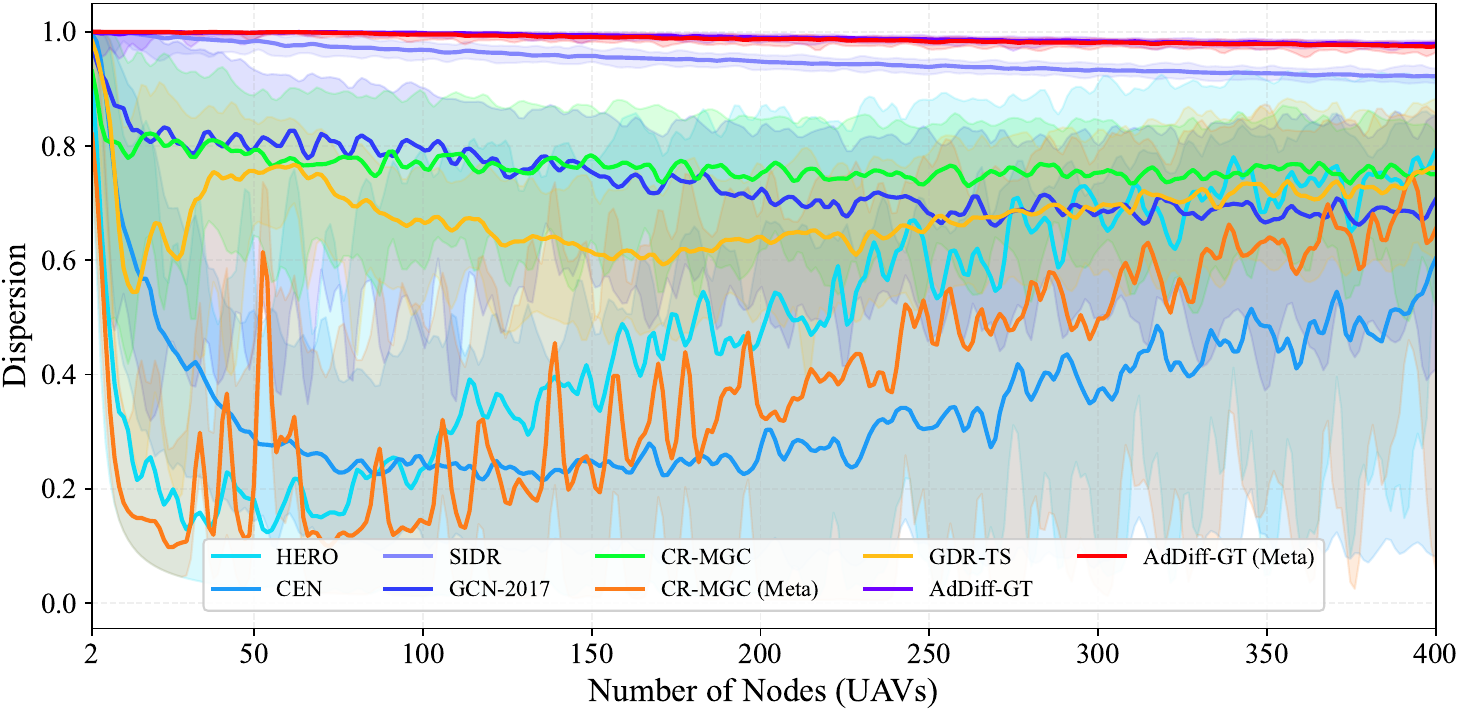}\\
\small (b) Setting II: $2000\times2000\times200 \, \mathrm{m}^3$, $N=400$
\end{minipage}
\caption{Spatial dispersion distributions across different residual swarm sizes and experimental settings.}
\label{fig:dispersion_distribution}
\end{figure*}
Fig.~\ref{fig:dispersion_distribution} compares the spatial dispersion of different methods. In Setting I, AdDiff-GT and AdDiff-GT+Meta achieve high average dispersions of 0.9548 and 0.9540, respectively, substantially outperforming
CEN, CR-MGC, and GDR-TS. AdDiff-GT attains the highest dispersion in 43.68$\%$ of scenarios, while AdDiff-GT+Meta achieves
the best recovery performance in 46.43$\%$ of scenarios. In Setting II, their dispersions further increase to
0.9900 and 0.9881, respectively, indicating effective preservation
of wide-area spatial distribution. In contrast, CEN achieves only
0.3493 due to excessive clustering, while GDR-TS reaches 0.6796.
Although SIDR exhibits high dispersion, its 5,000 failed scenarios
indicate that this dispersion mainly results from unreconnected
clusters. These results show that AdDiff-GT effectively balances connectivity
recovery and spatial coverage by jointly modeling UAV-level and
cluster-level topology.

\subsubsection{Visual Analysis of Self-Healing Trajectories}
\begin{figure*}[t]
\centering
\begin{minipage}[b]{0.30\textwidth}
\centering
\includegraphics[width=\linewidth]{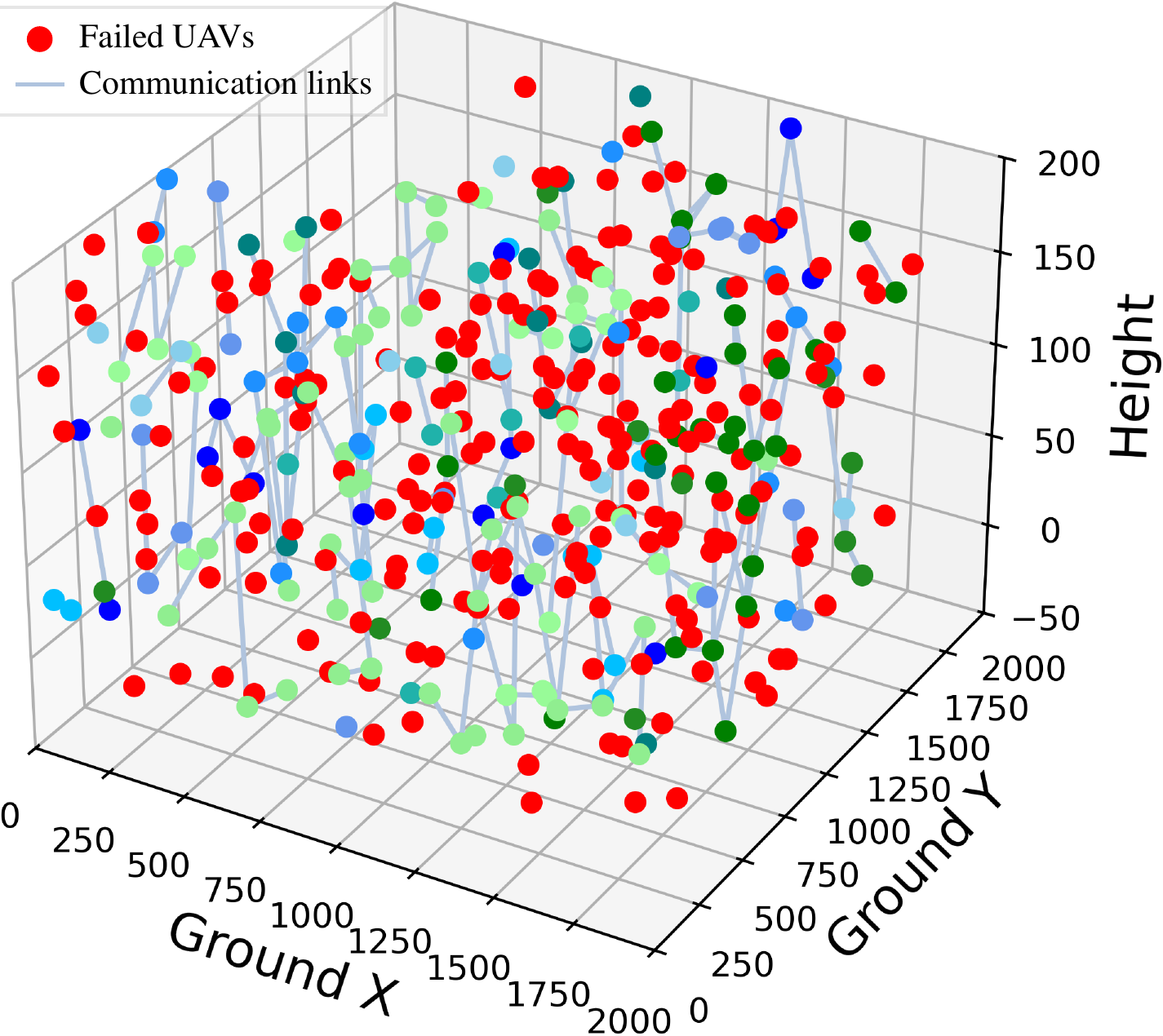}\\
\small (a) Initial residual topologies
\end{minipage}
\hfill
\begin{minipage}[b]{0.30\textwidth}
\centering
\includegraphics[width=\linewidth]{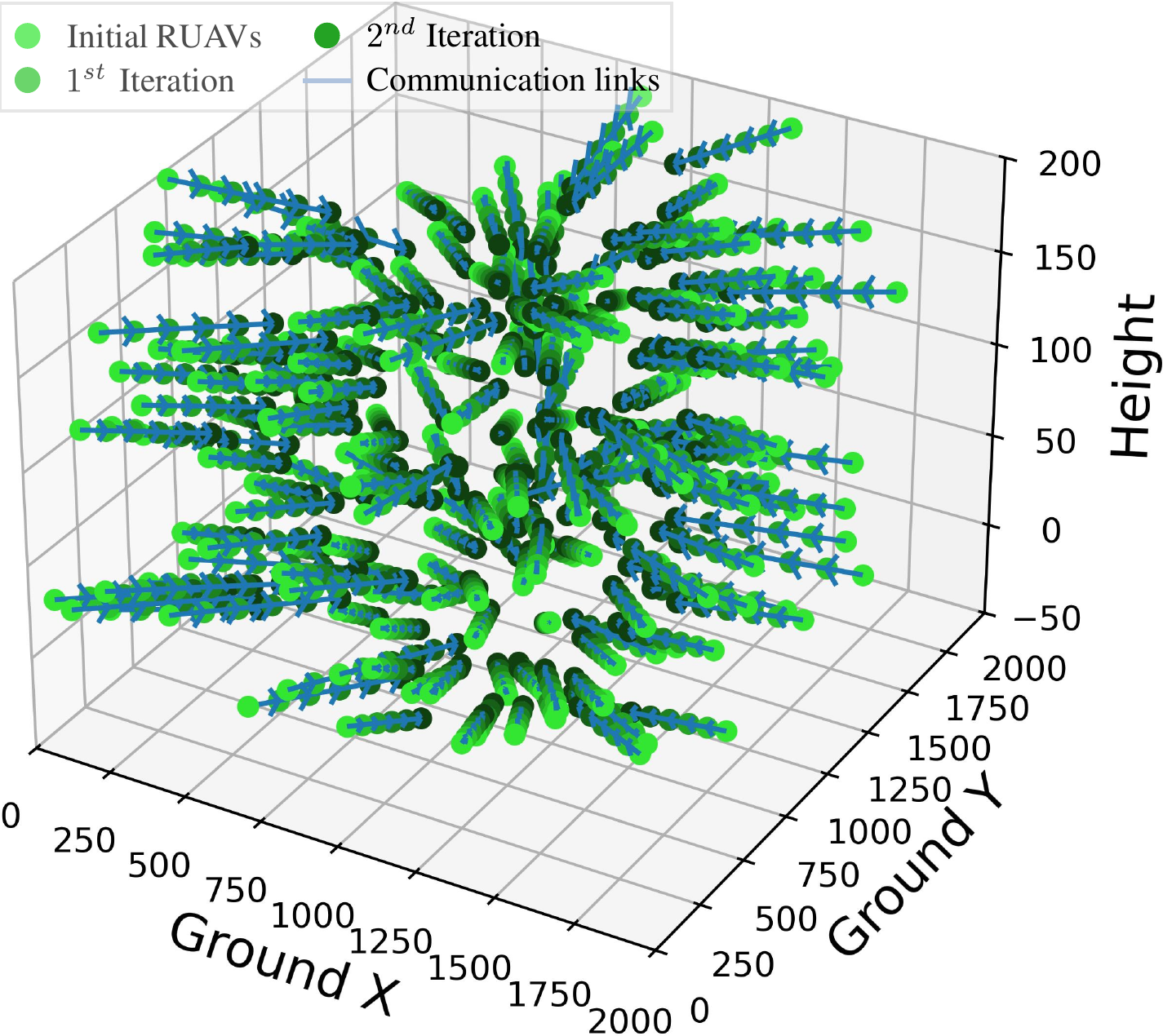}\\
\small (b) 3D recovery trajectories
\end{minipage}
\hfill
\begin{minipage}[b]{0.30\textwidth}
\centering
\includegraphics[width=\linewidth]{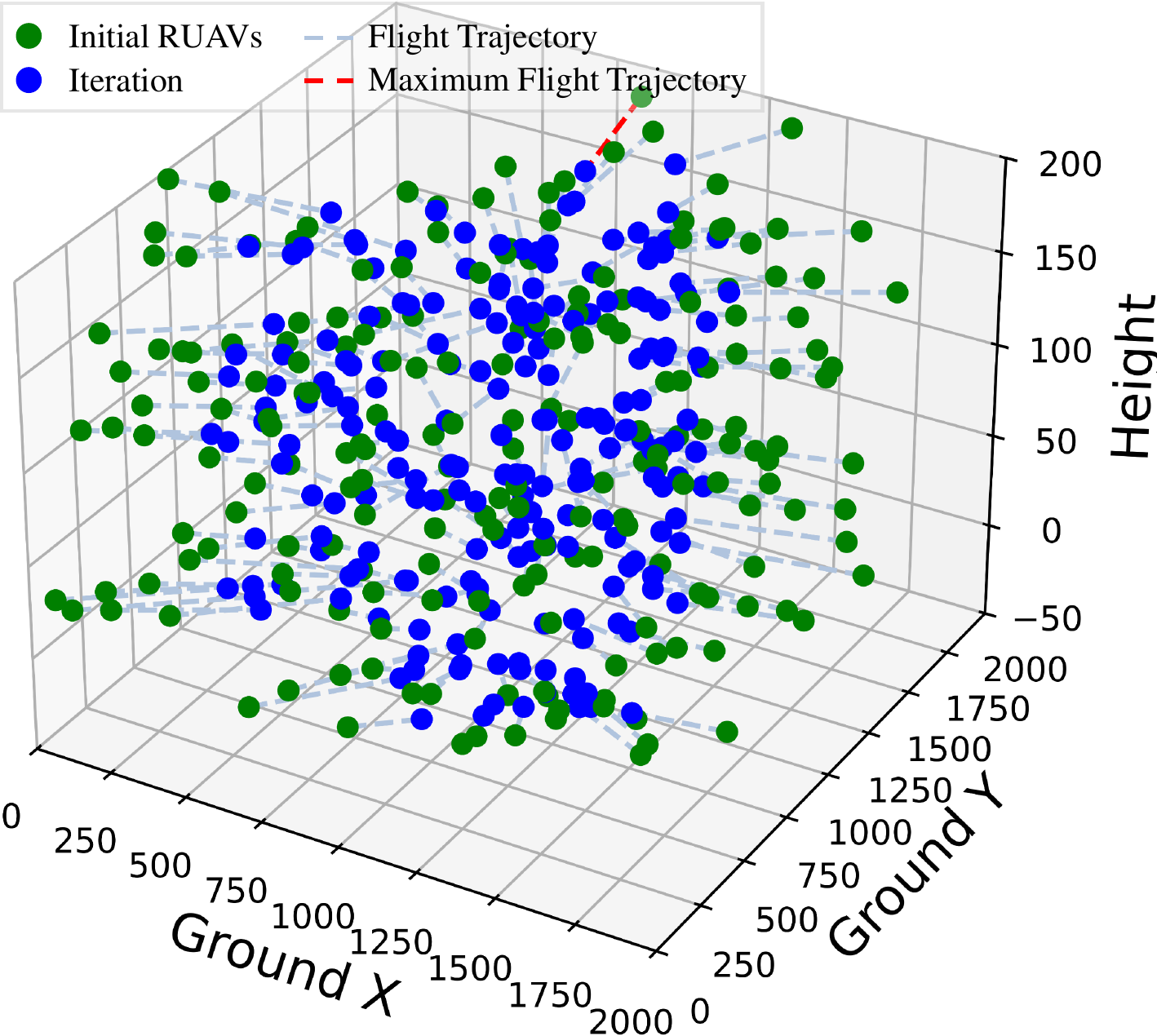}\\
\small (c) Final spatial distributions
\end{minipage}
\caption{Evolution of self-healing performance with 200 damaged nodes.}
\label{fig:cluster_groups}
\end{figure*}

Fig.~\ref{fig:cluster_groups} illustrates a representative self-healing trajectory with 200 damaged nodes in an initial network of $N=$ 
400 UAVs. The corresponding results for 100 and 300 damaged nodes are provided in Appendix~\ref{app:exp_details}.
As damage increases, the surviving topology becomes more fragmented with larger
inter-cluster gaps. AdDiff-GT generates asymmetric displacement
fields, where a small number of boundary or bridging nodes
undertake larger movements to establish cross-cluster links,
while most nodes undergo only local adjustments.  This mechanism
avoids excessive aggregation toward a geometric center and
preserves peripheral nodes and wide-area spatial distribution,
consistent with the high spatial dispersion values in Table~\ref{tab:experiment_results}. The visualization cases demonstrate how AdDiff-GT jointly achieves inter-cluster connectivity and spatial preservation through adaptive displacement fields.



\section{Conclusion}
\label{sec:conclusion}
This paper proposed AdDiff-GT, an integrated self-healing framework for USNET under severe network partitioning. An ID-level damage inference mechanism reconstructs post-damage topology from sparse local reports and topological closure, without requiring real-time position sharing. AdDiff-GT  integrates advection-diffusion dynamics with a node-supernode hybrid graph to capture intra- and inter-cluster relationships.  Meta-initialized air-ground asymmetric recovery further enables efficient online adaptation and distributed airborne inference.  Extensive experiments over 8,000 partitioned network scenarios, including 5,000 large-scale cases, demonstrate effective connectivity recovery with favorable spatial dispersion. Meta-initialization further reduces online planning overhead with only a minor recovery-step increase at large scales. Future work will consider channel uncertainties, heterogeneous UAV capabilities, and multi-round damage scenarios.

\bibliographystyle{IEEEtran}
\bibliography{Reference}

\appendices

\section{Notations and Operators}
\label{app:notation}
In this paper, scalars, vectors, and matrices are denoted by lowercase italic letters $x$, lowercase bold letters $\mathbf{x}$, and uppercase bold letters $\mathbf{X}$, respectively; $\mathbb{R}$ denotes the set of real numbers; $\|\cdot\|_2$ denotes the Euclidean or spectral norm; $|\mathcal{S}|$ denotes the cardinality of set $\mathcal{S}$; $\mathbb{I}\{\cdot\}$ denotes the indicator function; $\cup$, $\cap$, and $\setminus$ denote union, intersection, and set difference, respectively; $\mathbf{A}$, $\mathbf{D}$, and $\mathbf{L}$ denote the adjacency, degree, and Laplacian matrices, respectively; $\operatorname{diag}(\cdot)$ denotes the diagonal matrix operator, and $\odot$ denotes the element-wise (Hadamard) product; $\nabla$ and $\nabla^{\ast}$ denote the manifold gradient operator and its adjoint divergence operator, respectively, while $\nabla_{\mathcal{G}}$ and $\nabla_{\mathcal{G}}^{\ast}$ denote their discrete graph counterparts.

\section{AdDiff-GT-Based Self-Healing of USNET}
\label{app:process}
 The proposed AdDiff-GT self-healing process is illustrated in Fig.~\ref{fig:recovery_process}. The process consists of three stages: UED occurrence, local connectivity detection, and onboard forward inference and maneuver-based topology recovery.
\begin{figure*}[htp]
    \centering
    \includegraphics[width=0.9\linewidth]{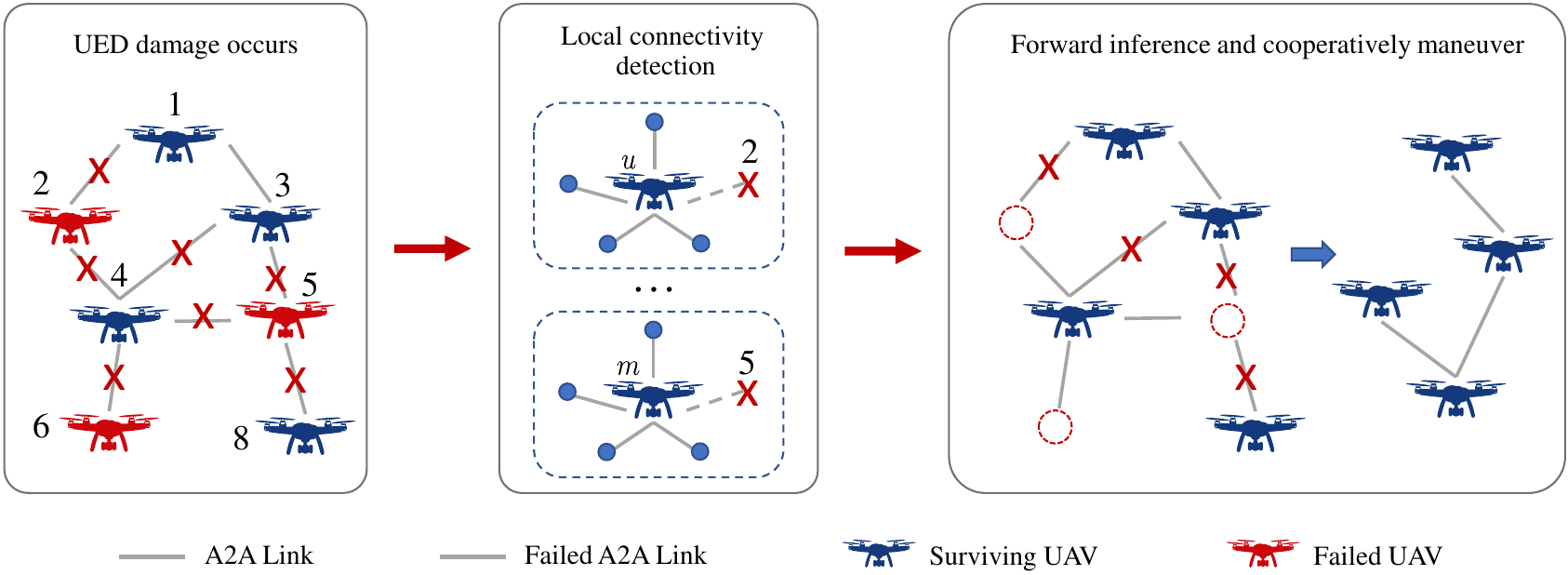}
    \caption{Overview of the proposed AdDiff-GT-based self-healing process for USNET.}
    \label{fig:recovery_process}
\end{figure*}

\section{Theoretical Analysis for AdDiff-GT}
\label{app:theory}
This appendix first provides the proofs of the propositions and the corollary stated in the main text, followed by additional propositions and their corresponding proofs.
\begin{proof}[Proof of Proposition~\ref{prop:convergence}]
Stacking the positions of all surviving nodes yields $\mathbf{P}\in\Omega^{|\mathcal{R}_t|}$. For the convergence analysis, we consider the corresponding position-space surrogate objective $\mathcal{L}(\mathbf{P})$. As defined in Eq.~\eqref{eq:loss}, $\mathcal{L}$ is a weighted combination of connectivity, displacement, maximum-displacement overflow, local smoothness, cluster alignment, inter-cluster attraction, soft-limit, and minimum-link-spacing terms. Each term is either a norm, squared norm, or ReLU-based penalty, or can be replaced by a smooth approximation such as soft-min or soft-plus during online tuning. Hence, there exists a constant $\mathcal{L}_{\min}\ge 0$ such that
\begin{equation}
\mathcal{L}(\mathbf{P})\ge\mathcal{L}_{\min},\quad \forall \mathbf{P}\in\Omega^{|\mathcal{R}_t|} .
\end{equation}

Consider the gradient-based update
\begin{equation} \label{gradient-based update}
\mathbf{P}_{k+1}=\mathbf{P}_{k}-\eta\nabla_{\mathbf{P}}\mathcal{L}(\mathbf{P}_k),
\end{equation}
where the step size is chosen such that the iterates remain within the feasible flight region. If necessary, Eq.~\eqref{gradient-based update} can be implemented using its projected form. 

Since $\nabla\mathcal{L}$ is $L$-Lipschitz continuous, the descent lemma gives
\begin{equation} \label{descent_lemma}
\begin{aligned}
\mathcal{L}(\mathbf{P}_{k+1})
\le\;&
\mathcal{L}(\mathbf{P}_k)
+\nabla\mathcal{L}(\mathbf{P}_k)^{\mathrm{T}}
(\mathbf{P}_{k+1}-\mathbf{P}_k)\\
&+\frac{L}{2}\|\mathbf{P}_{k+1}-\mathbf{P}_k\|_2^2 .
\end{aligned}
\end{equation}

Substituting Eq.~\eqref{gradient-based update} into Eq.~\eqref{descent_lemma} yields
\begin{equation}
\mathcal{L}(\mathbf{P}_{k+1})-\mathcal{L}(\mathbf{P}_k)
\le
-\eta\left(1-\frac{L\eta}{2}\right)
\|\nabla\mathcal{L}(\mathbf{P}_k)\|_2^2 .
\label{eq:appendix_descent}
\end{equation}

For $0<\eta<2/L$, define $c_{\eta}=\eta(1-L\eta/2)>0$. Then,   Eq.~\eqref{eq:appendix_descent} demonstrates that $\{\mathcal{L}(\mathbf{P}_k)\}$ is monotonically non-increasing. Since it is also lower-bounded by $L_{\min}$, the loss sequence converges. Summing Eq.~\eqref{eq:appendix_descent} from $k=0$ to $K$, we obtain
\begin{equation}
c_{\eta}\sum_{k=0}^{K}\|\nabla\mathcal{L}(\mathbf{P}_k)\|_2^2
\le
\mathcal{L}(\mathbf{P}_0)-\mathcal{L}(\mathbf{P}_{K+1})
\le
\mathcal{L}(\mathbf{P}_0)-\mathcal{L}_{\min}.
\end{equation}

Taking $K\to\infty$ yields
\begin{equation}
\sum_{k=0}^{\infty}\|\nabla\mathcal{L}(\mathbf{P}_k)\|_2^2<\infty,
\end{equation}
and therefore $\|\nabla\mathcal{L}(\mathbf{P}_k)\|_2\to0$. 

Because $\Omega^{|\mathcal{R}_t|}$ is bounded and closed, the sequence $\{\mathbf{P}_k\}$ has cluster points. For any cluster point $\mathbf{P}^{\star}$, continuity of $\nabla\mathcal{L}$ gives $\nabla\mathcal{L}(\mathbf{P}^{\star})=0$.
Thus, every cluster point of the online tuning sequence is a first-order stationary point of the surrogate objective. This establishes convergence of the proposed online tuning process to a stationary recovery state under the stated smoothness and step-size conditions. 
\end{proof}

\begin{proof}[Proof of Proposition~\ref{prop:boundedness}]
We first establish a perturbation bound for the matrix exponential. Let $\mathbf{X},\mathbf{E}\in \mathbb{R}^{n\times n}$, and let $\|\cdot\|$ be a submultiplicative matrix norm. Suppose there exist constants $M\ge1$ and $\omega\ge0$ such that $\|\mathrm{e}^{\mathbf{Y}}\|\le M \mathrm{e}^{\omega\|\mathbf{Y}\|}$. Define  $\mathbf{X}(s)=\mathbf{X}+s\mathbf{E}, s\in [0,1]$. Then, 
\begin{equation}
\mathrm{e}^{\mathbf{X}+\mathbf{E}}-\mathrm{e}^{\mathbf{X}}
=\int_0^1\frac{d}{ds}\mathrm{e}^{\mathbf{X}(s)}ds,
\end{equation}
and the derivative of the matrix exponential is 
\begin{equation}
\frac{d}{ds}\mathrm{e}^{\mathbf{X}(s)}
=\int_0^1\mathrm{e}^{(1-\theta)\mathbf{X}(s)}
\mathbf{E}
\mathrm{e}^{\theta\mathbf{X}(s)}d\theta.
\end{equation}

Using submultiplicativity and the assumed exponential growth bound gives
\begin{equation}
\begin{aligned}
\|\mathrm{e}^{\mathbf{X}+\mathbf{E}}-\mathrm{e}^{\mathbf{X}}\|
&\le
\int_0^1\int_0^1
\|\mathrm{e}^{(1-\theta)\mathbf{X}(s)}\|
\|\mathbf{E}\|
\|\mathrm{e}^{\theta\mathbf{X}(s)}\|d\theta ds\\
&\le
\|\mathbf{E}\|M^2
\int_0^1\int_0^1
\mathrm{e}^{\omega\|\mathbf{X}(s)\|}d\theta ds\\
&\le
\|\mathbf{E}\|M^2\mathrm{e}^{\omega(\|\mathbf{X}\|+\|\mathbf{E}\|)} .
\end{aligned}
\label{eq:appendix_exp_bound}
\end{equation}

Consider the continuous-form advection-diffusion propagation
\begin{equation}
\frac{\partial\mathbf{Z}}{\partial t}
=-(\mathbf{I}-\beta\mathbf{V}-\mathbf{C})\mathbf{Z},
\end{equation}
where $\mathbf{V}$ is the normalized local adjacency operator and $\mathbf{C}$ is the global attention coupling matrix. For fixed $\mathbf{V}$ and $\mathbf{C}$, its 
closed-form solution is
\begin{equation}
\mathbf{Z}(T_{\mathrm{prop}})=\mathrm{e}^{-\mathbf{L}T_{\mathrm{prop}}}\mathbf{Z}(0),\quad
\mathbf{L}=\mathbf{I}-\beta\mathbf{V}-\mathbf{C}.
\end{equation}

Under a topological perturbation, $\mathbf{V}'=\mathbf{V}+\Delta\mathbf{V}$ and $\mathbf{C}'=\mathbf{C}+\Delta\mathbf{C}$, and therefore
\begin{equation}
\mathbf{L}'-\mathbf{L}
=-\beta\Delta\mathbf{V}-\Delta\mathbf{C}.
\end{equation}

Applying Eq.~\eqref{eq:appendix_exp_bound} with  $\mathbf{X}=-\mathbf{L}T_{\mathrm{prop}}$ and $\mathbf{E}=-(\mathbf{L}'-\mathbf{L})T_{\mathrm{prop}}$, we obtain 
\begin{equation}
\begin{aligned}
&\left\|\mathrm{e}^{-\mathbf{L}'T_{\mathrm{prop}}}-\mathrm{e}^{-\mathbf{L}T_{\mathrm{prop}}}\right\|_2 \\
&\le T_{\mathrm{prop}}\|\mathbf{L}'-\mathbf{L}\|_2M^2
\mathrm{e}^{\omega T_{\mathrm{prop}}(\|\mathbf{L}\|_2+\|\mathbf{L}'-\mathbf{L}\|_2)}.
\end{aligned}
\label{eq:appendix_L_bound}
\end{equation}

For the purely local propagation scheme without attention compensation, $\Delta\mathbf{C}=0$, and hence $\|\mathbf{L}'-\mathbf{L}\|_2=\beta\|\Delta\mathbf{V}\|_2$. Accordingly, the exponential factor in Eq.~\eqref{eq:appendix_L_bound} contains a term that grows with 
${\omega T_{\mathrm{prop}}}\beta\|\Delta\mathbf{V}\|_2$, indicating that local propagation may amplify topological perturbations as the propagation time increases. By contrast, AdDiff-GT introduces the global coupling matrix $\mathbf{C}$, which provides additional cross-cluster interactions. Under the compensation condition in Eq.~\eqref{compensation condition}, 
\begin{equation}
\|\mathbf{L}'-\mathbf{L}\|_2=
\|\beta\Delta\mathbf{V}+\Delta\mathbf{C}\|_2
\le a_1\|\Delta\mathbf{V}\|_2
+a_2\|\Delta\mathbf{V}\|_2^2 .
\label{eq:attention_comp}
\end{equation}


Because the trained input domain, propagation time, and spectral quantities are bounded, the exponential factor in Eq.~\eqref{eq:appendix_L_bound} can be upper-bounded by a finite constant $B_E$. Thus,
\begin{equation}
\|\mathrm{e}^{-\mathbf{L}'T_{\mathrm{prop}}}-\mathrm{e}^{-{\mathbf{L}T_{\mathrm{prop}}}}\|_2
\le
B_E T_{\mathrm{prop}}
\left(a_1\|\Delta\mathbf{V}\|_2
+a_2\|\Delta\mathbf{V}\|_2^2\right).
\end{equation}

Finally, the encoder, residual transformations, normalization layers, activation functions, and decoder are Lipschitz continuous over the bounded operating domain. Let their composite Lipschitz constant be $L_{\Phi}$. The resulting output representation perturbation therefore satisfies
\begin{equation}
\|\Delta\mathbf{Z}_{\mathrm{out}}\|_2
\le
K_1\|\Delta\mathbf{V}\|_2
+K_2\|\Delta\mathbf{V}\|_2^2 ,
\end{equation}
where $K_1=L_{\Phi}B_ET_{\mathrm{prop}}a_1$ and $K_2=L_{\Phi}B_ET_{\mathrm{prop}}a_2$. This proves that, under the bounded spectral radius, finite propagation time, and attention compensation assumptions, the output representation perturbation of AdDiff-GT admits a first- and second-order bound with respect to the adjacency perturbation.
\end{proof}


\begin{proof}[Proof of Corollary~\ref{cor:finite_step}]
Let $C(t)$ denote the number of connected clusters in the surviving graph $\mathcal{G}'_t$ at time $t$.  The initial post-damage topology contains $C(0)=m$ clusters, where $m\le|\mathcal{R}_t|$ is finite. 

For any two separated clusters $\mathcal{C}_a$ and $\mathcal{C}_b$, define their shortest inter-cluster distance as
\begin{equation}
d_{ab}(t)=\min_{i\in\mathcal{C}_a,j\in\mathcal{C}_b}
\|\mathbf{p}_{i,t}-\mathbf{p}_{j,t}\|_2.
\end{equation}
When $d_{ab}(t)>R_c$, no physical A2A link exists between the two clusters. During each effective recovery phase, assume that at least one separated cluster pair $(a,b)$ satisfying
\begin{equation}
d_{ab}(t+1)-d_{ab}(t)\le-\delta, \delta>0,
\label{eq:delta_merge}
\end{equation}
until the pair becomes connected. Let $t_0$ denote the beginning of the corresponding merging phase. Recursively applying Eq.~\eqref{eq:delta_merge}, we obtain
\begin{equation}
d_{ab}(t_0+\ell)
\le d_{ab}(t_0)-\ell\delta.
\end{equation}

Hence, for
\begin{equation}
\ell\ge
\left\lceil\frac{d_{ab}(t_0)-R_c}{\delta}\right\rceil
\end{equation}
we have $d_{ab}(t_0+\ell)\le R_c$. Therefore, at least one UAV pair from $\mathcal{C}_a$ and  $\mathcal{C}_b$ satisfies the communication radius constraint, establishing a physical A2A link and merging the two clusters. Consequently, $C(t)$ decreases by at least one after each effective merge. 

The candidate solution update strategy prioritizes solutions that reduce the number of connected clusters when $C(T_c)>1$, while rejecting solutions that destroy connectivity once $C(T_c)=1$. Under the stated effective merging condition, $C(t)$ is a non-increasing finite integer sequence. Starting from $m$ clusters, at most $m-1$ effective merges are required to obtain a connected topology.

Let $d_r(0)$ denote the shortest inter-cluster distance at the beginning of the $r$-th effective merging phase. The duration of this phase is bounded by  
\begin{equation}
\ell\le
\left\lceil\frac{d_r(0)-R_c}{\delta}\right\rceil.
\end{equation}
Thus, the total number of recovery steps satisfies
\begin{equation}
T_c
\le
\sum_{r=1}^{m-1}
\left\lceil
\frac{d_r(0)-R_c}{\delta}
\right\rceil
<\infty.
\end{equation}
Therefore, the proposed recovery process achieves finite-step connectivity recovery.
\end{proof}


\begin{proposition}[Completeness of Damaged-Node Inference]
If the historical topology $\mathcal{G}_t$ is connected,
every boundary-surviving UAV adjacent to a failed UAV in  $\mathcal{R}_{\mathrm{rep}}$ reliably generates a report, and at least one surviving UAV maintains a GCS backhaul link, then Algorithm~\ref{alg:damage_inference} exactly recovers the failed-node set,  i.e., $\mathcal{D}_{\mathrm{inf}}=\mathcal{V}_{\mathrm{fail}}$.
\end{proposition}

\begin{proof}

First, we establish the boundary shielding property. Define the set of surviving UAVs capable of maintaining a GCS backhaul link as 
\begin{equation}
\mathcal{R}_{\mathrm{GCS}} = \{ u \in \mathcal{R}_t \mid u \text{ maintains a GCS backhaul link} \},
\end{equation}
and the reporting set as
\begin{equation}
\mathcal{R}_{\mathrm{rep}}
=\{u\in \mathcal{R}_t\mid \exists v\in\mathcal{V}_{\mathrm{fail}},(u,v)\in \mathcal{E}_t\}.
\end{equation}

By the reporting assumption, any surviving UAV in $\mathcal R_{\mathrm{GCS}}$ adjacent to a failed UAV triggers an event report. Therefore, if $x\in\mathcal{V}_{\mathrm{fail}}$ and $y\in \mathcal{R}_{\mathrm{GCS}} \setminus\mathcal{R}_{\mathrm{rep}}$, there cannot exist an edge $(x,y)\in \mathcal{E}_t$. Otherwise, $y$ would belong to $\mathcal{R}_{\mathrm{rep}}$ due to detecting the link disconnection with neighbor $x$, which contradicts $y\notin\mathcal{R}_{\mathrm{rep}}$. Thus, $\mathcal{R}_{\mathrm{rep}}$ forms the reported boundary between failed nodes and GCS-connected surviving nodes. 

Second, we prove that $\mathcal{D}_{\mathrm{inf}}\subseteq\mathcal{V}_{\mathrm{fail}}$. The initial set is 
\begin{equation}
\mathcal{D}_{\mathrm{inf}}^{(0)}
=\mathcal{D}_{\mathrm{rep}}
=\bigcup_{u\in\mathcal{R}_{\mathrm{rep}}}\mathcal{D}_{u}^{\mathrm{rep}},
\end{equation}
where each reported neighbor  belongs to $\mathcal{V}_{\mathrm{fail}}$. Assume that $\mathcal{D}_{\mathrm{inf}}^{(k)}\subseteq\mathcal{V}_{\mathrm{fail}}$ after the $k$-th BFS step. At  step $k+1$, the algorithm examines the historical neighbors $n$ of nodes $d\in\mathcal{D}_{\mathrm{inf}}^{(k)}$. The algorithm adds $n$ to the inferred set if and only if $n\notin\mathcal{D}_{\mathrm{inf}}^{(k)}$ and $n\notin\mathcal{R}_{\mathrm{rep}}$. If $n\in \mathcal{R}_t$, then  $d\in\mathcal{V}_{\mathrm{fail}}$ and $(d,n)\in \mathcal{E}_t$. By the boundary shielding property,  $n\in\mathcal{R}_{\mathrm{rep}}$, contradicting the expansion condition. Hence $n\in\mathcal{V}_{\mathrm{fail}}$. By induction, it follows that $\mathcal{D}_{\mathrm{inf}}\subseteq\mathcal{V}_{\mathrm{fail}}$.

Finally, we prove that $\mathcal{V}_{\mathrm{fail}}\subseteq\mathcal{D}_{\mathrm{inf}}$. 
For any $x\in\mathcal{V}_{\mathrm{fail}}$, since the historical topology $\mathcal{G}_t$ is connected and $\mathcal{R}_{\mathrm{GCS}}\neq \varnothing$, there must exist a shortest path from $x$ to $\mathcal{R}_{\mathrm{GCS}}$:  
\begin{equation}
P=(x=v_0,v_1,\ldots,v_{m-1},v_m=y),
\end{equation}
where $y\in\mathcal{R}_{\mathrm{GCS}}$ and $v_0,\ldots,v_{m-1}\in\mathcal{V}_{\mathrm{fail}}$.
Thus, $y$ is adjacent to the failed node $v_{m-1}$, so $y\in\mathcal{R}_{\mathrm{rep}}$ and $v_{m-1}\in\mathcal{D}_{\mathrm{rep}}$. Starting from $v_{m-1}$, the BFS expansion inversely traverses $v_{m-2},\ldots,v_0$. Since these nodes are failed and do not belong to $\mathcal{R}_{\mathrm{rep}}$, they satisfy the expansion condition. Here $x=v_0$ is eventually added to $\mathcal{D}_{\mathrm{inf}}$. Therefore, $\mathcal{V}_{\mathrm{fail}}\subseteq\mathcal{D}_{\mathrm{inf}}$. 

Combining the two inclusion relations gives $\mathcal{D}_{\mathrm{inf}}=\mathcal{V}_{\mathrm{fail}}$.

\end{proof}

This property shows that, under reliable boundary reporting, the GCS can reconstruct the post-damage topology from sparse event-triggered observations without global heartbeat polling or real-time position reporting. 

\begin{proposition}[Correctness and Complexity of BFS Connectivity Detection]
For an undirected RUAV graph $\mathcal{G}'_t=(\mathcal{R}_t,\mathcal{E}'_t)$ with  $\mathbf{L}_t=\mathbf{D}_t-\mathbf{A}_t$, Algorithm~\ref{alg:bfs_connectivity} correctly identifies all connected clusters by treating $[\mathbf{L}_t]_{uv}<0$ as the adjacency. Consequently, the returned cluster count $C_t$ is exact, the RUAV graph is connected if and only if $C_t=1$, and the algorithm has complexity  $\mathcal{O}(|\mathcal{R}_t|+|\mathcal{E}'_t|)$.    
\end{proposition}

\begin{proof}
Let $\mathbf{A}_t$ denote the adjacency matrix of $\mathcal{G}'_t=(\mathcal{R}_t,\mathcal{E}'_t)$, and let $\mathbf{L}_t=\mathbf{D}_t-\mathbf{A}_t$ denote its Laplacian matrix. For an undirected graph, we have 
\begin{equation}
[\mathbf{L}_t]_{uv}=
\begin{cases}
d_{u,t}, & u=v,\\
-a_{uv,t}, & u\ne v .
\end{cases}
\end{equation}
Hence, for $u\ne v$, $[\mathbf{L}_t]_{uv}<0$ if and only if $a_{uv,t}>0$, meaning that $(u,v)\in\mathcal{E}'_t$. This demonstrates that the neighbor set recovered by Algorithm~\ref{alg:bfs_connectivity} from the negative off-diagonal elements of $\mathbf{L}_t$ is identical to the true edge set.

Consider a BFS initiated at an unvisited node $s$, and let $\mathcal{B}(s)$ denote the visited set. Furthermore, for any node $v$ in the same connected cluster as $s$, there exists a path
\begin{equation}
s=v_0,v_1,\ldots,v_q=v,\quad (v_{\ell},v_{\ell+1})\in\mathcal{E}'_t.
\end{equation}
BFS traverses this path successively, so $v$ is eventually visited. Thus, $\mathcal{B}(s)$ is exactly the connected cluster containing $s$.

The outer loop starts a new BFS only from an unvisited node. Hence, each true connected cluster is visited exactly once, and the returned $C_t$ equals the exact number of connected clusters. Therefore, $C_t=1$ if and only if $\mathcal{G}'_t$ is connected.

For complexity, each node is enqueued and dequeued at most once, and each edge is examined at most twice, once from each endpoint. Thus, the time complexity is $\mathcal{O}(|\mathcal{R}_t|+|\mathcal{E}'_t|)$, while the visited array and queue require $\mathcal{O}(|\mathcal{R}_t|)$ auxiliary space. 
\end{proof}

This result provides an exact graph-theoretic basis for connectivity decisions during online recovery, avoiding heuristic connectivity judgments.

\begin{proposition}[Protocol-Level Non-Transmission of Positional Fields]
Let the post-damage communication transcript be ${\mathcal{M}=(\mathcal{M}_{\mathrm{up}},\mathcal{M}_{\mathrm{down}})}$, where ${\mathcal{M}_{\mathrm{up}}=
\left(\mathcal{R}_{\mathrm{rep}},
\{\mathcal{D}_{u}^{\mathrm{rep}}\}_{u\in\mathcal{R}_{\mathrm{rep}}}\right)}$ represents the uplink ID-level reports, and ${\mathcal{M}_{\mathrm{down}}=
\left(\mathcal{D}_{\mathrm{inf}},\mathbf{\Theta}^{+}\right)}$ denotes the downlink inferred damage set and shared model 
parameters, with $\mathbf{\Theta}^{+}$ being the AdDiff-GT parameters shared among all surviving UAVs. If the communication protocol prohibits setting coordinates, trajectories, or per-UAV waypoint fields in uplink or downlink messages, $\mathcal{M}$ contains no explicit positional fields.
Thus, scenarios with identical ID-level reports and shared model parameters are indistinguishable in terms of explicit positional information in $\mathcal{M}$.
\end{proposition}

\begin{proof}
Define the positional-field projection operator $\Pi_{\mathrm{pos}}(\cdot)$ to extract all explicit positional and trajectory-related fields, including coordinates, velocities, target positions, and per-UAV waypoints.  According to Algorithms~\ref{alg:damage_inference} and~\ref{alg:airground}, the post-disaster uplink message contains only reporting UAV IDs and disconnected-neighbor IDs:
\begin{equation}
\mathcal{M}_{\mathrm{up}}
=\left(\mathcal{R}_{\mathrm{rep}},
\{\mathcal{D}_{u}^{\mathrm{rep}}\}_{u\in\mathcal{R}_{\mathrm{rep}}}\right),
\end{equation} 
all of whose elements belong to the finite node-label set $\{1,\ldots,N\}$. Hence, $\Pi_{\mathrm{pos}}(\mathcal{M}_{\mathrm{up}})=\emptyset$. 

The post-damage downlink contains the inferred failed-node set and shared AdDiff-GT parameters:
\begin{equation}
\mathcal{M}_{\mathrm{down}}
=\left(\mathcal{D}_{\mathrm{inf}},\mathbf{\Theta}^{+}\right),
\end{equation}
where $\mathcal{D}_{\mathrm{inf}}\subseteq\{1,\ldots,N\}$ is an ID set, while $\mathbf{\Theta}^{+}$ consists of shared model parameters and is not indexed by individual UAV positions, waypoints, or trajectories. Consequently, by the definition of protocol fields, we have
\begin{equation}
\Pi_{\mathrm{pos}}(\mathcal{M}_{\mathrm{down}})=\emptyset.
\end{equation}
It follows that
\begin{equation}
\Pi_{\mathrm{pos}}(\mathcal{M})
=\Pi_{\mathrm{pos}}(\mathcal{M}_{\mathrm{up}}\cup\mathcal{M}_{\mathrm{down}})
=\emptyset .
\end{equation}

Furthermore, consider two post-damage position configurations $\mathbf{P}$ and $\mathbf{P}^{\dagger}$ that generate identical event-triggered reports, inferred damage sets, and shared parameters. By the message-generation rules,
\begin{equation}
\mathcal{M}(\mathbf{P})=\mathcal{M}(\mathbf{P}^{\dagger}).
\end{equation}

Thus, the two configurations are indistinguishable at the protocol-message layer with respect to explicit positional fields. In particular, field decoding of the intercepted transcript does not directly reveal real-time UAV coordinates, target positions, or trajectories.
\end{proof}

This proposition establishes the protocol-level information boundary of the proposed air-ground architecture: passive monitoring of the LDTL alone does not expose explicit real-time coordinates or trajectory fields. This guarantee is limited to protocol-level communication leakage and does not preclude positional inference from external sensing or historical side information. 

\begin{proposition}[Discretization Consistency of Advection-Diffusion]
Within a single propagation layer, if the hybrid graph, attention coupling matrix $\mathbf{C}$, and local structure operator $\mathbf{V}$ are fixed, let the joint propagation operator be defined as  
$\mathcal{P}=\mathbf{C}+\beta\mathbf{V}$.
The continuous dynamics in Eq.~\eqref{eq:advection_diffusion} then form a locally linearized system  with the matrix-exponential expansion expressed as the infinite series of propagation bases $\{\mathcal{P}^{k}\}_{k=0}^{\infty}$. If $\mathcal{P}$ has a bounded spectral radius, this matrix-power series converges, and a finite-order truncation $\{\mathcal{P}^{k}\}_{k=0}^{K_{\text{order}}}$ provides an increasingly accurate approximation as $K_{\mathrm{order}}$ increases under the convergence condition. The multi-head propagation of AdDiff-GT in Eqs.~\eqref{eq:series_stack}--\eqref{eq:series_forward} constructs these multi-order propagation bases for each attention head and adaptively fuses them through a learnable linear mapping, thereby implementing a learnable finite-order polynomial approximation of the layer-wise advection–diffusion evolution.
\end{proposition}

\begin{proof}
The proof establishes the correspondence between continuous advection-diffusion dynamics and the discrete propagation mechanism of AdDiff-GT through four components: field sampling, graph discretization, matrix-exponential propagation, and finite-order approximation.

\subsection{Field Sampling on the Post-Damage Flight Manifold}
Consider the constrained 3D flight airspace $\Omega\subset\mathbb{R}^3$ as a physical manifold with boundaries. Before the damage event, the UAV positions $\{\mathbf{p}_i\}_{i\in\mathcal{V}_t}$ form discrete sampling points on $\Omega$. After UED-induced damage, failed nodes are removed, leaving the surviving sampling points $\{\mathbf{p}_i\}_{i\in\mathcal{R}_t}$. Let $\varphi(m,t)$ denote the continuous repair state field. The field observation at a surviving UAV is represented as
\begin{equation}
\mathbf{z}_i(t)=\psi(\varphi(\mathbf{p}_{i,t},t)),\quad i\in\mathcal{R}_t,
\end{equation}
where $\psi(\cdot)$ denotes the feature embedding based on geometric, cluster-level, and local topological information. 

For the $c$-th partitioned cluster with centroid $\mathbf{s}_{c,t}$, let $s_c$ denote the corresponding supernode. The supernode representation is sampled as
\begin{equation}
\mathbf{z}_{s_c}(t)=\psi(\varphi(\mathbf{s}_{c,t},t)).
\end{equation}
Therefore, the hybrid graph representation
\begin{equation}
\mathbf{Z}(t)=\left[\mathbf{z}_u(t)\right]_{u\in\mathcal{V}_t^h}
\end{equation}
can be interpreted as discrete observations of the continuous state field $\varphi(m,t)$ at surviving UAV locations and cluster-level sampling points. UAV nodes correspond to executable maneuvering entities, whereas supernodes represent aggregated cluster-level states and do not generate movement commands. The final displacement $\Delta\mathbf{p}_i$ is decoded from $\mathbf{z}_i(T)$ as a tangent-space vector in a discrete approximation of $T_{\mathbf{p}_{i,t}}\Omega$.

\subsection{Discretization of Manifold Operators on the Hybrid Graph}
On the continuous manifold, the gradient $\nabla \varphi$ characterizes local state variation, while its adjoint $\nabla ^{\ast}$ represents the corresponding divergence. For the hybrid graph $\mathcal{G}_t^h=(\mathcal{V}_t^h,\mathcal{E}_t^h)$, the discrete graph gradient is defined as 
\begin{equation}
(\nabla_{\mathcal{G}}\mathbf{Z})_{ij}=\mathbf{z}_j-\mathbf{z}_i,
\quad (i,j)\in\mathcal{E}_t^h,
\end{equation}
and the graph divergence is defined as 
\begin{equation}
(\nabla_{\mathcal{G}} ^{\ast} \mathbf{E})_i
=\sum_{j:(i,j)\in\mathcal{E}_t^h}\mathbf{e}_{ij}.
\end{equation}
These operators provide discrete counterparts of local state variation and net flux divergence. When graph edges represent local or structurally meaningful neighborhoods, they provide finite-difference approximations of the corresponding manifold operators. In AdDiff-GT, physical intra-cluster edges, virtual inter-cluster edges, and UAV-supernode edges jointly define the discrete neighborhood structure.

For the diffusion term, let $\mathbf{C}=[c_{ij}]$ denote the non-local diffusion coupling matrix generated by global attention, where $c_{ij}\ge0$ and  $\sum_j c_{ij}=1$. The discrete diffusion flux is defined as 
\begin{equation}
\mathbf{e}_{ij}^{\mathrm{dif}}=c_{ij}(\mathbf{z}_j-\mathbf{z}_i).
\end{equation}
Applying the graph divergence gives
\begin{equation}
\begin{aligned}
\left(\nabla_{\mathcal{G}} ^{\ast} \mathbf{E}^{\mathrm{dif}}\right)_i
&=\sum_j c_{ij}(\mathbf{z}_j-\mathbf{z}_i)\\
&=\sum_j c_{ij}\mathbf{z}_j-\left(\sum_j c_{ij}\right)\mathbf{z}_i\\
&=(\mathbf{C}\mathbf{Z})_i-\mathbf{z}_i .
\end{aligned}
\end{equation}
Hence, the non-local diffusion term can be expressed in matrix form as
\begin{equation}
\nabla_{\mathcal{G}} ^{\ast} \!\left(\mathbf{C}\odot\nabla_{\mathcal{G}}\mathbf{Z}\right)
=(\mathbf{C}-\mathbf{I})\mathbf{Z}.
\label{eq:appendix_diffusion_discrete}
\end{equation}

In AdDiff-GT, the attention coupling matrix $\mathbf{C}$ serves as the discrete non-local counterpart of the continuous diffusion coefficient. Unlike purely local graph propagation, the global attention mechanism can establish information coupling between UAVs and supernodes belonging to different clusters.

For the advection term, the continuous transport operator describes directional state propagation along a velocity field. In the hybrid graph, the normalized adjacency
\begin{equation}
\mathbf{V}=\mathbf{D_t}^{-1/2}\mathbf{A}_t^{h}\mathbf{D_t}^{-1/2},
\end{equation}
serves as a graph-level surrogate for directional transport along the observable topology. Let
$\mathbf{e}_{ij}^{\mathrm{adv}}=v_{ij}\mathbf{z}_j$
denote the advection edge flux, where $v_{ij}$ is an element of $\mathbf{V}$. Then, 
\begin{equation}
\left(\nabla_{\mathcal{G}} ^{\ast} \mathbf{E}^{\mathrm{adv}}\right)_i
=\sum_j v_{ij}\mathbf{z}_j
=(\mathbf{V}\mathbf{Z})_i .
\label{eq:appendix_advection_discrete}
\end{equation}

Therefore, the advection operator propagates information along physical and virtual graph structures, preserving local structural relationships while providing directional coupling across inter-cluster links.

Consider the continuous advection-diffusion equation
\begin{equation}
\frac{\partial \varphi}{\partial t}
=\nabla ^{\ast} (S\odot\nabla \varphi)
+\beta\nabla ^{\ast} (V\varphi),
\end{equation}
where $S$ denotes the continuous diffusion coefficient and $V$ represents the directional transport field. Under graph discretization, $\mathbf{C}$ serves as the discrete counterpart of $S$, while the normalized hybrid adjacency provides the discrete transport operator.  Combining Eqs.
\eqref{eq:appendix_diffusion_discrete} and   \eqref{eq:appendix_advection_discrete} yields
\begin{equation}\label{graph-level advection-diffusion dynamics}
\frac{\partial\mathbf{Z}(t)}{\partial t}
=
(\mathbf{C}-\mathbf{I})\mathbf{Z}(t)
+\beta\mathbf{V}\mathbf{Z}(t)
=
\left(\mathbf{C}+\beta\mathbf{V}-\mathbf{I}\right)\mathbf{Z}(t),
\end{equation}
which gives the graph-level advection-diffusion dynamics in Eq.~\eqref{eq:advection_diffusion}
used in AdDiff-GT. 

\subsection{Matrix-Exponential Graph Propagation}
Within a single propagation layer, assuming that the hybrid graph, 
$\mathbf{C}$ and $\mathbf{V}$ remain fixed, and define 
\begin{equation}
\mathcal{P}=\mathbf{C}+\beta\mathbf{V},\quad\mathbf{L}= \mathbf{I}-\mathcal{P}=\mathbf{I}-\mathbf{C}-\beta\mathbf{V}.
\end{equation}
Then Eq.~\eqref{graph-level advection-diffusion dynamics} becomes 
\begin{equation}
\frac{\partial\mathbf{Z}(t)}{\partial t}
=-\mathbf{L}\mathbf{Z}(t).
\end{equation}
Integrating over time yields:
\begin{equation}
\mathbf{Z}(t)
=\mathbf{Z}(0)+\int_0^t-\mathbf{L}\mathbf{Z}(\tau)d\tau .
\end{equation}

The solution of this linear system is
\begin{equation}
\mathbf{Z}(T_{\mathrm{prop}})
=\mathrm{e}^{-\mathbf{L}T_{\mathrm{prop}}}\mathbf{Z}(0)
=\mathrm{e}^{-(\mathbf{I}-\mathbf{C}-\beta\mathbf{V})T_{\mathrm{prop}}}\mathbf{Z}(0).
\label{eq:appendix_matrix_solution}
\end{equation}

This result shows that the propagation mechanism of AdDiff-GT can be interpreted as an approximation to the state evolution of the graph-level advection-diffusion system rather than as a simple single-step neighbor-aggregation operation.

\subsection{Finite-Order Multi-Head Approximation}
Directly computing the matrix exponential in Eq.~\eqref{eq:appendix_matrix_solution} typically requires eigendecomposition or other computationally intensive matrix operations, which can be computationally demanding for online recovery in large-scale UAV swarms. To improve scalability, AdDiff-GT approximates the matrix-exponential propagation using a finite-order polynomial expansion of the propagation operator.

Since $\mathbf{L}=\mathbf{I}-\mathcal{P}$ and $\mathbf{I}$ commutes with $\mathcal{P}$, Eq.~\eqref{eq:appendix_matrix_solution} can be written as 
\begin{equation}
\mathrm{e}^{-\mathbf{L}T_{\mathrm{prop}}}
=\mathrm{e}^{-(\mathbf{I}-\mathcal{P})T_{\mathrm{prop}}}
=\mathrm{e}^{-T_{\mathrm{prop}}}\sum_{k=0}^{\infty}\frac{T_{\mathrm{prop}}^k}{k!}\mathcal{P}^k .
\label{eq:appendix_taylor_powers}
\end{equation}

For $\|\mathcal{P}\|_2\le B_P$, the truncation remainder after the $K$-th order satisfies  
\begin{equation}
\begin{split}
&\left\|
\mathrm{e}^{-\mathbf{L}T_{\mathrm{prop}}}
-\mathrm{e}^{-T_{\mathrm{prop}}}\sum_{k=0}^{K}\frac{T_{\mathrm{prop}}^k}{k!}\mathcal{P}^k
\right\|_2 \\
&\qquad\le
\mathrm{e}^{-T_{\mathrm{prop}}}\sum_{k=K+1}^{\infty}\frac{(T_{\mathrm{prop}}B_P)^k}{k!}
\to 0.
\end{split}
\end{equation}
Thus, the matrix-exponential propagation can be approximated using the finite propagation bases $\{\mathcal{P}^k\mathbf{Z}(0)\}_{k=0}^{K_{\text{order}}}$. 
These bases encode propagation information at different graph orders and avoid explicitly computing the matrix exponential.




AdDiff-GT further generalizes the fixed coefficients of the matrix-exponential expansion through learnable multi-head propagation. For the $h$-th attention head, define
\begin{equation}
{\mathcal{P}_h=\mathbf{C}_h+\beta\mathbf{V}.}
\end{equation}
The propagation bases are recursively generated as
\begin{equation}
\mathbf{Z}_{h}^{(0)}=\mathbf{Z}^{(\ell)},\quad
\mathbf{Z}_{h}^{(k)}=\mathcal{P}_h\mathbf{Z}_{h}^{(k-1)},\quad k=1,\ldots,K_{\mathrm{order}}.
\end{equation}
The representations from all propagation orders and attention heads are then concatenated and fused through a learnable linear mapping:
\begin{equation}
\widetilde{\mathbf{Z}}^{(\ell+1)}
=\frac{1}{H}\mathbf{W}_{\mathrm{o}}
\operatorname{Concat}_{h=1}^{H}
\left(\left[\mathbf{Z}_{h}^{(0)},\mathbf{Z}_{h}^{(1)},\ldots,
\mathbf{Z}_{h}^{(K_{\mathrm{order}})}\right]\right).
\end{equation}

Unlike the fixed Taylor coefficients in Eq.~\eqref{eq:appendix_taylor_powers}, the output mapping $\mathbf{W}_{\mathrm{o}}$ learns the contributions of different propagation orders and attention heads. Hence, AdDiff-GT implements a learnable finite-order polynomial approximation of the matrix-exponential advection-diffusion propagation. The factor $1/H$ preserves the averaging convention used in the multi-head aggregation. 

Consequently, Eqs.~\eqref{eq:series_stack}--\eqref{eq:learnable_series} can be interpreted as a finite multi-order approximation to the matrix-exponential solution, with the propagation coefficients learned from recovery data. 
\end{proof}

This result establishes the correspondence between the continuous advection-diffusion dynamics and the discrete propagation mechanism of AdDiff-GT when the attention coupling is fixed within a propagation layer. When $\mathbf{C}$ varies with $\mathbf{Z}(t)$, the system becomes dependent on the current state, 
and the closed-form solution in Eq.~\eqref{eq:appendix_matrix_solution} for a fixed operator no longer applies directly. In this case, the proposed propagation is interpreted as a layer-wise local linearization of the state-dependent advection-diffusion dynamics.

\begin{proposition}[SCD as the Minimum-Bottleneck Connectivity Radius]
Let $\mathcal{G}'_{\omega}=(\mathcal{R}_t,\mathcal{E}_{\omega},\mathbf{X}'_t,\mathbf{W})$ be a complete weighted graph of surviving UAVs, where $\mathcal{E}_{\omega}=\{(i,j)\mid i,j\in\mathcal{R}_t,\ i<j\}$ and the edge weight is $d_{ij}=\|\mathbf{p}_{i,t}-\mathbf{p}_{j,t}\|_2$.  Let $\mathcal{T}$  denote the set of all spanning trees of $\mathcal{G}'_{\omega}$. Then, the SCD returned by Algorithm~\ref{alg:scd} satisfies   
\begin{equation}
d_{\mathrm{SCD}}
=\min_{T\in\mathcal{T}}\max_{(i,j)\in T}d_{ij}.
\end{equation}
That is, $d_{\mathrm{SCD}}$ is the minimum communication radius for which the corresponding distance-threshold graph becomes connected. Equivalently, it is the minimum bottleneck value over all spanning trees, and hence equals the maximum edge weight of a minimum spanning tree.
\end{proposition}

\begin{proof}
Let $\mathcal{G}'_{\omega}=(\mathcal{R}_t,\mathcal{E}_{\omega},\mathbf{X}'_t,\mathbf{W})$  be the complete weighted graph defined above. For any communication radius $d$, define the corresponding distance threshold graph as $\mathcal{G}(d)=(\mathcal{R}_t,\mathcal{E}(d))$, where $\mathcal{E}(d)=\{(i,j)\in \mathcal{E}_{\omega}, d_{ij} \le d\}$.
If $d_1\le d_2$, then $\mathcal{E}(d_1)\subseteq\mathcal{E}(d_2)$, so connectivity is monotone with respect to the threshold $d$. 

First, consider the relationship between threshold graph connectivity and spanning tree bottleneck value. 
If $\mathcal{G}(d)$ is connected, it contains a spanning tree $T_d$. Hence,
\begin{equation}
\min_{T\in\mathcal{T}}\max_{(i,j)\in T}d_{ij}
\le
\max_{(i,j)\in T_d} d_{ij}
\le d.
\end{equation}

Conversely, if there exists a spanning tree $T$ satisfying $\max_{(i,j)\in T}d_{i,j}\le d$, then all edges of $T$ belong to $\mathcal{G}(d)$. Thus, $\mathcal{G}(d)$ contains at least one spanning tree and must be connected. Consequently, 
\begin{equation}
d_{\mathrm{SCD}}=\min\{d\mid \mathcal{G}(d)\ \text{is connected}\}
=\min_{T\in\mathcal{T}}\max_{(i,j)\in T}d_{ij}.
\end{equation}

Next, let $T_{\mathrm{MST}}$ be any minimum spanning tree and define $d^\star = \max_{(i,j)\in T_{\mathrm{MST}}} d_{ij}$. By Kruskal's algorithm, edges are processed in nondecreasing order of weight. For any $d<d^\star$, the edges with weight at most $d$ cannot contain a spanning tree. Otherwise, a spanning tree with maximum edge weight smaller than $d^\star$ would exist, contradicting the minimum spanning tree property. Hence, $\mathcal{G}(d)$ is disconnected for $d<d^\star$. At $d=d^\star$, all edges of $T_{\mathrm{MST}}$ are contained in $\mathcal{G}(d^\star)$, so  $\mathcal{G}(d^\star)$ is connected. Thus, 
\begin{equation}
d_{\mathrm{SCD}} = d^\star
=\max_{(i,j)\in T_{\mathrm{MST}}}d_{ij}
=\min_{T\in\mathcal{T}}\max_{(i,j)\in T}d_{ij}.
\end{equation}

Algorithm~\ref{alg:scd} sorts the candidate pairwise distances in ascending order and activates the corresponding edges sequentially. It returns the first distance at which BFS identifies a single connected cluster. Since the connectivity threshold can occur only at an inter-UAV distance, the returned value is exactly $d_{\mathrm{SCD}}$. 
\end{proof}

Therefore, $T_m=d_{\mathrm{SCD}}-R_c$ quantifies the additional spatial scale relative to the physical communication radius $R_c$ that must be compensated for to restore connectivity. A larger $d_{\mathrm{SCD}}$ indicates a larger connectivity gap in the current partitioned topology. Thus, $T_m$ provides a topology-aware reference for configuring the displacement search range of the online recovery algorithm.

\begin{proposition}[Exponential Perturbation Bound for Purely Local Diffusion]
If graph propagation reduces to purely local diffusion, where the propagation coupling is provided solely by the normalized adjacency operator without global attention compensation, then, under the same bounded spectral conditions, the perturbation induced deviation in the output node representations satisfies
\begin{equation}
\|\Delta\mathbf{Z}_{\mathrm{out}}\|_2
=\mathcal{O}\!\left(
\|\Delta\mathbf{V}\|_2
\mathrm{e}^{T_{\mathrm{prop}}\|\Delta\mathbf{V}\|_2}\right)
,
\end{equation}
where $\mathcal{O}(\cdot)$ characterizes the asymptotic order of the representation perturbation with respect to $\|\Delta\mathbf{V}\|_2$.
\end{proposition}

\begin{proof}
Consider the pure local diffusion propagation without global attention compensation:
\begin{equation}
\frac{\partial\mathbf{Z}}{\partial t}
=-(\mathbf{I}-\mathbf{V})\mathbf{Z},
\end{equation}
where $\mathbf{V}$ is the normalized adjacency operator. For a fixed initial representation $\mathbf{Z}(0)$, the propagated representation at time $T_{\mathrm{prop}}$ is 
\begin{equation}
\mathbf{Z}(T_{\mathrm{prop}};\mathbf{V})
=\mathrm{e}^{-(\mathbf{I}-\mathbf{V})T_{\mathrm{prop}}}\mathbf{Z}(0).
\end{equation}

Under a topological perturbation, $\mathbf{V}'=\mathbf{V}+\Delta\mathbf{V}$, and hence
\begin{equation}
\Delta\mathbf{L}
=-(\mathbf{V}'-\mathbf{V})
=-\Delta\mathbf{V}.
\end{equation}
Applying the matrix exponential perturbation bound in Eq.~\eqref{eq:appendix_exp_bound} with $\mathbf{X}=-(\mathbf{I}-\mathbf{V})T_{\mathrm{prop}}$ and $\mathbf{E}=T_{\mathrm{prop}}\Delta\mathbf{V}$, we obtain
\begin{equation}
\begin{aligned}
&\left\|
\mathrm{e}^{-(\mathbf{I}-\mathbf{V}')T_{\mathrm{prop}}}
-\mathrm{e}^{-(\mathbf{I}-\mathbf{V})T_{\mathrm{prop}}}
\right\|_2  \\
&\quad \le
T_{\mathrm{prop}}\|\Delta\mathbf{V}\|_2 M^2
\mathrm{e}^{\omega T_{\mathrm{prop}}(\|\mathbf{I}-\mathbf{V}\|_2
+\|\Delta\mathbf{V}\|_2)} .
\end{aligned}
\end{equation}

Within the bounded spectral domain, $\mathrm{e}^{\omega T_{\mathrm{prop}}\|\mathbf{I}-\mathbf{V}\|_2}$ can be absorbed into a constant term. Therefore, 
\begin{equation}
\left\|
\mathrm{e}^{-(\mathbf{I}-\mathbf{V}')T_{\mathrm{prop}}}
-\mathrm{e}^{-(\mathbf{I}-\mathbf{V})T_{\mathrm{prop}}}
\right\|_2
=O\!\left(
\|\Delta\mathbf{V}\|_2
\mathrm{e}^{T_{\mathrm{prop}}\|\Delta\mathbf{V}\|_2}
\right).
\end{equation}

If the encoder, inter-layer mappings, and decoder are Lipschitz continuous within the bounded domain, these mappings introduce only finite multiplicative constants. Consequently,
\begin{equation}
\|\Delta\mathbf{Z}_{\mathrm{out}}\|_2
=O\!\left(
\|\Delta\mathbf{V}\|_2
\mathrm{e}^{T_{\mathrm{prop}}\|\Delta\mathbf{V}\|_2}
\right).
\end{equation}

\end{proof}

This result shows that purely local diffusion exhibits exponential sensitivity to topological perturbations with the propagation depth $T_{\mathrm{prop}}$, whereas the global attention compensation in AdDiff-GT reduces the perturbation dependence to a polynomial form under bounded input and bounded compensation conditions.

\section{Supplementary Experimental Results}
\label{app:exp_details}

\begin{figure*}[t]
\centering
\begin{minipage}[b]{0.28\textwidth}
\centering
\includegraphics[width=\linewidth]{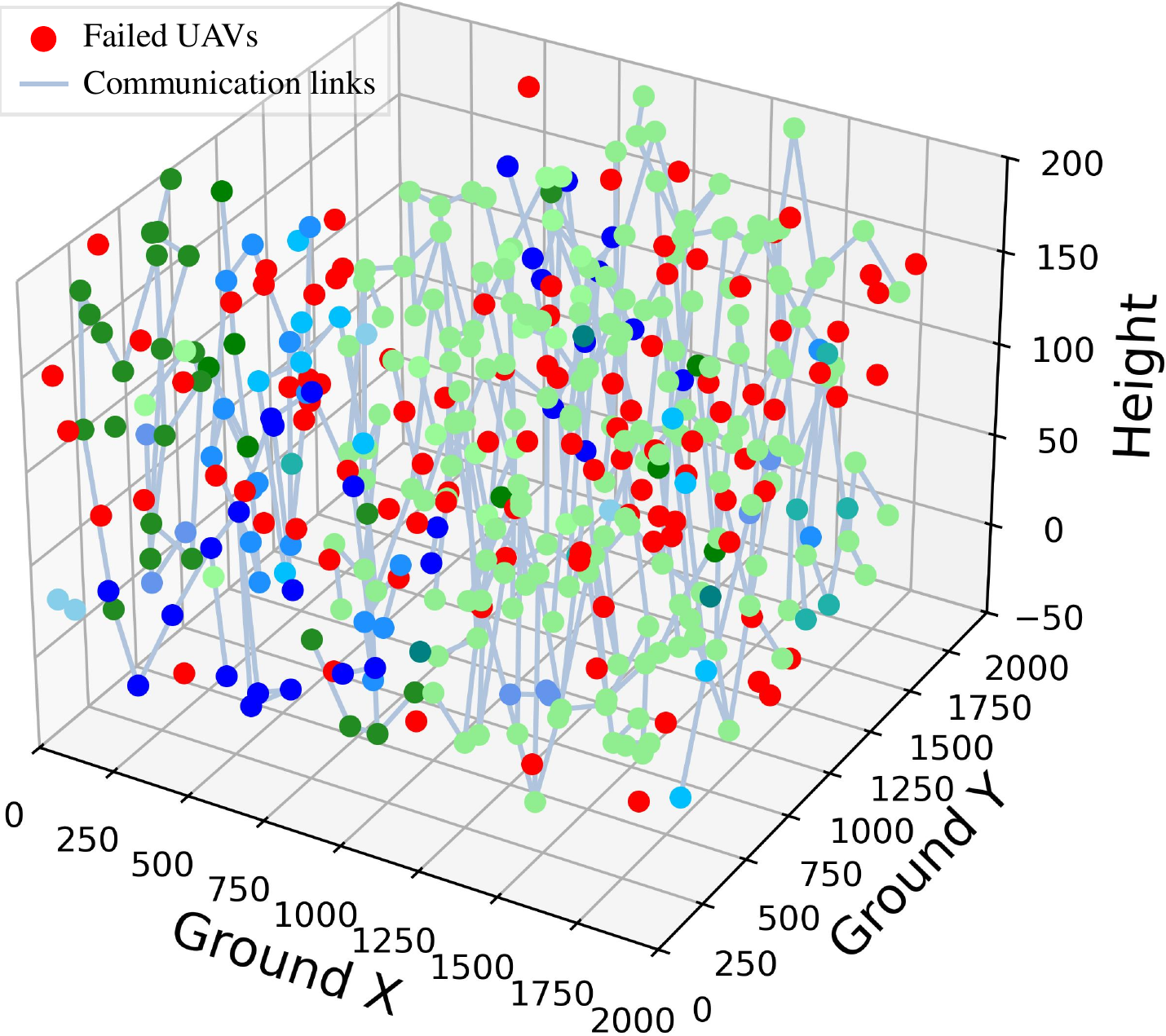}\\
\small (a) Initial residual topologies
\end{minipage}
\hfill
\begin{minipage}[b]{0.28\textwidth}
\centering
\includegraphics[width=\linewidth]{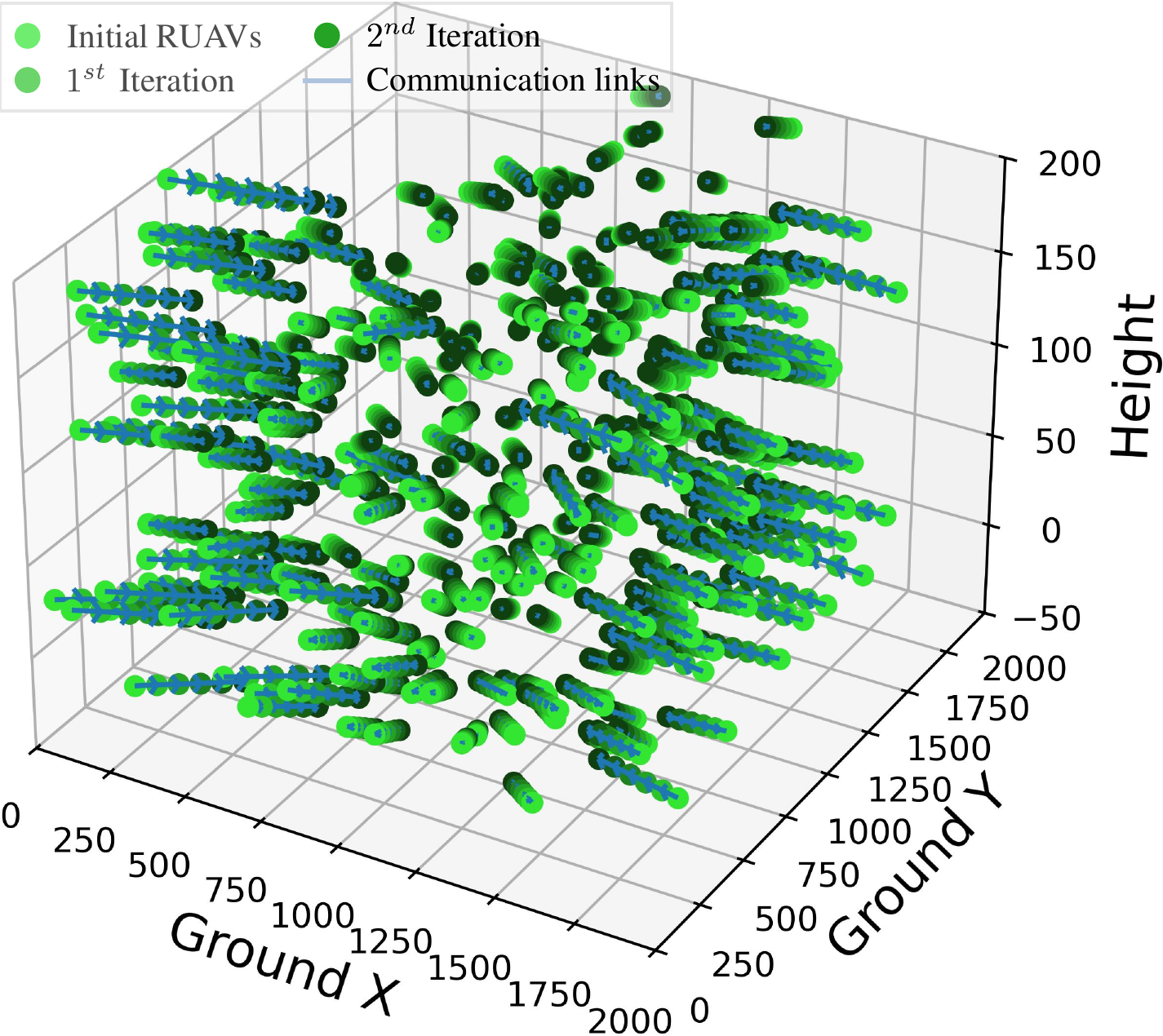}\\
\small (b) 3D recovery trajectories
\end{minipage}
\hfill
\begin{minipage}[b]{0.28\textwidth}
\centering
\includegraphics[width=\linewidth]{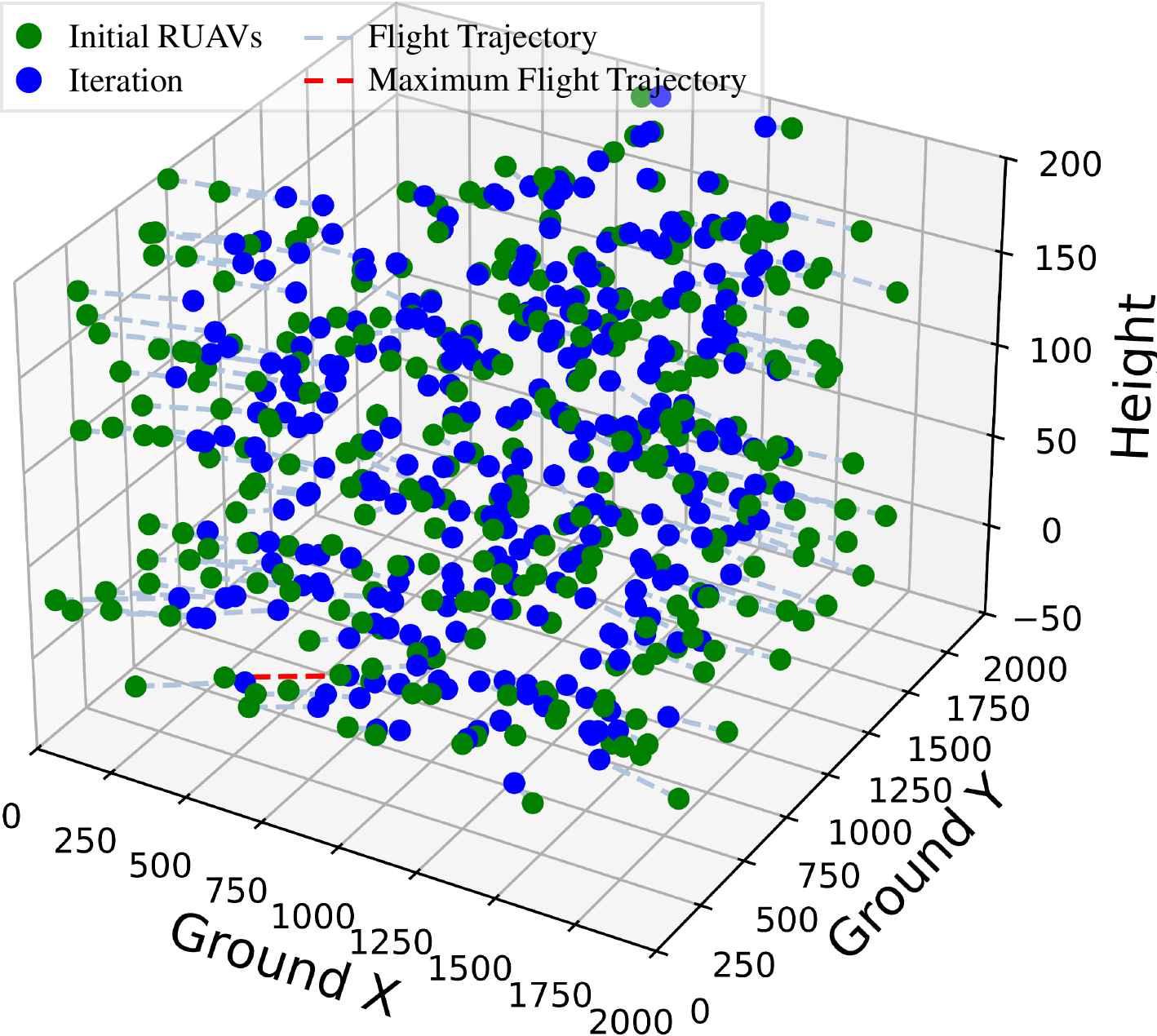}\\
\small (c) Final spatial distributions
\end{minipage}
\caption{Evolution of self-healing performance with 100 damaged nodes.}
\label{fig:Evolutionary_visualization_100}
\end{figure*}

\begin{figure*}[t]
\centering
\begin{minipage}[b]{0.28\textwidth}
\centering
\includegraphics[width=\linewidth]{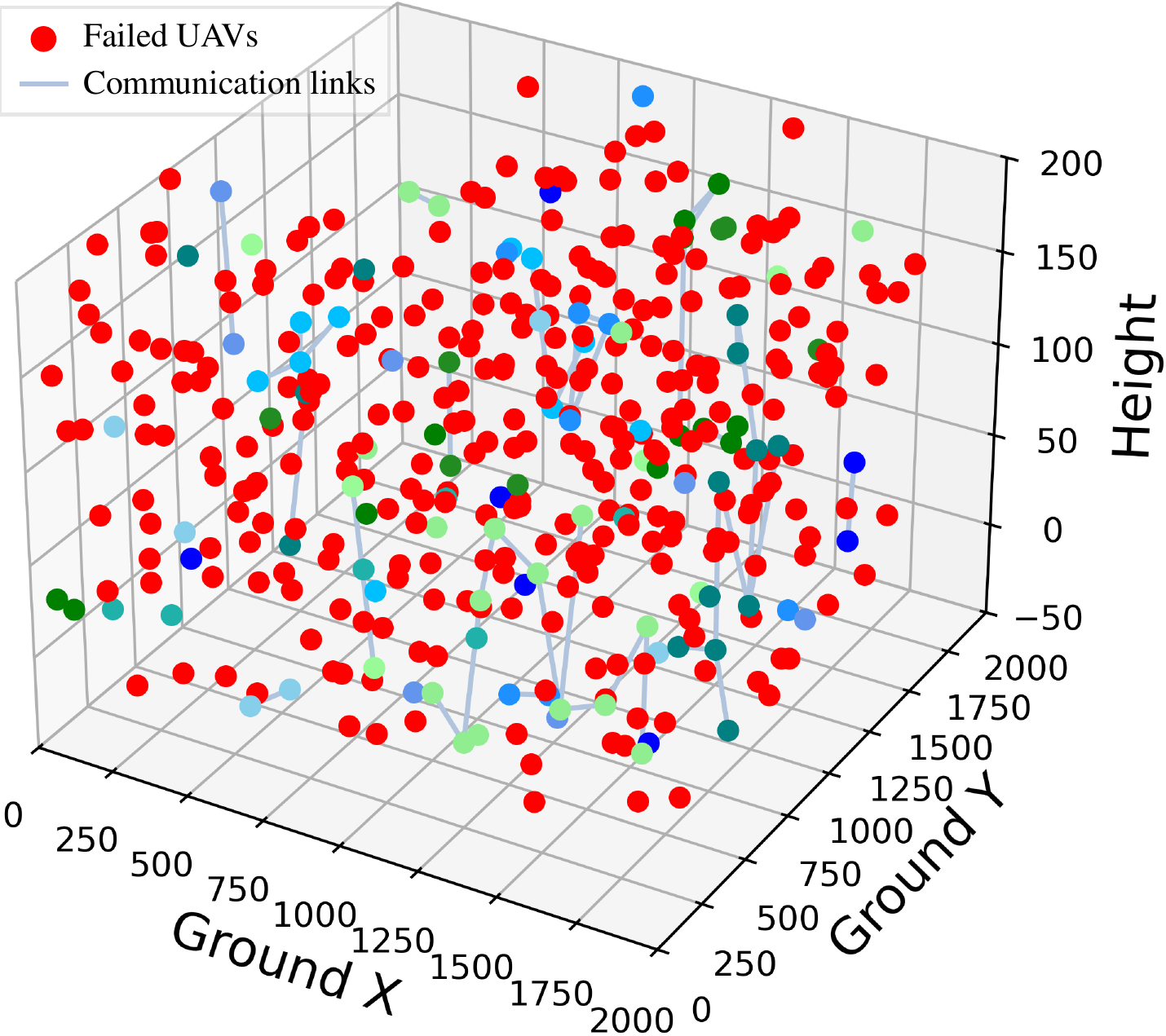}\\
\small (a) Initial residual topologies
\end{minipage}
\hfill
\begin{minipage}[b]{0.28\textwidth}
\centering
\includegraphics[width=\linewidth]{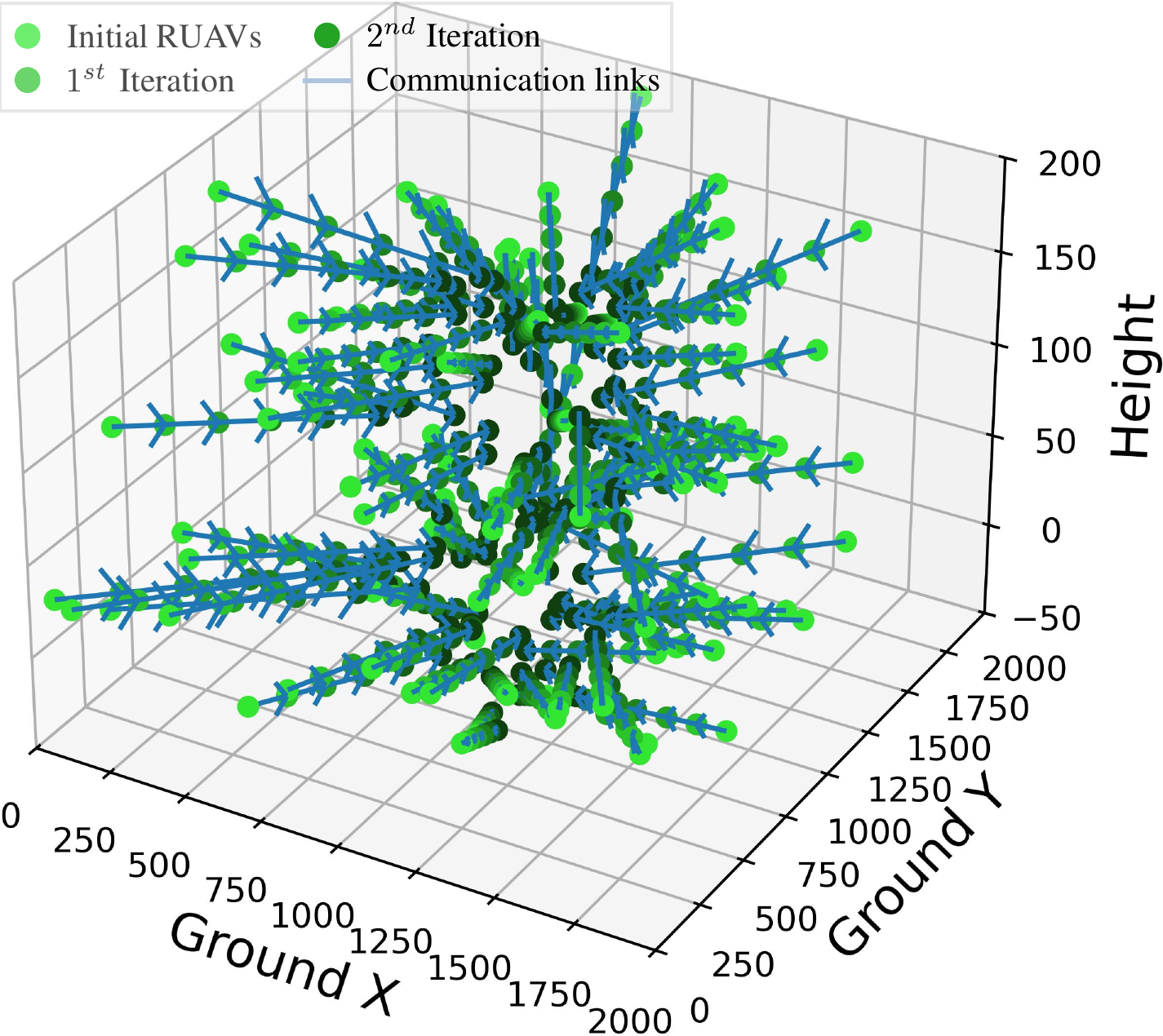}\\
\small (b) 3D recovery trajectories
\end{minipage}
\hfill
\begin{minipage}[b]{0.28\textwidth}
\centering
\includegraphics[width=\linewidth]{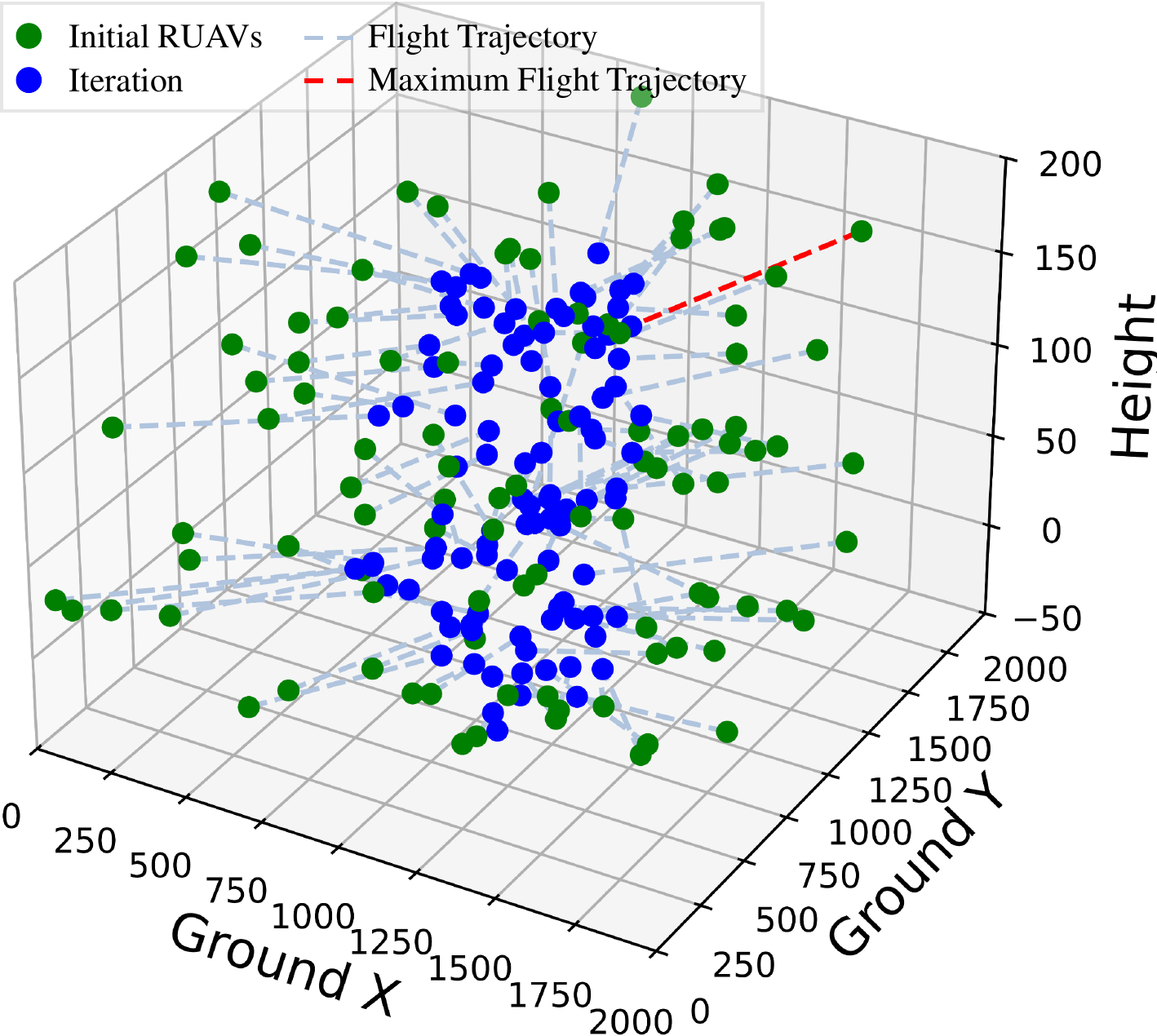}\\
\small (c) Final spatial distributions
\end{minipage}
\caption{Evolution of self-healing performance with 300 damaged nodes.}
\label{fig:Evolutionary_visualization_300}
\end{figure*}

Figs.~\ref{fig:Evolutionary_visualization_100} and \ref{fig:Evolutionary_visualization_300} illustrate representative self-healing trajectories under 100 and 300 damaged nodes
with an initial network size of $N=$ 
400.

\makeatletter
\setlength{\@fptop}{0pt}
\setlength{\@fpsep}{0.75\baselineskip}
\setlength{\@fpbot}{0pt plus 1fil}
\setlength{\@dblfptop}{0pt}
\setlength{\@dblfpsep}{0.75\baselineskip}
\setlength{\@dblfpbot}{0pt plus 1fil}
\makeatother

\vfill

\end{document}